\documentclass[10pt,a4paper]{article}
\usepackage{latexsym,amssymb,amsmath,amsthm,amsfonts,enumerate,verbatim,xspace,
exscale}

\input xy
\xyoption{all} \CompileMatrices \UseComputerModernTips

\usepackage{graphicx,tabularx}
\usepackage{amsmath,amsbsy,amsfonts,amssymb}
\usepackage{color}
\usepackage{amsthm}
\def\parentheses#1{\!\left(#1\right)}
\def\brackets#1{\!\left[#1\right]}

\def\R{\mathbb{R}}

\title{Reduction of nonholonomic systems in two stages and Hamiltonization}

\author{{\sc{Paula Balseiro}\thanks{
         Universidade Federal Fluminense, Instituto de Matem\'atica, Rua Mario Santos Braga S/N, 24020-140, Niteroi, Rio de Janeiro, Brazil. \newline{\texttt{E-mail: pbalseiro@vm.uff.br}}}} \ \  and \
{\sc{Oscar E. Fernandez}\thanks{
        Department of Mathematics, Wellesley College, 106 Central Street, Wellesley, MA 02481, USA. \newline{\texttt{E-mail:
ofernand@wellesley.edu}}}} }

\theoremstyle{plain}
\newtheorem{theorem}{Theorem}[section]
\newtheorem{lemma}[theorem]{Lemma}
\newtheorem{proposition}[theorem]{Proposition}
\newtheorem{corollary}[theorem]{Corollary}
\newtheorem*{theorem*}{Theorem}
\newtheorem{remarkth}[theorem]{Remark}

\theoremstyle{definition}

\newtheorem{notation}[theorem]{Notation}

\newenvironment{remark}{\begin{remarkth}\upshape}{\hfill$\diamond$\end{remarkth}}

\def\W{\mathcal{W}}
\def\M{\mathcal{M}}

\def\V{\mathcal{V}}
\def\S{\mathcal{S}}
\def\C{\mathcal{C}}
\def\Ham{\mathcal{H}}

\def\R{\mathbb{R}}

\def\L{\mathbb{F}L}
\def\red{{\mbox{\tiny{red}}}}
\def\nh{{\mbox{\tiny{nh}}}}

\def\B{{\mbox{\tiny{$B$}}}}
\def\subW{{\mbox{\tiny{$\W$}}}}

\def\subS{{\mbox{\tiny{$\S$}}}}
\def\subM{{\mbox{\tiny{$\M$}}}}

\def\vecOm{\boldsymbol{\Omega}}

\def\a{\alpha}
\def\b{\beta}
\def\cc{\textup{\bf c}}
\def\vv{\textup{\bf v}}

\def\vecom{\boldsymbol{\omega}}
\newcommand{\SO}{\mbox{$\textup{SO}$}}

\def\vecL{\boldsymbol{\lambda}}
\def\vecR{\boldsymbol{\rho}}
\def\vecgamma{\boldsymbol{\gamma}}
\def\vecalpha{\boldsymbol{\alpha}}
\def\vecbeta{\boldsymbol{\beta}}

\begin{document}
\maketitle

\begin{abstract}
In this paper we study the reduction of a nonholonomic system by a group of symmetries in two steps. Using the so-called {\it vertical-symmetry} condition, we first perform a {\it compression} of the nonholonomic system leading to an almost symplectic manifold. Second, we perform an {\it (almost) symplectic reduction}, relying on the existence of a momentum map. In this case, we verify that the resulting (almost) symplectic reduced spaces are the leaves of the characteristic foliation of the reduced nonholonomic bracket.
On each leaf we study the (Lagrangian) equations of motion, obtaining a nonholonomic version of the Lagrange-Routh equations, and we analyze the existence of a conformal factor for the reduced bracket in terms of the existence of  conformal factors for the almost symplectic leaves.
We also relate the conditions for the existence of a momentum map for the compressed system with gauge transformations by 2-forms. The results are illustrated in three different examples.
\end{abstract}

\tableofcontents

\section{Introduction} \label{S:Intro}

In mechanics, a central role is played by the so-called {\em nonholonomic systems}, i.e., mechanical systems subject to non-integrable
constraints on their velocities (see e.g. \cite{Bloch:Book,BKMM,BS93,Koiller1992}). From a mathematical standpoint, nonholonomic contraints are encoded in non-integrable
geometrical structures, e.g., 2-forms which fail to be closed, or bivector fields which fail to be Poisson \cite{BS93,IdLMM}. It has been observed that, in the presence of symmetries, it is often possible to control, or even eliminate, the non-integrable nature of a non-holonomic system by means of {\em reduction}. In other words, in many situations, reduced systems are closer to being Hamiltonian than the original ones. When the reduced system can  actually be brought to Hamiltonian form (possibly after a time reparametrization), one  speaks of the {\em Hamiltonization} of the nonholonomic system, see e.g.  \cite{BorMa,Chap,EhlersKoiller,JovaChap} (other references will be given throughout the paper).

This paper is concerned with the integrability properties of reduced nonholonomic brackets, and in
particular the existence of leaves integrating their characteristic distributions, as well as geometric ways to describe them.
The main ingredient in our approach to these issues is the observation that, under suitable assumptions on the symmetries,
the procedure of reduction by symmetries can be split in two steps, the first being a partial reduction by a smaller group of symmetries (a {\em Chaplygin reduction}, or {\em compression} \cite{MovingFrames,EhlersKoiller}),
and the second a Marsden-Weinstein type reduction with respect to the remaining symmetries. This second step depends on
the existence of a momentum map, and we explain how this condition provides information about the integrability of reduced brackets. We also give a description of the reduced nonholonomic dynamics on reduced spaces corresponding to different momentum levels. This leads to a unified viewpoint for various previous works on the subject \cite{paula,Fernandez,Hoch}.

The paper is organized as follows. In Section~\ref{S:OneStep} we recall the setup for nonholonomic systems. 
In the Hamiltonian framework one has a submanifold $\M$ of $T^*Q$ equipped with an almost Poisson bracket
$\{\cdot , \cdot \}_\nh$---called the {\em nonholonomic bracket} \cite{CdLMD,Marle,SchaftMaschke1994}---and a Hamiltonian function $\Ham_\subM$. The nonholonomic dynamics on $\M$ is given by the integral curves of the {\em nonholonomic vector field}, defined by $X_\nh = \{\cdot , \Ham_\subM\}_\nh.$ We note that nonholonomic brackets have an associated characteristic distribution that is non-integrable.
 If $G$ is a group of symmetries, that is, $G$ acts on $Q$ preserving all additional structures,
 then the nonholonomic system can be reduced to the quotient manifold $\M/G$.  The reduced dynamics is then described by a reduced Hamiltonian function $\Ham_\red$ and a reduced almost Poisson bracket $\{\cdot, \cdot\}_\red$ on $\M/G$.
We will be particularly interested in reduced brackets $\{\cdot, \cdot\}_\red$ that are {\it twisted Poisson} \cite{SeveraWeinstein}; such brackets may not satisfy the Jacobi identity, but admit an almost symplectic foliation integrating the characteristic distribution.

Assuming an additional condition on the symmetries (referred to as {\it vertical-symmetry condition} in \cite{paula}),
we show in Section~\ref{S:2steps} how the reduction process from $\M$ to $\M / G$ can be split into two steps.  First, we observe that there is an action on $\M$ of a normal subgroup $G_\subW$ of $G$ with respect to which the system is Chaplygin. Carrying out Chaplygin reduction is the first step in our reduction procedure, and it leads to an almost symplectic form $\bar\Omega$ on $\M/G_\subW$.
The second step of the process consists of reducing this almost symplectic manifold (\`a la Marsden-Weinstein \cite{MW}) by the remaining symmetries $G/G_\subW$.
At this point, we need conditions to guarantee the existence of a momentum map for $\bar{\Omega}$, which is closely related to the existence of conserved quantities for the $(G/G_\subW)$-symmetry. Whenever this second reduction can be performed, it leads to an almost symplectic foliation integrating the characteristic distribution of $\{\cdot, \cdot \}_\red$, see Theorem~\ref{T:Equivalence}.

Since the 2-form $\bar{\Omega}$ at the compressed level may not admit a momentum map, in Section \ref{S:Gauge} we explore the possibility of circumventing this problem by introducing a 2-form $\bar{B}$ in such a way that $\bar\Omega+\bar{B}$ still describes the compressed dynamics {\em and} admits a momentum map (c.f. \cite{Hoch,EhlersKoiller}). This procedure is related with the {\it dynamical gauge transformations by 2-forms} of the nonholonomic bracket $\{\cdot, \cdot\}_\nh$ studied in \cite{paula,PL2011}. By considering reduction in two steps, we obtain a description of the almost symplectic leaves of reduced brackets arising from these more general gauge-transformed brackets.

The goal of performing the reduction in two steps is to understand the possible foliations associated to the reduced nonholonomic bracket.  This viewpoint also helps to see the effect of applying a {\it gauge transformation}. It is worth mentioning that other approaches to the problem of nonholonomic reduction are possible---see \cite{FaSa,KoKo}---but we focus here on the formalism developed by Marsden and Koiller.

An important point in having a description of the almost symplectic foliation associated to reduced nonholonomic brackets is that we can study the dynamics on each leaf independently. In this paper we consider two key aspects of the dynamics on a leaf: the local version of the equations of motion and the existence of a conformal factor for the almost symplectic form.
 In Section~\ref{S:RouthEquations} we consider the classical Lagrange-Routh equations \cite{Marsden2000,Bavo,CM1} and explain how they express the equations of motion with respect to a canonical Poisson bracket on the reduced space.
  Then, we show that the equations of motion on each almost symplectic leaf of the reduced nonholonomic bracket are given
by a suitable modification of the Lagrange-Routh equations. Such {\em nonholonomic Lagrange-Routh equations} rely on the conservation of the momentum map studied in Sec.~\ref{S:2steps}. The modifications relating the classical and nonholonomic Lagrange-Routh equations can be naturally understood in terms of gauge transformations.

In Section \ref{sec:conf} we study conformal factors for reduced nonholonomic brackets by considering conformal factors on
each almost symplectic leaf. Recall that a {\em conformal factor} for $\{\cdot, \cdot\}_\red$ is a real-valued function $g$ with the property that $g \{\cdot, \cdot\}_\red$ is a Poisson bracket; the existence of a conformal factor guarantees that the reduced system can be transformed into a Hamiltonian system via a reparameterization of time \cite{FedorovJova}.  Using the almost symplectic foliation associated to $\{\cdot, \cdot\}_\red$, the problem of Hamiltonization is stated as the problem of finding a conformal factor for each  almost symplectic leaf. If the leaves are diffeomorphic to a cotangent bundle we can generalize the so-called Stanchenko and Chaplygin approaches \cite{Stanchenko,Chap,Oscar} to each leaf by including  gauge transformations.

Our results are illustrated in Section \ref{sec:examples} with three examples: the snakeboard, the Chaplygin ball and a mathematical example inspired by \cite{CM1}. For the snakeboard one obtains a reduced nonholonomic bracket $\{\cdot, \cdot \}_\red$ that is twisted Poisson, while for the Chaplygin ball one needs to consider the addition of a 2-form $\bar{B}$ at the compressed level in order to obtain a similar result.
In both examples we derive the equations of motion and the differential equations for the conformal factors on each almost symplectic leaf. In the case of the Chaplygin ball, this leads to a simple way to obtain the conformal factor previously found in \cite{BorMa}. The third example illustrates Section 5 for the case when the remaining Lie group $G/G_\subW$ is not abelian.


%

\noindent {\bf Acknowledgements:}  P.B. would like to thank CNPq (Brazil) for financial support and also David Iglesias-Ponte and Tom Mestdag for useful conversations. O.F. wishes to thank the Woodrow Wilson National Fellowship Foundation for their support through the Career Enhancement Fellowship. We thank IMPA for its hospitality in the initial stage of this project. Finally we would like to acknowledge the comments of the referees that improve the final version of the paper.

\section{Nonholonomic systems and the full reduction} \label{S:OneStep}

\subsection{Nonholonomic systems} \label{Ss:NH}
A nonholonomic system on a configuration manifold $Q$ is a
mechanical system with non-integrable constraints on the velocities.
It is described by a Lagrangian function $L:TQ \to \R$ of mechanical type,
i.e., $L= \frac{1}{2} \kappa - U$, where $\kappa$ is a metric and $U \in C^\infty(Q)$, and a non-integrable
subbundle $D\subset TQ$ (i.e., $D$ is not involutive).

Since $L$ is of mechanical type, the Legendre transformation $\L:TQ
\to T^*Q$ is an isomorphism such that $\L = \kappa^\flat$, where
$\kappa^\flat:TQ \to T^*Q$ is defined by  $\kappa^\flat(X)(Y) =
\kappa(X,Y)$, for $X,Y \in TQ$. The {\it constraint momentum space}
$\M$ is the submanifold of $T^*Q$ given by $\M :=\kappa^\flat(D)$.
Since $\kappa^\flat$ is linear on the fibers, it induces a
well-defined subbundle $\tau:\M \to Q$ of the cotangent bundle $T^*Q
\to Q$.

Using local canonical coordinates $(q, \dot q)$ on $TQ$, the
Hamiltonian function is defined by $\Ham = (\frac{\partial
L}{\partial \dot q^i} \dot q^i - L) \circ \kappa^\flat:T^*Q \to \R$.
If $\iota :\M \to T^*Q$ is the natural inclusion, we denote by
$\Ham_\subM := \iota^* \Ham: \M \to \R$ the restriction of $\Ham$ to
$\M$ and by $\Omega_\subM := \iota^*\Omega_Q$ the pull back of the
canonical 2-form $\Omega_Q$ on $T^*Q$ to $\M$.

The distribution $D$ on $Q$ defines a non-integrable distribution $\C$ on $\M$ by
\begin{equation} \label{Eq:C}
\C :=\{v \in T\M \ : \ \tau(v) \in D\},
\end{equation}
and, since  $\Omega_\subM |_\C$ is non-degenerate \cite{BS93}, the
dynamics is described by the (unique) vector field $X_\nh$ on $\M$
that takes values in $\C$ and satisfies the {\it Hamilton equations
for nonholonomic systems},
\begin{equation} \label{Eq:NHdynamics}
{\bf i}_{X_\nh} \Omega_\subM |_\C = d \Ham_\subM |_\C,
\end{equation}
where $\ |_\C$ represents the restriction to $\C$. Moreover, there
is a related $\R$-bilinear bracket $\{\cdot , \cdot \}_\nh :
C^\infty(\M) \times C^\infty(\M) \to C^\infty(\M)$  defined by
$$
\{\cdot , f \}_\nh = X_f \qquad \mbox{if and only if} \qquad {\bf i}_{X_f} \Omega_\subM |_\C = d f |_\C,
$$
for $f \in C^\infty(\M)$. The bracket $\{ \cdot , \cdot \}_\nh$ is known as the {\it nonholonomic bracket} (\cite{CdLMD,IdLMM,Marle,SchaftMaschke1994}), and $\{f, g \}_\nh = \Omega_\subM(X_f, X_g)$. The nonholonomic bracket is an almost Poisson bracket: it is skew-symmetric and satisfies the Leibniz condition,
$$
\{fg,h\}_\nh = f\{g,h\}_\nh + \{f,h\}_\nh g, \qquad \mbox{for all } f,g,h \in C^\infty(\M).
$$
In general, we say that a bracket $\{ \cdot, \cdot \}$ on a manifold is {\it Poisson} if it is an almost Poisson bracket that also satisfies the Jacobi identity,
$$
\{f,\{g,h\}\} + \{g,\{h,f\}\} + \{h,\{f,g\}\} = 0,
$$
for $f,g,h$ smooth functions on the manifold. The {\it
characteristic distribution} of an almost Poisson bracket is the
distribution generated by the Hamiltonian vector fields. If the
bracket is Poisson then its characteristic distribution is
integrable and each leaf carries a symplectic form. Observe that the
distribution $\C$ defined in \eqref{Eq:C} is the characteristic
distribution of $\{ \cdot, \cdot \}_\nh$, and it is not integrable
since $D$ is not integrable, \cite{BS93,PL2011}.

We say that an almost Poisson bracket $\{\cdot,\cdot\}$ {\it
describes the nonholonomic dynamics} if $X_\nh = \{ \cdot ,
\Ham_\subM \}$. In particular, the nonholonomic bracket describes the nonholonomic dynamics.


\subsection{Reduction by symmetries}
Consider a nonholonomic system on the manifold $Q$ given by a
Lagrangian $L$ of mechanical type and a non-integrable (regular)
distribution $D$. Let $G$ be a Lie group acting freely and properly on $Q$ such that the
lifted action to $TQ$ leaves $L$ and $D$ invariant. Then the
constraint momentum space $\M=\kappa^\flat(D)$ is an invariant
submanifold of $T^*Q$ by the cotangent lift of the $G$-action, so we
have a well defined action on $\M$.

The nonholonomic bracket $\{ \cdot , \cdot \}_\nh$ and the Hamiltonian $\Ham_\subM$ are also invariant by the $G$-action on $\M$. Thus, the orbit projection $\rho:\M \to \M/G$ induces a reduced bracket $\{\cdot, \cdot \}_\red$ on $\M/G$ given by
\begin{equation} \label{Eq:RedNHbracket}
\{f \circ \rho , g \circ \rho \}_\nh = \{f,g \}_\red \circ \rho,
\qquad \mbox{for } f,g \in C^\infty(\M/G).
\end{equation}
If $\Ham_\red: \M/G \to \R$ denotes the reduced Hamiltonian
function, $\rho^*\Ham_\red = \Ham_\subM$, then the (reduced)
dynamics is given by the vector field $X_\red$ on $\M/G$ such that
$$
\{ \cdot, \Ham_\red \}_\red = X_\red.
$$

The bracket $\{ \cdot , \cdot \}_\red$ is an almost Poisson bracket
on $\M/G$ whose characteristic distribution might also be
nonintegrable. In this paper we will study a particular situation
where the bracket $\{ \cdot , \cdot \}_\red$ has an integrable
characteristic distribution (even though the Jacobi identity might
not be satisfied), namely when the bracket is  {\it twisted
Poisson}.

More precisely, a {\it twisted Poisson bracket}
\cite{SeveraWeinstein} on a manifold $P$ is an almost Poisson
bracket for which there exists a closed 3-form $\phi$  on $P$ such
that
\begin{equation} \label{Def:Twisted}
\{f,\{g,h\}\} + \{g,\{h,f\}\} + \{h,\{f,g\}\} = \phi(X_f,X_g, X_h),
\end{equation}
where $f,g,h \in C^\infty(P)$ and $X_f, X_g, X_h$ are the
Hamiltonian vector fields defined by the bracket. The characteristic
distribution of a twisted Poisson bracket is integrable and its
associated leaves are almost symplectic manifolds. Twisted Poisson
brackets can be seen as an intermediate case between a bracket that
has a nonintegrable characteristic distribution (e.g. the
nonholonomic bracket) and a Poisson bracket. Moreover, if a regular
bracket admits a conformal factor (i.e., there is a function $f$
such that $f\{\cdot, \cdot \}$ is Poisson), then the bracket is
twisted Poisson. Searching for twisted Poisson brackets describing
the (reduced) dynamics can thus be seen as a first step toward the
Hamiltonization process.

Next, following \cite{paula} we introduce the necessary objects to
state the conditions under which the reduced bracket $\{ \cdot ,
\cdot \}_\red$ is twisted Poisson.
\bigskip

\noindent {\bf The dimension assumption and complements of the
constraints.} \ Let $G$ be a Lie group acting on $Q$ leaving $L$ and
$D$ invariant. In this paper, we assume, at each $q \in Q$,
\begin{equation} \label{Eq:DimensionAssumption}
 T_qQ = D_q +V_q,
\end{equation}
where $V_q$ is the tangent to the orbit at $q$ of the $G$-action.
Condition \eqref{Eq:DimensionAssumption} is called the {\it
dimension assumption} (see e.g. \cite{Bloch:Book}). Let us denote by
$S$ the distribution defined on $Q$, at each $q\in Q$, by
\begin{equation}\label{Eq:DefS}
S_q = D_q \cap V_q.
\end{equation}
A distribution $W$ on $Q$ is said to be a {\it vertical complement}
of the constraints \cite{paula} if, at each $q\in Q$,
\begin{equation}\label{Eq:D+W}
T_qQ = D_q \oplus W_q, \qquad \mbox{and} \qquad W_q \subseteq V_q.
\end{equation}

Consider the $G$-action on $\M$ and denote by $\V\subset T\M$ the
tangent distribution to the orbits. Observe that $T\tau|_\V :\V \to
V$ is an isomorphism (recall that $\tau:\M\to Q$ is the canonical
projection). The dimension assumption can be stated in terms of the
distribution $\C$ defined in \eqref{Eq:C}: at each $m \in \M$, the
dimension assumption \eqref{Eq:DimensionAssumption} guarantees that
$$
T_m\M = \C_m + \V_m.
$$
Analogously to \eqref{Eq:DefS}, we denote by $\S$ the distribution
on $\M$ such that, for $m \in \M$, $\S_m = \C_m \cap \V_m$. The
decomposition \eqref{Eq:D+W} also induces a decomposition on $T\M$,
such that, for each $m \in \M$, we have
\begin{equation} \label{Eq:C+W}
T_m\M = \C_m \oplus \W_m \qquad \mbox{and} \qquad \W_m \subseteq \V_m,
\end{equation}
where
\begin{equation} \label{Eq:DefW}
\W = (T\tau|_\V)^{-1}(W).
\end{equation}

\begin{remark}
 A $G$-invariant vertical complement of $D$ always exists since it can be chosen to be $W = V \cap D^{\perp}$, where $D^{\perp}$ is the orthogonal complement of $D$ with respect to the $G$-invariant kinetic energy metric $\kappa$. However, as it was observed in \cite{paula}, this choice of $W$ is not in general the most convenient when working with particular examples. In fact, in Sec. \ref{sec:examples} we study three different examples where the choice of $W$ is not $\kappa$-orthogonal to the constraints.
\end{remark}

\noindent {\bf The nonholonomic momentum map.}  \ Following \cite{BKMM}, define the subbundle $\mathfrak{g}_\subS \to \M$ of the
trivial bundle $\mathfrak{g}_\subM=\M \times \mathfrak{g} \to \M$ by
$$
\xi \in \mathfrak{g}_\subS |_m \ \Leftrightarrow \xi_\subM(m) \in \S_m.
$$
The {\it nonholonomic momentum map} is a map ${\mathcal J}^\nh:\M \to \mathfrak{g}_\subS^*$ given, at each $m \in \M$, by
\begin{equation} \label{Eq:NHMomMap}
\langle {\mathcal J}^\nh(m), \xi \rangle = {\bf i}_{\xi_\subM} \Theta_\subM, \qquad \mbox{for } \xi \in \mathfrak{g}_\subS,
\end{equation}
where $\Theta_\subM = \iota^*\Theta_Q$ and $\Theta_Q$ is the canonical 1-form on $T^*Q$.
Note that ${\mathcal J}^\nh:\M \to \mathfrak{g}_\subS^*$ satisfies
$$
{\bf i}_{\xi_\M} \Omega_\subM|_\C =  d\langle {\mathcal
J}^\nh,\xi\rangle |_\C - \pounds_{\xi_\M} \Theta_\subM|_\C,
$$
where $\xi$ is a (not necessarily constant) section of
$\mathfrak{g}_\subS$.  Note that $\langle {\mathcal J}^\nh,
\xi\rangle$ is not necessarily conserved by the nonholonomic motion,
since $\pounds_{\xi_\M} \Theta_\subM (X_\nh)$ might be different
from zero.
\bigskip

%

\noindent {\bf The 2-form $\langle {\mathcal J} , \mathcal{K}_\subW \rangle$.}  \ Denote by $P_{\mbox{\tiny{$D$}}}:TQ \to D$ and $P_{\mbox{\tiny{$W$}}}:TQ \to W$ the projections associated with the decomposition \eqref{Eq:D+W} for a $G$-invariant $W$. Following \cite{BKMM}, consider the map ${\bf A}_{\mbox{\tiny{$W$}}}: TQ \to \mathfrak{g}$ such that, for $v_q \in T_qQ$,
 ${\bf A}_{\mbox{\tiny{$W$}}} (v_q)= \xi$  if and only if $P_{\mbox{\tiny{$W$}}} (v_q) = \xi_Q(q).$
We define the $\mathfrak{g}$-valued 2-form ${\bf K}_{\mbox{\tiny{$W$}}}$ on $Q$ given at each $X, Y \in TQ$  by
\begin{equation} \label{Eq:Def_K}
{\bf K}_{\mbox{\tiny{$W$}}} (X, Y ) = d{\bf A}_{\mbox{\tiny{$W$}}} (P_{\mbox{\tiny{$D$}}} (X), P_{\mbox{\tiny{$D$}}} (Y)).
\end{equation}
The projection $\tau:\M \to Q$ induces the $\W$-curvature $\mathcal{K}_{\mbox{\tiny{$\mathcal{W}$}}}$, which is the $\mathfrak{g}$-valued 2-form on $\M$ defined by
$$
\mathcal{K}_\subW(X,Y) =  {\bf K}_{\mbox{\tiny{$W$}}} (T\tau (X), T\tau(Y)) \qquad \mbox{for } X,Y \in T\M.
$$
This definition coincides with the one given in \cite{paula}: $\mathcal{K}_{\mbox{\tiny{$\mathcal{W}$}}} (X, Y ) = d \mathcal{A}_\subW(P_{\mbox{\tiny{$\C$}}} (X), P_{\mbox{\tiny{$\C$}}} (Y)), $
where $ {\mathcal A}_\subW  = \tau^*{\bf A}_{\mbox{\tiny{$W$}}} : T\M \to \mathfrak{g}$, i.e., for $v_m \in T_m\M$,
 \begin{equation}\label{Eq:W-connection}
{\mathcal A}_\subW (v_m)= \xi \qquad \mbox{if and only if} \qquad P_\subW (v_m) = \xi_\subM(m),
\end{equation}
with $P_{\mbox{\tiny{$\C$}}}:T\M \to \C$ and $P_\subW:T\M \to \W$ the projections associated with the decomposition \eqref{Eq:C+W}.

Therefore, if  ${\mathcal J}:\M \to \mathfrak{g}^*$ is the restriction to $\M$ of the canonical momentum map $J:T^*Q \to \mathfrak{g}^*$ (i.e., ${\mathcal J} = \iota^* J$) then there is a well defined $G$-invariant 2-form on $\M$  given by
\begin{equation}\label{Eq:JK}
\langle {\mathcal J}, {\mathcal K}_\subW \rangle,
\end{equation}
where $\langle \cdot , \cdot \rangle$ is the natural pairing between
$\mathfrak{g}^*$ and $\mathfrak{g}$, see \cite{paula}.
\bigskip

\noindent {\bf The vertical-symmetry condition.} \ Following
\cite{paula}, we say that $W$ satisfies the {\it vertical-symmetry
condition} if there is a Lie subalgebra $\mathfrak{g}_\subW \subseteq
\mathfrak{g}$ such that
\begin{equation}\label{Eq:VerticalSymmetries}
W_q=\{\eta_Q(q) \ : \ \eta \in \mathfrak{g}_\subW \}.
\end{equation}

If $W$ satisfies the vertical-symmetry condition then the vertical
complement $\W$ of $\C$, defined in \eqref{Eq:DefW}, will satisfy
this condition as well, since for the Lie algebra
$\mathfrak{g}_\subW$,
$$
\W_m = \{\eta_\subM(m) \ : \ \eta \in \mathfrak{g}_\subW \}.
$$

\ \\ 

We remark that the vertical-symmetry condition is a condition imposed on the chosen complement $W$ of the constraints. We may find two different complements where one admits the vertical-symmetry condition while the other not. 
On the other hand, the {\it horizontal-symmetry condition} \cite{BKMM} is a condition of a different nature since it is imposed on the constraints. More precisely, $\xi \in \mathfrak{g}$ is an {\it horizontal symmetry} if $\xi_Q \in \Gamma(D)$ (or $\xi_\subM \in \Gamma(\C)$).  If $\mathfrak{g}_\subS$ is a Lie subalgebra of $\mathfrak{g}$ then the nonholonomic system has horizontal symmetries. Most of the examples developed in \cite{paula} and, in particular the Chaplygin ball carried out in section \ref{Ss:Ex:Ball} do not have horizontal symmetries. However we show that it is possible to choose a vertical complement $\W$ of the constraints satisfying the vertical-symmetry condition. 

When it is possible to choose a vertical complement $\W$ of the constraints such that the vertical-symmetry condition is satisfied,  several facts are simplified:

\begin{proposition}[{\cite{paula}}] \label{Prop:ConseqVertSym} If the vertical complement $W$ of the constraints $D$  is $G$-invariant and satisfies the vertical-symmetry condition \eqref{Eq:VerticalSymmetries} then
\begin{enumerate}
 \item[$(i)$] The Jacobi identity of the reduced bracket $\{ \cdot, \cdot \}_{\emph\red}$ on $\M/G$ becomes
 $$\{f,\{g,h\}_{\emph\red}\}_{\emph\red} \circ \rho + \textup{cyclic} = - d\langle {\mathcal J}, {\mathcal K}_\subW \rangle (X_{\rho^*f},X_{\rho^*g},X_{\rho^*h}),$$
 where $f,g,h \in C^\infty(\M/G)$ and $X_{\rho^*\!f}= \{ \cdot, \rho^*f \}_{\emph\nh}$ \textup{(}analogously with $X_{\rho^*g},X_{\rho^*h}$\textup{)}.

%
 \item[$(ii)$] At each $m \in \M$, the Lie algebra $\mathfrak{g}$ splits as follows:
\begin{equation} \label{eq:proj}
\mathfrak{g} = \mathfrak{g}_\subS |_m \oplus \mathfrak{g}_\subW.
\end{equation}

 \item[$(iii)$] For each $\eta \in \mathfrak{g}$, let us denote by $X_\eta = \left(  P_{\mathfrak{g}_\subS}(\eta) \right)_{\!\subM}$
 the infinitesimal generator on $\M$ associated to $P_{\mathfrak{g}_\subS}(\eta) \in \mathfrak{g}_\subS$ where $P_{\mathfrak{g}_\subS} : \mathfrak{g}_\subM \to \mathfrak{g}_\subS$ is the projection defined in \eqref{eq:proj}.
 Then
 $$
{\bf i}_{X_\eta} (\Omega_\subM + \langle {\mathcal J}, \mathcal{K}_\subW \rangle ) |_\C = d   \langle {\mathcal J}^{\emph{\nh}} , P_{\mathfrak{g}_\subS}(\eta) \rangle|_\C.
$$
\end{enumerate}
\end{proposition}

\noindent {\bf The Reduction Theorem.}  \ From Prop.~\ref{Prop:ConseqVertSym} we observe that the reduced bracket $\{\cdot, \cdot \}_\red$ is not necessarily twisted Poisson \eqref{Def:Twisted} since $d \langle {\mathcal J}\!,\!{\mathcal K}_\subW \rangle$ might not be a well defined 3-form on $\M/G$.

Recall that a form $\alpha$ on a manifold $P$ is {\it basic} with respect to the projection $\rho_M:P\to M$ if there is a form $\bar{\alpha}$ on $M$ such that $\rho^*_M\bar{\alpha} = \alpha$.

\begin{theorem}[\cite{paula}]\label{T:Reduced-Dyn}
Let $(\M, \{\cdot , \cdot \}_{\emph\nh}, \Ham_\subM)$ be a nonholonomic
system with a $G$-symmetry, and assume that the $G$-invariant vertical complement $W$ verifies
the vertical-symmetry condition \eqref{Eq:VerticalSymmetries}.
If the 3-form  $d\langle {\mathcal J}, \mathcal{K}_\subW \rangle$, defined in \eqref{Eq:JK}, is basic with
respect to the orbit projection $\rho: \M \to \M /G$, 
then the reduced bracket $\{\cdot, \cdot \}_{\emph\red}$ is twisted Poisson.
\end{theorem}

\begin{corollary} \label{Cor:ConservedFunction} \textup{(\cite[Corollary 6.9]{paula})} If $\langle {\mathcal J}, \mathcal{K}_\subW \rangle$ is basic with respect to the orbit projection $\rho: \M \to \M /G$, then at each $\eta \in \mathfrak{g}$ the function  $f_\eta$ on $\M$ defined by
$$
f_\eta (m)= \langle {\mathcal J}^{\emph\nh}(m), P_{\mathfrak{g}_\subS}(\eta)\rangle
$$
is conserved by the nonholonomic motion, where $P_{\mathfrak{g}_\subS}:\mathfrak{g} \to \mathfrak{g}_\subS$ is the projection associated to the decomposition \eqref{eq:proj}.
\end{corollary}

\begin{remark} In references \cite{BoBoMa,BoMaBi,BoMaTsi} the authors studied the problem of Hamiltonization based in finding a Poisson bracket or a conformally Poisson bracket describing the reduced dynamics. Several examples are analyzed in these references.  
\end{remark}

\begin{remark} In \cite{paula} there is a description of the almost symplectic leaves in terms of the nonholonomic momentum map. The 2-form on each leaf involves the symplectic form of the Marsden-Weinstein reduction of the canonical symplectic manifold $(T^*Q, \Omega_{\mbox{\tiny{$Q$}}})$ and a diffeomorphism between $\M$ and the annihilator $W^\circ$ of the vertical complement $W$ of $D$.
In this paper, we will study these leaves from another viewpoint, involving conserved quantities.  As a consequence it will be possible to see the equations of motion as a nonholonomic version of the Lagrange-Routh equations on each almost symplectic leaf.
\end{remark}

\section{Reduction in two steps: compression and reduction} \label{S:2steps}

In this section we will perform, assuming the vertical-symmetry
condition, the reduction of a nonholonomic system in two steps,
unifying previous works \cite{paula,EhlersKoiller,Hoch,Oscar}.

Consider a nonholonomic system on a manifold $Q$ given by the
mechanical-type Lagrangian $L$ and the non-integrable distribution
$D$ with a $G$-symmetry. On the manifold $\M$ the nonholonomic
bracket $\{\cdot, \cdot \}_\nh$ and the Hamiltonian $\Ham_\subM$ are
invariant by the $G$-action on $\M$, denoted by
$$
\varphi : G \times \M \to \M.
$$

Let us choose a $G$-invariant vertical complement $W$ of $D$ in $TQ$
satisfying the {\it vertical-symmetry condition}
\eqref{Eq:VerticalSymmetries} and define the complement $\W$ of $\C$
in $T\M$ as in \eqref{Eq:DefW}. The Lie subalgebra
$\mathfrak{g}_\subW$ defined in \eqref{Eq:VerticalSymmetries} is an
ideal, due to the $G$-invariance of $\W$: if $\eta\in \mathfrak{g}$
and $\tilde \eta \in \mathfrak{g}_\subW$, then $[\eta,\tilde
\eta]_\subM = [\eta_\subM,\tilde \eta_\subM] \in \W$ and thus
$[\eta,\tilde \eta] \in \mathfrak{g}_\subW$.  Therefore, the
connected integration of $\mathfrak{g}_\subW$ is a normal subgroup
$G_\subW$ of $G$ that acts freely and properly on $\M$ by restricting the
action:
\begin{equation} \label{Eq:Gw}
\varphi_\subW: G_\subW \times \M \to \M \qquad \mbox{such that} \qquad \varphi_\subW(g,m) = \varphi(g,m) \mbox{ \ for \ } g\in G_\subW, \ m \in \M.
\end{equation}
Hence, we can consider the orbit projection $\rho_\subW :\M  \to \M/G_\subW$. On the other hand, $H:=G/G_\subW$ is a Lie group and the natural projection  $\varrho_G:G \to H$ is a smooth group homomorphism. The Lie group $H$ acts on $\M/G_\subW$ with the action $\varphi_H: H \times \M/G_\subW \to \M/G_\subW$ such that
\begin{equation}\label{Eq:H-action}
\varphi_H(\varrho_G(g), \rho_\subW(m)) = \rho_\subW \varphi (g,m) \qquad \mbox{for \ } g \in G, \ m\in \M.
\end{equation}
We denote by $\mathfrak{h}: = \mathfrak{g}/\mathfrak{g}_\subW$ the Lie algebra of $H$ (for more details see e.g.
\cite{DuisKo2000:Book}).

Our goal is to realize the reduced almost symplectic leaves
associated to $\{ \cdot , \cdot \}_\red$ in a two-step procedure:
first we {\it compress} by $G_\subW$ and then, under certain
conditions,  we perform an {\it almost symplectic reduction} by the
Lie group $H$.

\begin{remark}
 Since $G$ acts originally on $Q$ (the action on $\M$ is the restriction of the lifted action to $T^*Q$) then $G_\subW$ acts naturally on $Q$ as in \eqref{Eq:Gw} and $H$ acts on $Q/G_\subW$ as in \eqref{Eq:H-action}.
\end{remark}

From now on, we fix a vertical complement $W$ of $D$ with the
vertical-symmetry condition \eqref{Eq:VerticalSymmetries}, so that
the Lie group $G_\subW$ acts on $Q$ and on $\M$.

\subsection{Compression} \label{Ss:2S-Compression}
Since $G$ is a symmetry of the nonholonomic system, so is the normal
subgroup $G_\subW$ \cite{MMORR}. The vertical space associated to
the $G_\subW$-action on $Q$ is $W$, which is a complement of the
constraints by \eqref{Eq:D+W}. Thus the nonholonomic system is
$G_\subW$-Chaplygin \cite{Koiller1992}. Let us denote by  $\bar Q :=
Q/G_\subW$ the reduced manifold, and recall that $\M/G_\subW \simeq
T^*\bar{Q}$, so that $\rho_\subW: \M \to T^*\bar{Q}$ denotes the
orbit projection induced by the $G_\subW$-action on $\M$. After
compression, the dynamics takes place on $T^*\bar{Q}$ (see e.g.,
\cite{Koiller1992,Bloch:Book}).

In this case, note that the map $\mathcal{A}_\subW:TQ \to
\mathfrak{g}_\subW$ defined in \eqref{Eq:W-connection} is a
principal connection, and thus  ${\mathcal K}_\subW$ can be viewed
as a $\mathfrak{g}_\subW$-valued 2-form that is just its associated
curvature.

\begin{lemma}\label{L:JKbasic}
 The 2-form  $\langle {\mathcal J}, \mathcal{K}_\subW \rangle$ on $\M$ defined in \eqref{Eq:JK} is basic with respect to the orbit projection $\rho_\subW:\M \to T^*\bar{Q}$. 
\end{lemma}
\begin{proof}
Denote by ${\mathcal J}_\subW : \M \to \mathfrak{g}_\subW^*$ the
restriction to $\M$ of the canonical momentum map associated to the
$G_\subW$-action on $T^*Q$; then  $\langle {\mathcal J},\xi \rangle
= \langle \mathcal{J}_\subW,\xi\rangle$ for all $\xi \in
\mathfrak{g}_\subW$. Viewing $\mathcal{K}_\subW$ as a
$\mathfrak{g}_\subW$-valued 2-form one can define the 2-form
$\langle {\mathcal J}_\subW,{\mathcal
K}_\subW\rangle_{\mathfrak{g}_\subW}$, where $\langle\cdot, \cdot
\rangle_{\mathfrak{g}_\subW}$ is the natural pairing on
$\mathfrak{g}_\subW$. It was shown in  \cite{MovingFrames} that
$\langle {\mathcal J}_\subW,{\mathcal
K}_\subW\rangle_{\mathfrak{g}_\subW}$ is basic with respect to the
orbit projection $\rho_\subW: \M \to \M/G_\subW$.  Finally, observe
that $\langle {\mathcal J}, \mathcal{K}_\subW \rangle$ defined in
\eqref{Eq:JK} coincides with  $\langle {\mathcal J}_\subW,{\mathcal
K}_\subW\rangle_{\mathfrak{g}_\subW}$.
\end{proof}
We denote by $\overline{\langle {\mathcal J},{\mathcal K}_\subW\rangle}$ the 2-form on $T^*\bar{Q}$  such that $\rho_\subW^*\overline{\langle {\mathcal J},{\mathcal K}_\subW\rangle} = \langle {\mathcal J},{\mathcal K}_\subW\rangle. $
If $\Ham_{\mbox{\tiny{$T^*\bar{Q}$}}}$ denotes the Hamiltonian
function on $T^*\bar{Q}$ so that
$\rho_\subW^*\Ham_{\mbox{\tiny{$T^*\bar{Q}$}}} = \Ham_\subM$, then
the {\it compressed Hamilton equations} are given by
\begin{equation}\label{Eq:ChapReduction}
{\bf i}_{\bar{X}_\nh}\bar{\Omega} = d\Ham_{\mbox{\tiny{$T^*\bar{Q}$}}},
\end{equation}
where
\begin{equation}\label{Eq:Chap2form}
 \bar{\Omega} := \Omega_{\mbox{\tiny{$\bar{Q}$}}} - \overline{\langle {\mathcal J},{\mathcal K}_\subW\rangle},
\end{equation}
with $\Omega_{\mbox{\tiny{$\bar{Q}$}}}$ the canonical 2-form on
$T^*\bar{Q}$, \cite{MovingFrames}. The {\it compressed dynamics} is
described by the integral curves of the vector field $\bar{X}_\nh$
on $T^*\bar{Q}$ (see \cite{Koiller1992,Hoch,BS93}).

\begin{remark}
In \cite{Fernandez2} the particular case when $d \overline{\langle
{\mathcal J},{\mathcal K}_\subW\rangle}$ vanishes was studied. In
this case,  $\bar{\Omega}$ is a genuine symplectic 2-form and the
compressed system \eqref{Eq:ChapReduction} is Hamiltonian.  However,
in general $d\bar{\Omega} \neq 0$, and so $\bar{\Omega}$ is only an
almost symplectic 2-form.
\end{remark}



Recall that $\mathfrak{h}:= \mathfrak{g}/\mathfrak{g}_\subW$ is the Lie algebra associated to the Lie group $H$.
\begin{proposition}\label{Prop:f-conserved}
Let $W$ be a $G$-invariant complement of the constraints \eqref{Eq:D+W} that satisfies the
vertical-symmetry condition \eqref{Eq:VerticalSymmetries}, and
denote by  $\varrho : \mathfrak{g} \to \mathfrak{h}$ the natural
projection. Then:
\begin{itemize}
 \item[$(i)$] For each $\eta \in \mathfrak{g}$ the function $f_\eta$ defined in Corollary \ref{Cor:ConservedFunction} is basic with respect to the orbit projection $\rho_\subW: \M \to \M/G_\subW \simeq T^*\bar{Q}$.
 \item[$(ii)$] For all $\eta, \tilde \eta \in \mathfrak{g}$ for which $\varrho(\eta)=\varrho(\tilde\eta)$ holds, $f_\eta = f_{\tilde\eta}$ .
\end{itemize}
\end{proposition}

\begin{proof}
To show $(i)$ it is sufficient to prove that $df_\eta (\xi_\subM) = 0$ for all $\xi \in \mathfrak{g}_\subW$.
First, observe that since the annihilator $W^\circ \subset T^*Q$  of $W$ is $G$-invariant we can consider $J_0 := J \circ \iota_0 :W^\circ \to \mathfrak{g}^*$, the restriction of the canonical momentum map $J:T^*Q \to \mathfrak{g}^*$ to $W^\circ$.
It was observed in \cite[Theorem 6.5]{paula} that, at each $ m \in \M, \eta \in \mathfrak{g}$,
$$ 
 \langle {\mathcal J}^\nh (m) , P_{\mathfrak{g}_\subS}(\eta)
\rangle = \langle J_0 \circ \Psi (m) , \eta \rangle,
$$
where $\Psi: \M \to W^\circ$ is the $G$-invariant isomorphism $\Psi:\M \to W^\circ$ given by $\Psi = (\kappa_0^\flat |_D) \circ (\kappa|_D^\flat)^{-1} : \M \to W^\circ,$
for $\kappa_0$ a metric on $Q$ such that $D$ and $W$ are $\kappa_0$-orthogonal and $P_{\mathfrak{g}_\subS}: \mathfrak{g} \to \mathfrak{g}_\subS$ is the projection associated to decomposition \eqref{eq:proj}. 
Now,  for $\eta$ and $\xi \in \mathfrak{g}$, 
$$
df_\eta(\xi_\subM) = d \langle J_0 \circ \Psi, \eta \rangle( \xi_\subM) = \Psi^* d \langle J_0, \eta \rangle (\xi_\subM) = d\langle J_0 , \eta \rangle (T\Psi (\xi_\subM)).
$$
Observe that the map $J_0:W^\circ \to \mathfrak{g}^*$ is a momentum map for the presymplectic manifold $(W^\circ, \Omega_{\mbox{\tiny{$W^\circ$}}})$, where $\Omega_{\mbox{\tiny{$W^\circ$}}} = \iota^*_0 \Omega_Q$ and $\iota_0:W^\circ \to T^*Q$ is the natural inclusion. Moreover, it was shown in \cite[Lemma 5.6]{paula} that $\textup{Ker}\, \Omega_{\mbox{\tiny{$W^\circ$}}} = T\Psi(\W).$ Therefore, we see that $d\langle J_0 , \eta \rangle (T\Psi (\xi_\subM)) = \Omega_{\mbox{\tiny{$W^\circ$}}} (\eta_{\mbox{\tiny{$W^\circ$}}},T\Psi( \xi_\subM)) = 0$ and we conclude that $df_\eta(\xi_\subM) =0$.

Finally, to prove $(ii)$ observe that $P_{\mathfrak{g}_\subS}(\tilde\eta) = P_{\mathfrak{g}_\subS}(\eta)$ when $\varrho(\eta)=\varrho(\tilde\eta)$.
Thus, for all $m \in \M$, $
f_{\tilde\eta}(m) = \langle {\mathcal J}^\nh(m), P_{\mathfrak{g}_\subS}(\tilde\eta) \rangle = \langle {\mathcal J}^\nh(m), P_{\mathfrak{g}_\subS}(\eta) \rangle = f_\eta(m).$
\end{proof}

Therefore,  for each $\xi \in \mathfrak{h}$ there is a well-defined function $g_\xi \in C^\infty(T^*\bar{Q})$ such that
\begin{equation}\label{Eq:Defg_xi}
\rho_\subW^*g_\xi = f_\eta = \langle {\mathcal J}^\nh , P_{\mathfrak{g}_\subS}(\eta) \rangle.
\end{equation}
where  $\eta$ is any element in $\mathfrak{g}$ such that $\varrho(\eta)=\xi$.
If the function $f_\eta$ is conserved by $X_\nh$ then $g_\xi$ is conserved by the compressed motion $\bar{X}_\nh$ defined in  \eqref{Eq:ChapReduction}.

\subsection{Almost symplectic reduction} \label{Ss:2S-H_action}

Since  the Lie group $G$ acts on $Q$, there is an induced action of
the Lie group $H=G/G_\subW$  on $\bar{Q}$ as in \eqref{Eq:H-action}.
The lifted $H$-action to $T^*\bar{Q}$ is a symmetry of the
compressed system \eqref{Eq:ChapReduction}, that is, the Hamiltonian
$\Ham_{\mbox{\tiny{$T^*\bar{Q}$}}}$ and the 2-form $\bar\Omega$
given in \eqref{Eq:Chap2form} are $H$-invariant. As a consequence,
we have that $\overline{\langle \mathcal{J},
\mathcal{K}_\subW\rangle}$ is also $H$-invariant.

\noindent {\bf A momentum map for $(T^*\bar{Q}, \bar{\Omega})$.} \
Recall that the lifted $H$-action to $T^*\bar{Q}$, denoted by
$\varphi_H$, is Hamiltonian with respect to the canonical symplectic
form, and thus the canonical momentum map $\bar{J}: T^*\bar{Q} \to
\mathfrak{h}^*$ satisfies
$$
{\bf i}_{\xi_{T^*\bar{Q}}} \Omega_{\mbox{\tiny{$\bar{Q}$}}} = d \langle \bar{J}, \xi\rangle, \qquad \mbox{for } \xi \in \mathfrak{h}.
$$


\begin{proposition}\label{Prop:g_xi} For each $\xi \in \mathfrak{h}$ the function $g_\xi$ defined in \eqref{Eq:Defg_xi} is the canonical momentum map on  $T^*\bar{Q}$; that
is,
$$g_\xi (\bar m) = \langle \bar{J}(\bar m), \xi \rangle \qquad \mbox{for all} \quad \bar m \in T^*\bar Q.$$
\end{proposition}

\begin{proof}
First, observe that by Prop. \ref{Prop:ConseqVertSym}(ii) and
recalling that $\W$ is the vertical space associated to the
$G_\subW$-action, $T\!\rho_\subW \left(
(P_{\mathfrak{g}_\subS}(\eta))_\M \right) = T\! \rho_\subW
(\eta_\subM)$ for $\eta \in \mathfrak{g}$. In addition, if $\xi \in
\mathfrak{h}$ is such that $\varrho(\eta)= \xi$, then
$T\!\rho_\subW(\eta_\M) = \xi_{T^*\bar{Q}}$ since, at each $m\in\M$,
we have
$$
\xi_{T^*\bar{Q}} (\rho_\subW(m)) = \frac{d}{dt} \varphi_H(exp(t\varrho(\eta)), \rho_\subW(m)) |_{t=0} =  \frac{d}{dt} \rho_\subW\varphi(exp(t\eta), m ) |_{t=0} = T\rho_\subW (\eta_\M (m)),
$$
where in the second equality we are using that $\varrho_G:G\to H$ is the homomorphism of Lie groups and that $\varrho =
T\varrho_G:\mathfrak{g}\to\mathfrak{h}$ (see \cite[Prop.~9.1.5]{MarsdenRatiu}). Hence, letting
$X_\eta = (P_{\mathfrak{g}_\subS}(\eta))_\M$, we obtain that
$T\!\rho_\subW (X_\eta ) = \xi_{T^*{\bar{Q}}}$. On the other hand,
Prop.~\ref{Prop:ConseqVertSym}(iii) implies that $ {\bf i}_{X_\eta}
\Omega_\subM |_\C = df_\eta |_\C - {\bf i}_{X_\eta} \langle
\mathcal{J}, \mathcal{K}_\subW \rangle |_\C$, and using
\eqref{Eq:Defg_xi} and Lemma \ref{L:JKbasic} we get
\begin{equation}\label{Proof:Prop:gxi}
{\bf i}_{X_\eta} \Omega_\subM |_\C = \rho_\subW^*dg_\xi |_\C - \rho_\subW^* \left( {\bf i}_{\xi_{T^*{\bar{Q}}}} \overline{\langle \mathcal{J}, \mathcal{K}_\subW \rangle} \right) |_\C.
\end{equation}
Using that $\rho_\subW^* \bar \Omega |_\C = \Omega_\subM |_\C$ \cite{MovingFrames}, then $\rho_\subW^* \left( {\bf i}_{\xi_{T^*\bar{Q}}} [\Omega_{\mbox{\tiny{$\bar{Q}$}}} - \overline{\langle \mathcal{J}, \mathcal{K}_\subW \rangle}] \right) |_\C = {\bf i}_{X_\eta} \Omega_\subM |_\C$, and thus, by \eqref{Proof:Prop:gxi}, we obtain that ${\bf i}_{\xi_{T^*\bar{Q}}} \Omega_{\mbox{\tiny{$\bar{Q}$}}} = d g_\xi.$  Finally, since ${\bf i}_{\xi_{T^*\bar{Q}}} \Omega_{\mbox{\tiny{$\bar{Q}$}}} = d \langle \bar{J}, \xi\rangle$, we obtain the desired result.
\end{proof}

Note that $\bar{J}:T^*\bar{Q} \to \mathfrak{h}^*$ is not necessarily a momentum map for the (almost symplectic) manifold $(T^*\bar{Q}, \bar{\Omega}= \Omega_{\bar Q} - \overline{\langle {\mathcal J}, {\mathcal K}_\subW\rangle})$ since
\begin{equation} \label{Eq:NotMomMap}
{\bf i}_{\xi_{T^*\bar{Q}}} \bar{\Omega}  = d \langle \bar{J}, \xi \rangle - {\bf i}_{\xi_{T^*\bar{Q}}} \overline{\langle {\mathcal J}, {\mathcal K}_\subW \rangle}.
\end{equation}
The next Proposition shows that  $\bar{J}:T^*\bar{Q} \to \mathfrak{h}^*$ may be conserved by the nonholonomic dynamics even if it is not a momentum map of $(T^*\bar{Q}, \bar{\Omega})$.

\begin{proposition} \label{Prop:JK=0}
For each  $\xi \in \mathfrak{h}$, the function $g_\xi \in C^\infty (T^*\bar{Q})$ is conserved by the compressed nonholonomic motion $\bar{X}_{\emph\nh}$ on $T^*\bar{Q}$ $($respectively $f_\eta \in C^\infty(\M)$ is conserved by the nonholonomic motion $X_{\emph\nh}$ for $\eta\in \mathfrak{g}$ such that $\xi = \varrho(\eta)$ $)$  if and only if $\overline{\langle {\mathcal J}, \mathcal{K}_\subW \rangle} (\bar{X}_{\emph\nh} , \xi_{\mbox{\tiny{$T^*\bar{Q}$}}} ) = 0$.
\end{proposition}
\begin{proof}
If $g_\xi$ is conserved then, by Prop.~\ref{Prop:g_xi}, $d\bar{J}_\xi (\bar{X}_\nh )= 0$ for $\bar{J}:T^*\bar{Q} \to \mathfrak{h}^*$ the canonical momentum map, that is, $0 =  \Omega_{\mbox{\tiny{$\bar{Q}$}}} (\bar{X}_\nh, \xi_{\mbox{\tiny{$T^*\bar{Q}$}}})$. On the other hand, by the $H$-invariance of the Hamiltonian $\Ham_{\mbox{\tiny{$T^*\bar{Q}$}}}$, we have that $0= d\Ham_{\mbox{\tiny{$T^*\bar{Q}$}}} (\xi_{\mbox{\tiny{$T^*\bar{Q}$}}}) =   \bar\Omega(\bar{X}_\nh, \xi_{\mbox{\tiny{$T^*\bar{Q}$}}})$ and thus by \eqref{Eq:Chap2form} $\overline{\langle {\mathcal J}, \mathcal{K}_\subW \rangle} (\bar{X}_{\nh} , \xi_{\mbox{\tiny{$T^*\bar{Q}$}}} ) = 0$. Conversely, again by the $H$-invariance of the Hamiltonian $\Ham_{\mbox{\tiny{$T^*\bar{Q}$}}}$ and \eqref{Eq:Chap2form}, we have that
$$
0 = d\Ham_{\mbox{\tiny{$T^*\bar{Q}$}}} (\xi_{\mbox{\tiny{$T^*\bar{Q}$}}} ) = \bar{\Omega}(\bar{X}_\nh, \xi_{\mbox{\tiny{$T^*\bar{Q}$}}})= \Omega_{\mbox{\tiny{$\bar{Q}$}}} (\bar{X}_\nh, \xi_{\mbox{\tiny{$T^*\bar{Q}$}}}) = - d\bar{J}_\xi (\bar{X}_\nh).
$$
Finally, by Prop.~\ref{Prop:g_xi} we conclude that $g_\xi$ is conserved and thus by Prop.~\ref{Prop:f-conserved} $f_\eta$ is conserved by $X_\nh$ for $\eta \in \mathfrak{g}_\subW$ such that $\xi = \varrho(\eta)$.
\end{proof}

Proposition \ref{Prop:JK=0} derives a necessary but not sufficient condition for $\bar{J}:T^*\bar{Q} \to \mathfrak{h}^*$ to be a momentum map of $\bar{\Omega}$. In fact, we have

\begin{corollary} \label{C:JKbasic=conservationLaws}
The canonical momentum map $\bar{J}:T^*\bar{Q} \to \mathfrak{h}^*$
is a momentum map for the manifold $(T^*\bar{Q}, \bar{\Omega}=
\Omega_{\bar Q} - \overline{\langle {\mathcal J}, {\mathcal
K}_\subW\rangle})$ if and only if $\overline{\langle {\mathcal J},
{\mathcal K}_\subW \rangle}$ is basic with respect to the orbit
projection $\rho_H: T^*\bar{Q}\to T^*\bar{Q}/H$.
\end{corollary}
\begin{proof} If $\bar{J}:T^*\bar{Q}\to \mathfrak{h}^*$ satisfies that ${\bf i}_{\xi_{\mbox{\tiny{$T^*\bar{Q}$}}}} \bar\Omega = d\bar{J}_\xi$ and ${\bf i}_{\xi_{\mbox{\tiny{$T^*\bar{Q}$}}}} \Omega_{\mbox{\tiny{$\bar{Q}$}}} = d\bar{J}_\xi$, then $\overline{\langle {\mathcal J}, {\mathcal K}_\subW\rangle}$ is semi-basic with respect to $\rho_H: T^*\bar{Q}\to T^*\bar{Q}/H$. Therefore $\overline{\langle {\mathcal J}, {\mathcal K}_\subW\rangle}$ is basic because it is $H$-invariant.  The converse is a direct consequence of \eqref{Eq:NotMomMap} and the $H$-invariance of $\overline{\langle {\mathcal J}, {\mathcal K}_\subW\rangle}$.
\end{proof}

\begin{notation} \label{Notation}  By Lemma \ref{L:JKbasic}, the 2-form $\overline{\langle {\mathcal J}, {\mathcal K}_\subW \rangle}$  defined on $\M/G_\subW \simeq T^*\bar{Q}$ is basic with respect to the orbit projection $\rho_H:T^*\bar{Q} \to T^*\bar{Q}/H$ if and only if the 2-form $\langle {\mathcal J}, {\mathcal K}_\subW \rangle$ on $\M$ is basic with respect to $\rho: \M \to \M/G$. If we assume that $\langle {\mathcal J}, {\mathcal K}_\subW\rangle$ is basic with respect to $\rho: \M \to \M/G$ then we denote by
\begin{enumerate}
 \item[$(i)$]  $\langle {\mathcal J}, {\mathcal K}_\subW \rangle_\red$ the 2-form on $\M/G$ such that $\rho^* \langle {\mathcal J}, {\mathcal K}_\subW \rangle_\red = \langle {\mathcal J}, {\mathcal K}_\subW \rangle$.
 \item[$(ii)$] $\overline{\langle {\mathcal J}, {\mathcal K}_\subW \rangle}_\red$ the 2-form on $T^*\bar{Q}/H$ such that $\rho_H^* \overline{\langle {\mathcal J}, {\mathcal K}_\subW \rangle}_\red = \overline{\langle {\mathcal J}, {\mathcal K}_\subW \rangle}$.
\end{enumerate}
It is clear that $\overline{\langle {\mathcal J}, {\mathcal K}_\subW \rangle}_\red = \langle {\mathcal J}, {\mathcal K}_\subW \rangle_\red$.
\end{notation}

%


\medskip

\noindent {\bf (Almost) symplectic reduction.}  \ Let us assume that the 2-form $\overline{\langle {\mathcal J}, {\mathcal K}_\subW \rangle}$ on $\M/G_\subW \simeq T^*\bar{Q}$ defined in Lemma \ref{L:JKbasic} is basic with respect to the orbit projection $\rho_H:T^*\bar{Q} \to T^*\bar{Q}/H$. Then by Corollary \ref{C:JKbasic=conservationLaws} 
$\bar{J}:T^*\bar{Q} \to \mathfrak{h}^*$ is a momentum map for the (not necessarily closed) 2-form $\bar \Omega= \Omega_{\bar{Q}} -  \overline{\langle {\mathcal J}, {\mathcal K}_\subW \rangle}$, and thus we can perform an (almost) symplectic reduction on  $(T^*\bar{Q}, \bar \Omega)$.
Let us denote by $H_{\bar\mu}$ the coadjoint isotropy group at $\bar{\mu} \in \mathfrak{h}^*$. Then, performing an (almost) symplectic reduction \cite{MW}, we obtain, at each $\bar{\mu} \in \mathfrak{h}^*$, an almost symplectic form $\bar{\omega}_{\bar{\mu}}$ on $\bar{J}^{-1}(\bar{\mu})/H_{\bar{\mu}}$ such that
\begin{equation} \label{Eq:SymplReduction}
\pi^*_{\bar{\mu}}(\bar{\omega}_{\bar{\mu}}) = \iota_{\bar{\mu}}^* (\bar{\Omega}),
\end{equation}
where $\pi_{\bar{\mu}}:\bar{J}^{-1}(\bar{\mu}) \to \bar{J}^{-1}(\bar{\mu})/H_{\bar{\mu}}$ is the projection to the quotient and $\iota_{\bar{\mu}}: \bar{J}^{-1}(\bar{\mu}) \to T^*\bar{Q}$ is the inclusion.

Therefore, the reduced dynamics is restricted to the level set $(\bar{J}^{-1}(\bar{\mu})/H_{\bar{\mu}}, \bar{\omega}_{\bar\mu})$, and is described by the vector field $X_\red^{\bar{\mu}} \in \mathfrak{X}(\bar{J}^{-1}(\bar{\mu}))$ given by
\begin{equation} \label{Eq:Hamilt-MWreduction}
{\bf i}_{X_\red^{\bar{\mu}}} \bar{\omega}_{\bar{\mu}} = d (\Ham_{\bar\mu} ),
\end{equation}
where $\Ham_{\bar\mu}: \bar{J}^{-1}(\bar{\mu})/H_{\bar{\mu}} \to \R$ is the Hamiltonian function such that $\pi^*_{\bar{\mu}} \Ham_{\bar\mu} = \iota_{\bar{\mu}}^* \Ham_{T^*\bar{Q}}$.

Now, we analyze the 2-form $\bar\omega_{\bar\mu}$ given in \eqref{Eq:SymplReduction}.  First, observe that since $\bar{J}:T^*\bar{Q} \to \mathfrak{h}^*$ is the canonical momentum map, then the Marsden-Weinstein reduction of the canonical symplectic manifold $(T^*\bar{Q}, \Omega_{\bar{Q}})$, gives the symplectic leaves
\begin{equation}\label{Eq:Proof:equivalence}
({\bar J}^{-1}(\bar{\mu})/H_{\bar{\mu}} , {\Omega}_{\bar{\mu}}),
\end{equation}
where $\pi^*_{\bar{\mu}}({\Omega}_{\bar{\mu}}) = \iota_{\bar{\mu}}^* (\Omega_{\bar{Q}} )$. Second, since the 2-form $\overline{\langle {\mathcal J}\!,\!{\mathcal K}_\subW \rangle}$ on $T^*\bar{Q}$ is basic with respect to the projection $\rho_H : T^*\bar{Q} \to T^*\bar{Q}/H$, we can consider the pull back of the 2-form $\overline{\langle {\mathcal J}\!,\!{\mathcal K}_\subW \rangle}_\red$ to the submanifold ${\bar J}^{-1}(\bar{\mu})/H_{\bar{\mu}}$, denoted by $[\overline{\langle {\mathcal J}\!,\!{\mathcal K}_\subW \rangle}_\red]_{\bar{\mu}}$. Then, observe that $\iota_{\bar\mu}^* \overline{\langle {\mathcal J}\!,\!{\mathcal K}_\subW \rangle} = \pi_{\bar\mu}^* \left( [\overline{\langle {\mathcal J}\!,\!{\mathcal K}_\subW \rangle}_\red]_{\bar{\mu}} \right)$. Hence,

\begin{proposition} \label{P:Reduced2-form} If the 2-form $\overline{\langle {\mathcal J}\!,\!{\mathcal K}_\subW \rangle}$ is basic with respect to the projection $\rho_H : T^*\bar{Q} \to T^*\bar{Q}/H$, then at each $\bar{\mu} \in \mathfrak{h}^*$ the
 (almost symplectic) reduction of $(T^*\bar{Q}, \bar\Omega = \Omega_{\mbox{\tiny{$\bar{Q}$}}}- \overline{\langle {\mathcal J}\!,\!{\mathcal K}_\subW \rangle})$ gives the almost symplectic manifold $({\bar J}^{-1}(\bar{\mu})/H_{\bar{\mu}} , \bar{\omega}_{\bar\mu})$, where
\begin{equation}\label{Eq:BarOmega-SymplecReduc}
\bar{\omega}_{\bar\mu} = {\Omega}_{\bar{\mu}} - [\overline{\langle {\mathcal J}\!,\!{\mathcal K}_\subW \rangle}_{\emph\red}]_{\bar{\mu}}.
\end{equation}
\end{proposition}

For each $\bar{\mu} \in \mathfrak{h}^*$, let us denote by $\tilde{\mathcal{O}}_{\bar\mu}$ the associated coadjoint bundle $\tilde{\mathcal{O}}_{\bar\mu} =  (\bar{Q} \times \mathcal{O}_{\bar\mu})/H$ over $\bar{Q}/H$ with standard fiber the coadjoint orbit $\mathcal{O}_{\bar\mu}$.  Following \cite{MMORR}, consider the symplectic manifold
\begin{equation} \label{Ap:Eq:Ham-Red_Sympl_mfld}
\left(T^*(\bar{Q}/H) \times_{\mbox{\tiny{$\bar{Q}/H$}}} \tilde{\mathcal{O}}_{\bar\mu}, \Omega_{\mbox{\tiny{$\bar{Q}/H$}}} - \beta_{\bar\mu} \right),
\end{equation}
where $\Omega_{\mbox{\tiny{$\bar{Q}/H$}}}$ is the canonical 2-form on $T^*(\bar{Q}/H)$ and
$\beta_{\bar\mu} \in \Omega^2(\tilde{\mathcal O}_{\bar{\mu}})$ such that
\begin{equation}\label{Ap:DefBeta}
\pi_H^*\beta_{\bar\mu}= d\alpha_{\mbox{\tiny{$\bar{\mathcal A}$}}} + \pi_2^*\omega_{\mathcal{O}_{\bar\mu}},
\end{equation}
where $\pi_H:\bar{Q} \times \mathcal{O}_{\bar\mu} \to \tilde{\mathcal O}_{\bar\mu}$, $\alpha_{\mbox{\tiny{$\bar{\mathcal A}$}}} \in \Omega^1(\bar{Q} \times \mathcal{O}_{\bar\mu})$ is such that $\alpha_{\mbox{\tiny{$\bar{\mathcal A}$}}}(\bar{q},\bar\mu) = \langle \bar\mu, \bar{\mathcal A} \rangle$, where $\bar{\mathcal A}:T\bar{Q} \to \mathfrak{h}$ is a connection, $\pi_2: \bar{Q} \times \mathcal{O}_{\bar\mu} \to \mathcal{O}_{\bar\mu}$ and $\omega_{\mathcal{O}_{\bar\mu}}$ is the canonical 2-form on $\mathcal{O}_{\bar\mu}$.

Next, following \cite{MMORR}, we identify the quotients \eqref{Eq:Proof:equivalence} with \eqref{Ap:Eq:Ham-Red_Sympl_mfld} and consequently we will obtain an identification of \eqref{Eq:BarOmega-SymplecReduc} with a corresponding almost symplectic manifold.

\begin{lemma}[\cite{Marsden2000,MMORR}]\label{L:identifications}
For each $\bar\mu \in \mathfrak{h}^*$, consider the bundles $\bar{J}^{-1}(\bar\mu)/H_{\bar\mu} \to \bar{Q}/H$, $T(\bar{Q}/H) \to \bar{Q}/H$ and $\rho_{\bar\mu}:\bar{Q}/H_{\bar\mu}\to \bar{Q}/H$. Then
\begin{enumerate}
\item[(i)] There is a bundle symplectomorphism $\Upsilon^{\bar\mu}: (T^*(\bar{Q}/H) \times_{\mbox{\tiny{$\bar{Q}/H$}}} \bar{Q}/H_{\bar\mu}, \Omega_{\mbox{\tiny{$\bar{Q}/G$}}} - \beta_\mu)  \to ({\bar J}^{-1}(\bar\mu)/H_{\bar\mu}, \Omega_{\bar\mu}) $ over the identity.
\item[(ii)] There is a global diffeomorphism $\psi_{\bar\mu}:\bar{Q}/H_{\bar\mu} \to \tilde{\mathcal{O}}_{\bar\mu} $ covering the identity on $\bar{Q}/H$.
\end{enumerate}
\end{lemma}
 When the level $\bar\mu$ is clear, we will write $\Upsilon$ instead of  $\Upsilon^{\bar\mu}$.

 Since $\bar{\omega}_{\bar{\mu}}$ satisfies \eqref{Eq:BarOmega-SymplecReduc} we conclude that

 \begin{proposition} \label{C:AlmostReducedMFLD}
  If $\overline{\langle {\mathcal J}, \mathcal{K}_\subW\rangle}$ is basic with respect to the orbit projection $\rho_H:T^*\bar{Q} \to T^*\bar{Q}/H$, then, at each $\bar\mu\in\mathfrak{h}^*$, the almost symplectic quotient $(\bar{J}^{-1}(\bar{\mu})/H_{\bar{\mu}}, \bar{\omega}_{\bar{\mu}})$ given in \eqref{Eq:SymplReduction} is symplectomorphic to
\begin{equation} \label{Eq:Ham-Red_Sympl_mfld}
\left(T^*(\bar{Q}/H) \times_{\mbox{\tiny{$\bar{Q}/H$}}} \tilde{\mathcal{O}}_{\bar\mu}, \Omega_{\mbox{\tiny{$\bar{Q}/H$}}} - \beta_{\bar\mu} - \Upsilon^*[\overline{\langle {\mathcal J}, \mathcal{K}_\subW\rangle}_{\emph\red}]_{\bar\mu}\right).
\end{equation}
 \end{proposition}


%

\subsection{Almost symplectic leaves of the reduced bracket} \label{Ss:2S-Reduction}

Consider the nonholonomic system on a manifold $Q$ given by a Lagrangian $L$ and a nonintegrable distribution $D$ with a $G$-symmetry.  Let $W$ be a $G$-invariant vertical complement of $D$ in $TQ$ satisfying the vertical-symmetry condition \eqref{Eq:VerticalSymmetries}.
Define the Lie subgroup $G_\subW$ of $G$ as in \eqref{Eq:Gw} and the quotient Lie group $H=G/G_\subW$.

\begin{theorem} \label{T:Equivalence}
If $\overline{\langle {\mathcal J}\!,\!{\mathcal K}_\subW \rangle}$ is basic with respect to $\rho_H:T^*\bar{Q} \to T^*\bar{Q}/H$, then \begin{enumerate}                                                                                                                                       \item[$(i)$] The reduced nonholonomic bracket $\{ \cdot , \cdot \}_{\emph\red}$ on $\M/G$ given in \eqref{Eq:RedNHbracket} is twisted Poisson.                                                                                                                                 \item[$(ii)$] The connected components of $(\bar{J}^{-1}(\bar{\mu})/H_{\bar{\mu}}, \bar{\omega}_{\bar{\mu}})$, given in \eqref{Eq:BarOmega-SymplecReduc}, are the almost symplectic leaves of $\{ \cdot , \cdot \}_{\emph\red}$. In other words,  the leaves of $\{ \cdot, \cdot \}_{\emph\red}$ are identified with the (connected components of) \eqref{Eq:Ham-Red_Sympl_mfld}.                                                                                                                                        \end{enumerate}
\end{theorem}

\begin{proof} By Notation \ref{Notation}, item $(i)$ is a consequence of Theorem \ref{T:Reduced-Dyn}.  To show $(ii)$, first note that the nonholonomic bracket $\{ \cdot,  \cdot \}_\nh$ is $G_\subW$-invariant (because it is $G$-invariant). Then, the reduction of $\{ \cdot,  \cdot \}_\nh$ by $G_\subW$ in the sense of \eqref{Eq:RedNHbracket} gives an almost Poisson manifold denoted by $(\M/G_\subW, \{ \cdot,  \cdot \}_{\bar{\Omega}})$. Since $\{ \cdot,  \cdot \}_{\bar{\Omega}}$ is $H$-invariant, there is a reduced bracket $\{ \cdot , \cdot \}$ induced on $(\M/G_\subW)/H$. Using an argument of reduction by stages, it is straightforward to see that $\{ \cdot , \cdot \}_\red$ and $\{ \cdot, \cdot \}$  are diffeomorphic and thus their leaves are symplectomorphic. 

On the other hand, the almost Poisson bracket $\{ \cdot , \cdot \}_{\bar{\Omega}}$ on $\M/G_\subW \simeq T^*\bar{Q}$ is non-degenerate and given by the 2-form $\bar{\Omega}=\Omega_{\mbox{\tiny{$\bar{Q}$}}} - \overline{\langle {\mathcal J}\!,\!{\mathcal K}_\subW \rangle}$ on $T^*\bar{Q}$.
The reduction by $H$ of the canonical symplectic manifold $(T^*\bar{Q}, \Omega_{\mbox{\tiny{$\bar{Q}$}}})$ gives a Poisson bracket whose symplectic leaves are the Marsden-Weinstein quotients \eqref{Eq:Proof:equivalence}. Hence, the reduction of $(T^*\bar{Q}, \Omega_{\mbox{\tiny{$\bar{Q}$}}}- \overline{\langle {\mathcal J}\!,\!{\mathcal K}_\subW \rangle})$ gives the almost Poisson bracket $\{\cdot, \cdot \}$ that has almost symplectic leaves given by $({\bar J}^{-1}(\bar{\mu})/H_{\bar{\mu}} , \bar{\Omega}_{\bar{\mu}}- [\overline{\langle {\mathcal J}\!,\!{\mathcal K}_\subW \rangle}_\red]_{\bar{\mu}})$.  Therefore the $\mu$-leaves of $\{ \cdot, \cdot \}_\red$ are symplectomorphic to connected components of $({\bar J}^{-1}(\bar{\mu})/H_{\bar{\mu}} , \bar{\omega}_{\bar{\mu}})$ for $\bar\mu \in \mathfrak{h}^*$ such that $\varrho^*(\bar\mu) = \mu$.
\end{proof}

Finally, we observe that for the special case of horizontal symmetries (i.e., $\mathfrak{g}_\subS$ is a subalgebra of $\mathfrak{g}$) the 2-form $\langle {\mathcal J}, {\mathcal K}_\subW \rangle$ is basic with respect to $\rho: \M \to \M/G$ using the $G_\subW$-invariance. Then we observe

\begin{corollary}
 If the system has {\it horizontal symmetries} then the reduced nonholonomic bracket $\{\cdot, \cdot \}_{\emph\red}$ is $(d\langle {\mathcal J}, {\mathcal K}_\subW \rangle_\red)$-twisted Poisson with the almost symplectic leaves given in \eqref{Eq:Ham-Red_Sympl_mfld}.
\end{corollary}

\begin{proof}
 If $\mathfrak{g}_\subS$ is a Lie algebra, then then any Lie subalgebra $\mathfrak{g}_\subW$ complementing $\mathfrak{g}$ generates a complement $\W$ of the constraints with the vertical-symmetry condition. Then for $\xi \in \mathfrak{g}_\subS$ and $X\in \Gamma(\C)$ 
 $${\bf i}_{\xi_\subM} \langle {\mathcal J}, {\mathcal K}_\subW \rangle (X) = - \langle {\mathcal J}, {\mathcal A}_\subW  [ \xi_\subM, X] \rangle = 0, 
 $$ 
 using the $G$-invariance of $\C$. Therefore this Corollary is a direct consequece of Theorem \ref{T:Equivalence}.
 \end{proof}

In \cite[Sec.~6]{Oscar} and \cite{Hoch} the authors also propose an almost symplectic reduction at the compressed level. In the present section we explain when it is possible to perform such a reduction, the existence for the extra symmetry $H$ and the connection with the nonholonomic momentum map (and as a consequence the connection with the conserved quantities of the system).
 This viewpoint allows us to unify the results in \cite{paula} and \cite{Hoch,Oscar} and to describe the almost symplectic foliation of  $\{ \cdot, \cdot \}_\red$ in terms of \eqref{Eq:Ham-Red_Sympl_mfld}.

In general, the 2-form $\overline{\langle {\mathcal J}, {\mathcal K}_\subW \rangle}$ might not be basic with respect to $\rho_H:T^*\bar{Q} \to T^*\bar{Q}/H$; this is the case of the Chaplygin ball (discussed in Section \ref{Ss:Ex:Ball}).  In Section \ref{S:Gauge} we explore the possibility of describing the compressed dynamics using a modified 2-form that admits the canonical momentum map $\bar{J}:T^*\bar{Q} \to \mathfrak{h}^*$ as a momentum map. We relate this modification with a {\it dynamical gauge transformation by a 2-form} of the nonholonomic bracket, as was done in \cite{paula,PL2011}. However it is worth mentioning that the idea of modifying $\bar\Omega$ using a 2-form appeared for the first time in \cite{EhlersKoiller}.


\section{Gauge transformations} \label{S:Gauge}

A gauge transformation is a process by which one {\it deforms} a bracket using a 2-form yet retains some of the original geometric properties, such as the characteristic distribution (even though the Hamiltonian vector fields may have changed) \cite{SeveraWeinstein}. In this paper we are interested in considering gauge transformations of the nonholonomic bracket by a 2-form $B$, since the failure of the Jacobi identity can be controlled using the 2-form $B$ \cite{PL2011}.
\bigskip

\noindent {\bf Gauge transformation of a bracket.} \ To begin, recall that a regular almost Poisson manifold $(P, \{\cdot , \cdot \})$ is given by a distribution $F$ on $P$ and a 2-form $\Omega$ on $P$ such that $\Omega|_F$ is non-degenerate:
$$
{\bf i}_{X_f} \Omega |_F = df |_F \mbox { \ if and only if \ } \{ \cdot , f \} = X_f  \qquad \mbox{for \ } f,g, \in C^\infty(P).
$$
Consider now a 2-form $B$ such that $(\Omega - B)|_F$ is still non-degenerate. Then a {\it gauge transformation} of $\{\cdot, \cdot \}$ by a 2-form $B$, \cite{SeveraWeinstein}, gives a new bracket $\{\cdot , \cdot \}_\B$ on $P$ defined by
$$
{\bf i}_{Y_f} (\Omega_F - B) |_F = df |_F \mbox { \ if and only if \ } \{ \cdot , f \}_\B = Y_f  \qquad \mbox{  for \ } f,g, \in C^\infty(P).
$$
In this case, we say that the brackets $\{\cdot, \cdot \}$ and $\{\cdot , \cdot \}_\B$ are gauge-related by the 2-form $B$.

\begin{remark} \label{R:Gauge:Examples}  The gauge transformation by a 2-form $B$ of a nondegenerate bracket $\{\cdot, \cdot \}$ given by the 2-form $\Omega$ defines a new nondegenerate bracket given by $\Omega-B$ (recall that, in this case, $\Omega - B$ is nondegenerate). In this sense, a gauge transformation by $B$ of a 2-form $\Omega$  gives the new 2-form $\Omega-B$. In the same way, if a Poisson bracket $\{\cdot , \cdot \}$  has the symplectic foliation $(\mathcal{P}_\mu, \omega_\mu)$, then a gauge transformation of $\{\cdot , \cdot \}$  by the 2-form $B$ gives a bracket $\{\cdot , \cdot \}_\B$ that has the almost symplectic foliation  $(\mathcal{P}_\mu, \omega_\mu - B_\mu)$ where $B_\mu$ is the restriction of $B$ to leaf $\mathcal{P}_\mu$; if $B$ is closed then the bracket $\{\cdot , \cdot \}_\B$ is still Poisson.
\end{remark}

\begin{remark} \label{R:DefG:General}
A gauge transformation by a 2-form $B$ of a (regular) bracket given by the 2-form $\Omega$ and the distribution $F$ is defined in \cite{SeveraWeinstein} without asking for $(\Omega - B )|_F$ to be nondegenerate since the definition is given on more general structures called Dirac structures.

 On the other hand, a gauge transformation of an almost Poisson manifold $(P, \{\cdot , \cdot \})$ can be defined even in the non-regular case using the bivector field $\pi$ on $P$ associated to the bracket $\{\cdot , \cdot \}$.  In this case, if the endomorphism $\textup{Id} + B^\flat \circ \pi^\sharp$ on $T^*P$ is invertible, then the gauge transformation of $\pi$ by the 2-form $B$ is the new bivector field $\pi_\B$ on $P$ defined by $ \pi_\B^\sharp = \pi^\sharp \circ (\textup{Id} + B^\flat \circ \pi^\sharp)$ (see \cite{SeveraWeinstein}).
\end{remark}

As in \cite{paula,PL2011} we want to consider a gauge transformation of the almost Poisson manifold $(\M, \{ \cdot,\cdot \}_\nh)$ such that the new bracket $\{ \cdot, \cdot \}_\B$ describes the same dynamics. If the 2-form $B$ on $\M$ is such that $(\Omega_\subM - B)|_\C$ is non-degenerate and ${\bf i}_{X_\nh}B =0$, then we say that $B$ induces a {\it dynamical gauge transformation} of $\{ \cdot,\cdot \}_\nh$ \cite{PL2011}.  The new dynamically-gauge related bracket $\{ \cdot , \cdot \}_\B$ describes the same dynamics in the sense that
$$
\{ \cdot, \Ham_\subM \}_\B = X_\nh.
$$

If the nonholonomic system $(\M, \{ \cdot ,\cdot \}_\nh, \Ham_\subM)$ has a $G$-symmetry, then any dynamically-gauge related bracket  $\{ \cdot, \cdot \}_\B$ to $\{ \cdot, \cdot \}_\nh$ is $G$-invariant if the 2-form $B$ is invariant. In this case, there is also a reduced bracket $\{ \cdot, \cdot \}_\red^\B$ on $\M/G$ given by
$$
\{ f,g\}_\red^\B \circ \rho = \{ f\circ \rho,g\circ \rho\}_\B  \qquad \mbox{for } f,g \in C^\infty(\M/G),
$$
that also describes the reduced dynamics: $\{ \cdot, \Ham_\red \}_\red^\B = X_\red.$

Observe that a gauge transformation of the nonholonomic bracket $\{ \cdot , \cdot \}_\nh$ is determined by the values of $B$ on $\C$. Therefore, following \cite{paula}, once we fix a decomposition of \eqref{Eq:C+W}, we can then consider gauge transformations by 2-forms $B$ such that
\begin{equation}\label{Eq:Bsection}
{\bf i}_Z B\equiv 0 \qquad \mbox{for all } \ Z\in \W.
\end{equation}

The interesting consequence of having the freedom to choose such a 2-form $B$ is that the reduced bivector $\{ \cdot , \cdot \}_\red^\B$ might have more desirable properties than $\{ \cdot , \cdot \}_\red$---as we have discussed in \cite{paula,PL2011}---since the Jacobi identity changes by $B$. More precisely, it was shown in \cite{paula} that if the vertical complement $W$ \eqref{Eq:D+W} has the vertical symmetry condition \eqref{Eq:VerticalSymmetries}, then
\begin{equation}\label{Eq:Gauged-Jac}
\{f,\{g,h\}^\B_\red\}^\B_\red \circ \rho + \textup{cyclic} = -d(\langle {\mathcal J}, {\mathcal K}_\subW \rangle +B) (Y_{\rho^*\!f},Y_{\rho^*\!g},Y_{\rho^*\!h}),
\end{equation}
where $f,g,h \in C^\infty(\M/G)$, and $Y_{\rho^*\!f}$ is the Hamiltonian vector field with respect to the nonholonomic bracket $\{ \cdot, \cdot \}^\B_\nh$.
Therefore, if the 3-form $d(\langle {\mathcal J}, {\mathcal K}_\subW \rangle + B)$ defined on $\M$ is basic with respect to the orbit projection $\rho: \M \to \M /G$ then the Jacobi identity of $\{ \cdot , \cdot \}^\B_\red$ becomes
 $$\{f,\{g,h\}^\B_\red\}^\B_\red  + \textup{cyclic} = \Phi (\mathcal{Y}_{f},\mathcal{Y}_{g},\mathcal{Y}_{h}),
 $$
where $\Phi$ is a 3-form on $\M/G$ such that $\rho^*\Phi = - d(\langle {\mathcal J}, {\mathcal K}_\subW \rangle + B)$, and $\mathcal{Y}_{f}$ is the Hamiltonian vector field associated to $\{ \cdot, \cdot \}^\B_\red$; that is,  $\{ \cdot, \cdot \}^\B_\red$ is a twisted Poisson bracket and thus it admits an (almost) symplectic foliation.

We can now restate Theorem \ref{T:Reduced-Dyn} considering the whole family of brackets $\{ \cdot, \cdot \}^\B_\red$.

\noindent {\bf Theorem \ref{T:Reduced-Dyn}'} \cite{paula}
\emph{Suppose that $B$ is a $G$-invariant 2-form on $\M$ satisfying
\eqref{Eq:Bsection} defining a dynamical gauge transformation of $\{ \cdot, \cdot \}_{\emph\nh}$.
If $W$ verifies the vertical-symmetry condition \eqref{Eq:VerticalSymmetries} and  $d(\langle {\mathcal J}, {\mathcal K}_\subW \rangle + B)$ is basic with respect to
the orbit projection $\rho: \M \to \M /G$, then $\{ \cdot, \cdot \}_{\emph\red}^\B$ is a twisted Poisson bracket.
}

\subsection{Reduction in two stages}\label{Ss:Gauge2Setps}

Analogously to what we did in Section \ref{S:2steps}, let us now perform a reduction of $(\M, \{\cdot , \cdot \}_\B)$ in two stages.

Consider the nonholonomic system $(\M, \{\cdot , \cdot \}_\nh, \Ham_\subM)$ with a $G$-symmetry. Suppose that the vertical complement $\W$ of $\C$ satisfies the vertical-symmetry condition \eqref{Eq:VerticalSymmetries}, so that $G_\subW$ is a symmetry of the system.

\begin{lemma}
 Let $\{\cdot , \cdot \}_\B$ be the dynamically-gauge related bracket to $\{\cdot , \cdot \}_{\emph\nh}$ by a $G$-invariant 2-form $B$ that satisfies \eqref{Eq:Bsection}. Then,
 \vspace{-\topsep}
 \begin{enumerate} \itemsep0em
  \item[$(i)$] The nonholonomic system $(\M, \{\cdot, \cdot \}_\B, \Ham_\subM)$ is $G_\subW$-Chaplygin.
  \item[$(ii)$] The compression by $G_\subW$ of  $\{\cdot, \cdot \}_\B$ induces the almost symplectic manifold
\begin{equation} \label{Eq:Gauge-Chap2form}
(T^*\bar{Q}, \Omega_{\mbox{\tiny{$\bar{Q}$}}} - \overline{\langle {\mathcal J}, {\mathcal K}_\subW \rangle} - \bar{B}),
\end{equation}
where $\bar{B}$ is the 2-form on $T^*\bar{Q}$ such that $\rho_\subW^* \bar{B} = B$.
\end{enumerate}
\end{lemma}
\begin{proof}
 $(i)$ Since $G_\subW$ is a normal subgroup of $G$, then  $\{\cdot, \cdot \}_\nh$ and  the 2-form $B$ are also $G_\subW$-invariant because they are already $G$-invariant.  Therefore, the bracket $\{\cdot, \cdot \}_\B$ is $G_\subW$-invariant. Moreover, since the characteristic distribution of $\{\cdot, \cdot \}_\B$ is still $\C$ then we have the splitting \eqref{Eq:C+W} where $\W$ is the vertical space with respect to the $G_\subW$-action; that is why the system is still $G_\subW$-Chaplygin.

To see $(ii)$ observe that $B$ is $G_\subW$-invariant and satisfies \eqref{Eq:Bsection}, hence $B$ is basic with respect to $\rho_\subW:\M \to \M / G_\subW$.
 Since  $\{\cdot, \cdot \}_\nh$ and  $\{\cdot, \cdot \}_\B$ are gauge related by the $G_\subW$-basic 2-form $B$ then, following  \cite[Prop~4.8]{paula}, the reductions of  $\{\cdot, \cdot \}_\nh$ and  $\{\cdot, \cdot \}_\B$ by $G_\subW$ will be also gauge-related by $\bar{B}$, where $\bar{B}$ is the 2-form on $T^*\bar{Q}$ such that $\rho_\subW^* \bar{B} = B$. Moreover, by item $(i)$ the reduction by $G_\subW$ of  $\{\cdot, \cdot \}_\B$ gives a nondegenerate bracket described by a 2-form $\bar\Omega_\B$, and recalling that the $G_\subW$-reduction of $\{\cdot, \cdot \}_\nh$ gives a nondegenerate bracket described by $\bar\Omega$ \eqref{Eq:Chap2form}, from Remark \ref{R:Gauge:Examples} the 2-forms $\bar\Omega_\B$ and $\bar{\Omega}$ are gauge-related by $\bar{B}$ and thus we obtain that $\bar\Omega_\B = \bar{\Omega} - \bar{B}$.
\end{proof}

In other words,  a dynamical gauge transformation of $\{\cdot, \cdot\}_\nh$ by a $G$-invariant 2-form $B$ verifying \eqref{Eq:Bsection} is equivalent to a dynamical gauge transformation of $\bar\Omega$ by the $H$-invariant 2-form $\bar{B}$ on $T^*\bar{Q}$ such that $\rho_\subW^*\bar{B} = B$. 

To clarify the role of $\bar{B}$, recall that the obstruction to performing an almost symplectic reduction of $(T^*\bar{Q}, \bar{\Omega}= \Omega_{\mbox{\tiny{$\bar{Q}$}}} - \overline{\langle \mathcal{J},\mathcal{K}_\subW \rangle})$ as done in Sec.~\ref{Ss:2S-Reduction} is that the canonical momentum map $\bar{J}:T^*\bar{Q} \to \mathfrak{h}^*$ might not be a momentum map for $\bar{\Omega}$ (since $\overline{\langle \mathcal{J},\mathcal{K}_\subW \rangle}$ might not be basic). However, if we find a 2-form $\bar B$ on $T^*\bar{Q}$ such that $\overline{\langle {\mathcal J}, \mathcal{K}_\subW \rangle} + \bar{B}$ is basic with respect to $\rho_H:T^*\bar{Q}\to T^*\bar{Q}/H$ and such that
${\bf i}_{\bar{X}_\nh}\bar B=0$ for $\bar{X}_\nh$ the vector field on $T^*\bar{Q}$ describing the {\it compressed} dynamics \eqref{Eq:ChapReduction}, then the compressed equations of motion become
\begin{equation} \label{Eq:Gauge:Chap}
{\bf i}_{\bar{X}_\nh}( \bar{\Omega} - \bar{B})=d\Ham_{T^*\bar{Q}},
\end{equation}
which is equivalent to  \eqref{Eq:ChapReduction}, but now
$\bar{J}:T^*\bar{Q} \to \mathfrak{h}^*$ is a momentum map for $\Omega_{\mbox{\tiny{$\bar{Q}$}}} -  (\overline{\langle {\mathcal J}, \mathcal{K}_\subW \rangle} + \bar{B})$.

Analogously to Corollary \ref{C:JKbasic=conservationLaws} we obtain:

\begin{corollary} \label{C:B+JK-Basic}
 If there is a 2-form $\bar{B}$ on $T^*\bar{Q}$ such that
${\bf i}_{\bar{X}_{\emph\nh}} \bar{B} =0$ and $\overline{\langle {\mathcal J}, \mathcal{K}_\subW \rangle} + \bar{B}$ is basic with respect to $\rho_H:T^*\bar{Q} \to T^*\bar{Q}/H$ then,
\vspace{-\topsep}
\begin{enumerate} \itemsep0em
\item[$(i)$] $\bar{J}:T^*\bar{Q} \to \mathfrak{h}^*$ is a momentum map for the manifold $(T^*\bar{Q}, \Omega_{\mbox{\tiny{$\bar{Q}$}}} - (\overline{\langle {\mathcal J}, \mathcal{K}_\subW \rangle} + \bar{B}))$.
\item[$(ii)$] The function $g_\xi \in C^\infty(\bar{Q})$ defined in \eqref{Eq:Defg_xi} is conserved by the compressed motion $\bar{X}_{\emph\nh}$.
\end{enumerate}
\end{corollary}
\begin{proof} $(i)$ is straightforward. To show $(ii)$, if $\bar{B}$ is a 2-form on $T^*\bar{Q}$ for which $\bar{J}:T^*\bar{Q} \to \mathfrak{h}^*$ is a momentum map of $\bar{\Omega} -  \bar{B}$, then  $ d  \bar{J}_\xi(\bar{X}_\nh) = (\bar{\Omega}-\bar B)( \xi_{\mbox{\tiny{$T^*\bar{Q}$}}}, \bar{X}_\nh)= d\Ham_{\mbox{\tiny{$T^*\bar{Q}$}}} (\xi_{\mbox{\tiny{$T^*\bar{Q}$}}}) = 0 $ by the $H$-invariance of the Hamiltonian.  Using Prop.~\ref{Prop:g_xi} we obtain that $g_\xi$ is conserved.
\end{proof}

Assume now that there is a 2-form $\bar{B}$ on $T^*\bar{Q}$ such that
${\bf i}_{\bar{X}_\nh} \bar{B} =0$ and $\overline{\langle {\mathcal J}, \mathcal{K}_\subW \rangle} + \bar{B}$ is basic with respect to $\rho_H:T^*\bar{Q} \to T^*\bar{Q}/H$. Then the subsequent steps are equivalent to those carried out in Section \ref{S:2steps}, but now using the 2-form $\overline{\langle {\mathcal J}, \mathcal{K}_\subW \rangle} + \bar{B}$ instead of simply $\overline{\langle {\mathcal J}, {\mathcal K}_\subW \rangle}$. That is, since the canonical momentum map $\bar{J}:T^*\bar{Q}\to \mathfrak{h}^*$ is a momentum map for the 2-form \eqref{Eq:Gauge-Chap2form}, we can perform an (almost) symplectic reduction on the (almost) symplectic manifold \eqref{Eq:Gauge-Chap2form}. For each $\bar{\mu} \in \mathfrak{h}^*$, we obtain an almost symplectic form $\bar{\omega}_{\bar{\mu}}^{\bar\B}$ on $\bar{J}^{-1}(\bar{\mu})/H_{\bar{\mu}}$ such that
\begin{equation} \nonumber
\pi^*_{\bar{\mu}}(\bar{\omega}_{\bar{\mu}}^{\bar{\B}}) = \iota_{\bar{\mu}}^* (\Omega_{\bar{Q}} - (\overline{\langle {\mathcal J}, \mathcal{K}_\subW \rangle} + \bar{B})),
\end{equation}
where $\pi_{\bar{\mu}}:\bar{J}^{-1}(\bar{\mu}) \to \bar{J}^{-1}(\bar{\mu})/H_{\bar{\mu}}$ is the projection to the quotient and $\iota_{\bar{\mu}}: \bar{J}^{-1}(\bar{\mu}) \to \M/G \simeq T^*\bar{Q}$ is the inclusion.
Analogously to Prop.~\ref{P:Reduced2-form}, if we denote by $\mathfrak{B}$ the 2-form on $T^*\bar{Q}/H$ such that $\rho_H^*\mathfrak{B}= \overline{\langle {\mathcal J}, \mathcal{K}_\subW \rangle} + \bar{B}$, then
\begin{equation} \label{Eq:Gauge:Reduced2form}
\bar{\omega}_{\bar\mu}^{\bar\B} = {\Omega}_{\bar\mu} - \mathfrak{B}_{\bar\mu},
\end{equation}
where  ${\Omega}_{\bar\mu}$ is the 2-form defined in \eqref{Eq:Proof:equivalence} and $\mathfrak{B}_{\bar\mu}$ is the pullback of $\mathfrak{B}$ to the submanifold $\bar{J}^{-1}(\bar{\mu})/H_{\bar{\mu}}$.
As in \eqref{Eq:Hamilt-MWreduction}, the reduced dynamics is given, on  each leaf, by
\begin{equation} \label{Eq:Gauge:Hamilt-MWreduction}
{\bf i}_{X_\red^{\bar{\mu}}} ({\Omega}_{\bar\mu} - \mathfrak{B}_{\bar\mu}) = d (\Ham_{\bar\mu} ).
\end{equation}
Moreover, following Prop.~\ref{C:AlmostReducedMFLD}, the almost symplectic manifold $( \bar{J}^{-1}(\bar{\mu})/H_{\bar{\mu}},  \bar{\omega}_{\bar\mu}^{\bar\B})$ is symplectomorphic to
\begin{equation} \label{Eq:Gauge-Ham-Red_Sympl_mfld}
\left(T^*(\bar{Q}/H) \times_{\mbox{\tiny{$\bar{Q}/H$}}} \tilde{\mathcal{O}}_{\bar\mu}, \Omega_{\mbox{\tiny{$\bar{Q}/H$}}} - \beta_{\bar\mu} - \Upsilon^*\mathfrak{B}_{\bar\mu}\right).
\end{equation}
Therefore, following the same steps as the proof of Theorem~\ref{T:Equivalence} we obtain the following Theorem.

%
%
%
%
%
%

\begin{theorem} \label{T:Gauge:equivalence}
Suppose that a 2-form $\bar{B}$ on $T^*\bar{Q}$ satisfies
${\bf i}_{\bar{X}_{\emph\nh}} \bar{B} =0$ and that $\overline{\langle {\mathcal J}, \mathcal{K}_\subW \rangle} + \bar{B}$ is basic with respect to $\rho_H:T^*\bar{Q} \to T^*\bar{Q}/H$. Then the connected components of the almost symplectic manifolds $(\bar{J}^{-1}(\bar{\mu})/H_{\bar{\mu}}, \bar{\omega}_{\bar{\mu}}^{\bar\B})$ given in \eqref{Eq:Gauge:Reduced2form} are the almost symplectic leaves of the twisted Poisson bracket $\{ \cdot, \cdot \}_{\emph\red}^\B$ on $\M/G$ given in Theorem \ref{T:Reduced-Dyn}' for $B$ the 2-form on $\M$ such that $\rho_\subW^*(\bar{B})= B$.
\end{theorem}

Observe that under the hypotheses of Theorem \ref{T:Gauge:equivalence} the almost symplectic leaves associated to the twisted Poisson bracket $\{ \cdot, \cdot \}_{\red}^\B$ are identified with (the connected components of) the almost symplectic manifolds \eqref{Eq:Gauge-Ham-Red_Sympl_mfld}.

\begin{remark} If there is a $\bar{B}$ satisfying the hypotheses of Corollary \ref{C:B+JK-Basic} then, for each $\xi \in \mathfrak{h}^*$, by Prop.~\ref{Prop:JK=0},   $\overline{\langle {\mathcal J}, \mathcal{K}_\subW \rangle} (\bar{X}_{\nh} , \xi_{\mbox{\tiny{$T^*\bar{Q}$}}} ) = 0$. Therefore, the fact that ${\bf i}_{\bar{X}_\nh} \overline{\langle {\mathcal J}, \mathcal{K}_\subW \rangle}$ is basic with respect to $\rho_H:T^*\bar{Q} \to T^*\bar{Q}/H$ is a necessary condition for the existence of $B$.
\end{remark}

\medskip

\noindent {\bf The canonical bracket.} Let $\{\cdot, \cdot\}_{\Lambda}$ be the Poisson bracket on $T^*\bar{Q}/H$ induced by the reduction (in the sense of \eqref{Eq:RedNHbracket}) of the canonical Poisson bracket on $T^*\bar{Q}$ given by $\Omega_{\mbox{\tiny{$\bar{Q}$}}}$.  That is, for $f,g \in C^\infty(T^*\bar{Q}/H)$,
\begin{equation} \label{Eq:LambdaBracket}
\{f, g\}_{\Lambda} \circ \rho_H = \Omega_{\mbox{\tiny{$\bar{Q}$}}} (X_{\rho_H^*f}, X_{\rho_H^*g}),
\end{equation}
where $X_{\rho_H^*f}, X_{\rho_H^*g}$ are the hamiltonian vector fields associated to the 2-form $\Omega_{\mbox{\tiny{$\bar{Q}$}}}$.

The symplectic leaves of the Poisson bracket $\{\cdot, \cdot \}_{\Lambda}$ are (the connected components of) the Marsden-Weinstein quotients $(J^{-1}(\bar\mu) / H_{\bar\mu}, \Omega_{\bar{\mu}})$.

\begin{proposition}
Under the hypothesis of Theorem \ref{T:Gauge:equivalence} the reduced bracket $\{\cdot, \cdot \}_{\emph\red}^\B$ is gauge related with the Poisson bracket $\{\cdot , \cdot\}_{\Lambda}$ by the 2-form $\mathfrak{B}$, where $\mathfrak{B}$ is the 2-form on $T^*\bar{Q}/H$ such that $\rho_H^*\mathfrak{B}= \overline{\langle {\mathcal J}, \mathcal{K}_\subW \rangle} + \bar{B}$.
\end{proposition}

\begin{proof}
 The reduction of the almost Poisson bracket $\{\cdot, \cdot \}_\B$ by the $G_\subW$-action gives the almost symplectic manifold \eqref{Eq:Gauge-Chap2form}.  Observe also that $\Omega_{\mbox{\tiny{$\bar{Q}$}}}$ is gauge related with $\bar{\Omega} - \bar{B}$ by the 2-form $\overline{\langle {\mathcal J}, \mathcal{K}_\subW \rangle} + \bar{B}$ which is basic with respect to $\rho_H:T^*\bar{Q} \to T^*\bar{Q}/H$. Therefore, by \cite[Prop~4.8]{paula}, the reduced brackets $\{\cdot, \cdot\}_{\Lambda}$ and $\{\cdot, \cdot\}_\red^\B$ are gauge related by $\mathfrak{B}$.
\end{proof}

\begin{remark}
 The Poisson bracket $\{\cdot, \cdot\}_{\Lambda}$ is the Poisson bracket $\Lambda$ considered in \cite{paula}. In fact, $\{\cdot, \cdot\}_{\Lambda}$ can be seen as the reduction by the $G$-action of the almost Poisson bracket given by the gauge transformation of $\{\cdot, \cdot\}_\nh$ by the 2-form $\langle {\mathcal J}, \mathcal{K}_\subW \rangle$.
\end{remark}

\section{Nonholonomic Lagrange-Routh equations} \label{S:RouthEquations}

In this section we study a version of the Lagrange-Routh equations for nonholonomic systems that depends on the form $\langle {\mathcal J}, \mathcal{K}_\subW \rangle$. We will see that if $\overline{\langle {\mathcal J}, \mathcal{K}_\subW \rangle}$ is basic with respect to the projection $\rho_H:T^*\bar{Q} \to T^*\bar{Q}/H$, then the Lagrangian equations of motion on each almost symplectic leaf \eqref{Eq:Ham-Red_Sympl_mfld} are a nonholonomic version of the Lagrange-Routh equations. Therefore, the equations of motion of the twisted Poisson bracket $\{\cdot, \cdot \}_\red$ are related to the {\it nonholonomic Lagrange-Routh equations}.

We consider a nonholonomic system on the manifold $Q$ given by a Lagrangian $L$ and a distribution $D$ with a $G$-symmetry.  Let $W$ be a $G$-invariant vertical complement of $D$ on $TQ$ as in \eqref{Eq:D+W} and assume that it satisfies the vertical-symmetry condition \eqref{Eq:VerticalSymmetries}. Thus the nonholonomic system is $G_\subW$-Chaplygin and the compressed equations of motion take place in $T^*\bar{Q}$ or $T\bar{Q}$, depending on whether we work in the Hamiltonian or Lagrangian formalism. The Lagrange-Routh equations are defined on the tangent bundle (i.e., using the Lagrangian formalism) and thus we start by discussing the Lagrangian formalism of nonholonomic systems.

\subsection{Lagrangian formulation and reduction} \label{Ss:LagragianFormulation}

As we have seen in Section \ref{Ss:NH}, the Legendre transformation $(\mathbb{F}L):TQ \to T^*Q$ is an isomorphism linear on the fibers, and thus the moment map $\mathcal{J}:\M \to \mathfrak{g}^*$ induces a function ${\mathcal J}_L: D \to \mathfrak{g}^*$ given by
$$
\mathcal{J}_L := \mathcal{J} \circ (\mathbb{F}L)|_D.
$$
On the other hand, we denote also by ${\mathcal K}_\subW$ the $\mathfrak{g}$-valued 2-form on $D$ given by ${\mathcal K}_\subW = \tau_D^*({\bf K}_{\mbox{\tiny{$W$}}})$, where $\tau_D :D \to Q$ is the canonical projection and ${\bf K}_{\mbox{\tiny{$W$}}}$ is the $\mathfrak{g}$-valued 2-form on $Q$ defined in \eqref{Eq:Def_K}. Therefore, we have a well defined 2-form on $D$ given by $\langle {\mathcal J}_L, {\mathcal K}_\subW \rangle$ (analogous to \eqref{Eq:JK}) that is basic with respect to the orbit projection $\phi_\subW:D \to D/G_\subW \simeq T\bar{Q}$ (as in Lemma \ref{L:JKbasic}).
As before, we denote by $\overline{\langle {\mathcal J}_L, {\mathcal K}_\subW \rangle}$ the 2-form on $T\bar{Q}$ such that $\phi_\subW^*\overline{\langle {\mathcal J}_L, {\mathcal K}_\subW \rangle} = \langle {\mathcal J}_L, {\mathcal K}_\subW \rangle.$

On the manifold $\bar{Q}:= Q/G_\subW$, we have the  compressed Lagrangian $l:T\bar{Q} \to \R$ induced by $L$, i.e., $\phi_\subW^* l = L|_D$, that is also of mechanical type: $l = \frac{1}{2} \bar\kappa - V$ where $\bar\kappa$ is the (partially reduced) kinetic energy and $V$ is the (partially reduced) potential on $\bar{Q}$.  Denote by $(\mathbb{F}l): T\bar{Q} \to T^*\bar{Q}$ the associated Legendre transformation.
 Observe also that $\overline{\langle {\mathcal J}_L, {\mathcal K}_\subW \rangle} = (\mathbb{F} l)^*\overline{\langle {\mathcal J}, {\mathcal K}_\subW \rangle}$. It follows that the compressed equations of motion on $T\bar{Q}$ are given by the integral curves of the vector field $\bar{X}_\nh^l$ on $T\bar{Q}$ such that
\begin{equation}\label{Eq:Lagr-AlmostSymplectic}
{\bf i}_{\bar{X}^l_\nh} (\Omega_{l} - \overline{\langle {\mathcal J}_L, {\mathcal K}_\subW \rangle} ) = dE_l,
\end{equation}
where $\Omega_{l} = (\mathbb{F}l)^* \Omega_{\mbox{\tiny{$\bar{Q}$}}}$ and $E_{l}:T\bar{Q} \to \R$ is the Lagrangian energy associated to the Lagrangian $l:T\bar{Q} \to \R$.
\bigskip

Analogous to what we did in Section \ref{Ss:2S-H_action}, the action of the Lie group $H:= G/G_\subW$ is a symmetry of the nonholonomic system \eqref{Eq:Lagr-AlmostSymplectic}. Then the $H$-action on $\bar{Q}$ defines the canonical momentum map $\bar{J}_l:T\bar{Q} \to \mathfrak{h}^*$  such that 
\begin{equation}\label{Eq:LagMoment}
\bar{J}_l \circ (\mathbb{F}l) = \bar{J}.
\end{equation}

Observe that we have the same ingredients as in Section \ref{Ss:2S-Reduction}, but now on the Lagrangian side: the compressed dynamics is defined on the almost symplectic manifold $(T\bar{Q}, \Omega_l-\overline{\langle \mathcal{J}_L, \mathcal{K}_\subW\rangle} )$ by the (Lagrangian) energy $E_l$. The $H$-action on $\bar{Q}$ induces the canonical momentum map $\bar{J}_l:T\bar{Q}\to \mathfrak{h}^*$ which is a momentum map for $\Omega_l-\overline{\langle \mathcal{J}_L, \mathcal{K}_\subW\rangle}$  if and only if $\overline{\langle {\mathcal J}_L, \mathcal{K}_\subW \rangle}$ is basic with respect to $\phi_H:T\bar{Q} \to T\bar{Q}/H$, (see Corollary \ref{C:JKbasic=conservationLaws} and \eqref{Eq:LagMoment}). 
If $\overline{\langle {\mathcal J}_L, \mathcal{K}_\subW \rangle}$ is basic we denote by $\overline{\langle {\mathcal J_L}, {\mathcal K}_\subW \rangle}_{\red}$ the 2-form on $T\bar{Q}/H$ such that $\phi_H^* \overline{\langle {\mathcal J_L}, {\mathcal K}_\subW \rangle}_{\red} = \overline{\langle {\mathcal J_L}, {\mathcal K}_\subW \rangle}$. Then following Prop. \ref{P:Reduced2-form}, we obtain the reduced almost symplectic manifold 
\begin{equation}\label{Eq:Leaf}
({\bar J}_l^{-1}(\bar{\mu})/H_{\bar{\mu}} , \bar{\omega}^l_{\bar\mu}-[\overline{\langle {\mathcal J_L}, {\mathcal K}_\subW \rangle}_{\red}]_{\bar{\mu}}),
\end{equation}
such that $\iota^*_l \Omega_l = \pi_l^*\omega_{\bar\mu}^l$, where $\iota_l : \bar{J}^{-1}_l({\bar\mu}) \to T\bar{Q}$ is the inclusion, $\pi_l: \bar{J}^{-1}_l({\bar\mu}) \to \bar{J}^{-1}_l({\bar\mu})/H_{\bar\mu} $ is the projection to the quotient and $[\overline{\langle {\mathcal J_L}, {\mathcal K}_\subW \rangle}_{\red}]_{\bar\mu}$ is the restriction of $ \overline{\langle {\mathcal J_L}, {\mathcal K}_\subW \rangle}_{\red}$ to the leaf.

\subsection{Intrinsic version of the Lagrange-Routh equations.}  \label{Ss:TheRouthian}
The classical Routhian is a function on $T\bar{Q}$ that involves the Lagrangian $l:T\bar{Q}\to \R$ and whose definition depends on the conservation of the canonical momentum map.

Observe that $\overline{\langle \mathcal{J}, \mathcal{K}_\subW\rangle}$ is basic with respect to $\rho_H:T^*\bar{Q} \to T^*\bar{Q}/H$ if and only if $\overline{\langle \mathcal{J}_L, \mathcal{K}_\subW\rangle}$ is basic with respect to $\phi_H:T\bar{Q} \to T\bar{Q}/H$. Then, if we assume that $\overline{\langle {\mathcal J}_L, \mathcal{K}_\subW \rangle}$ is basic, by Corollary \ref{C:JKbasic=conservationLaws} and \eqref{Eq:LagMoment} the canonical moment map $\bar{J}_l:T\bar{Q} \to \mathfrak{h}^*$ is a moment map for $\Omega_{l} - \overline{\langle {\mathcal J}_L, {\mathcal K}_\subW \rangle} $ (and thus it is conserved by the compressed motion \eqref{Eq:Lagr-AlmostSymplectic}).
Therefore, we can restrict the system to the level set $\bar{J}_l^{-1}(\bar\mu)$ for $\bar{\mu} \in \mathfrak{h}^*$.

For each $\bar{\mu} \in \mathfrak{h}^*$, the {\em Routhian} $R^{\bar\mu} : T\bar{Q} \to \R$ is defined, at each $v_{\bar{q}} \in T_{\bar{q}}\bar{Q}$ by
\begin{equation}\label{re}
R^{\bar\mu}(v_{\bar{q}}) = l(v_{\bar{q}}) - \langle \bar\mu, \bar{\mathcal{A}}(v_{\bar{q}})\rangle,
\end{equation}
where $\bar{\mathcal{A}} : T\bar{Q} \to \mathfrak{h}$ is the mechanical (principal) connection with respect to the metric given by the kinetic energy $\bar{\kappa}$ of the reduced Lagrangian $l:T\bar{Q}\to \R$ (i.e,. the horizontal space $H =\textup{Ker}\bar{\mathcal A}$ is $\bar{\kappa}$-orthogonal to the vertical space). 
It can be checked that the Euler-Lagrange equations coincide for the Lagrangian $l$ and $R^{\bar{\mu}}$, \cite{Bavo,LC,Marsden2000}.

Analogously as in Lemma \ref{L:identifications}, let us consider the bundle $\bar{J}_l^{-1}(\bar\mu) /H_{\bar\mu} \to \bar{Q}/H$ and denote by \begin{equation} \label{Eq:Tangent-Diffeo}
\Upsilon_l^{\bar\mu}: T(\bar{Q}/H) \times_{\mbox{\tiny{$\bar{Q}/H$}}} \bar{Q}/H_{\bar\mu} \to {\bar J}_l^{-1}(\bar\mu)/H_{\bar\mu}
\end{equation}
the bundle diffeomorphism over the identity. Again, if the level $\bar\mu$ is clear, we write $\Upsilon_l$ to denote $\Upsilon_l^{\bar\mu}$.

Since the Routhian $R^{\bar\mu}$ on $T\bar{Q}$ is $H_{\bar\mu}$-invariant \cite{Marsden2000}, by restriction and reduction there is a well-defined function
$$
\mathfrak{R}^{\bar\mu}: T(\bar{Q}/H) \times_{\mbox{\tiny{$\bar{Q}/H$}}} \bar{Q}/H_{\bar\mu} \to \R
$$
called the {\it reduced Routhian}, such that $R_c^{\bar\mu} = (\Upsilon^{-1}_l \circ \pi_l)^* \mathfrak{R}^{\bar\mu}$, where $R_c^{\bar\mu}$ is the restriction of the Routhian $R^{\bar\mu}$ to $\bar{J}_l^{-1}({\bar\mu})$ and $\pi_l: \bar{J}_l^{-1}(\bar\mu) \to \bar{J}_l^{-1}(\bar\mu)/H_{\bar\mu}$ is the orbit projection.
In order to study the reduced Routhian $\mathfrak{R}^{\bar\mu}$, in \cite{CM1,Marsden2000} was introduced the $H_{\bar\mu}$-invariant function on $\bar{Q}$ given by
\begin{equation} \label{Eq:Def:C_mu}
C_{\bar\mu}(\bar{q}) = \frac{1}{2} \langle \bar\mu , I^{-1}(\bar{q}) \bar\mu\rangle,
\end{equation}
where, for each $\bar{q} \in \bar{Q}$, the map $I(\bar{q}) :\mathfrak{h} \to \mathfrak{h}^*$ is an isomorphism such that $\langle I(\bar{q}) \xi , \eta \rangle = \bar\kappa(\xi_{\bar{Q}}, \eta_{\bar{Q}})$.
On the other hand, the kinetic energy $\bar{\kappa}$ on $\bar{Q}$ induces a metric $|| \cdot ||$ on $\bar{Q}/H$ so that, for $v_x \in T_x(\bar{Q}/H)$, $||v_{x}|| = \bar\kappa((v_x)_{\bar{q}}^h,
(v_x)_{\bar{q}}^h)$ where $(v_x)_{\bar{q}}^h \in T_{\bar{q}}\bar{Q}$ is the horizontal lift of $v_x$ at $\bar{q}$ associated to the mechanical connection $\bar{\mathcal{A}}$. Therefore, if we denote by $\bar{\tau}_{\bar\mu}:\bar{Q}\to \bar{Q}/H_{\bar\mu}$ the orbit projection, then the reduced Routhian can be expressed as
\begin{equation} \label{Eq:RedRouthian}
\mathfrak{R}^{\bar\mu} (x, \dot x, \bar{\tau}_{\bar\mu}(\bar{q})) = \frac{1}{2} ||\dot x||^2 - \mathcal{V}_{\bar\mu}(\bar{\tau}_{\bar\mu}(\bar{q})),
\end{equation}
where $\mathcal{V}_{\bar\mu}\in C^\infty(\bar{Q}/H_{\bar\mu})$ such that $\bar{\tau}_{\bar\mu}^* \mathcal{V}_{\bar\mu} = V(\bar{q}) + C_{\bar\mu}$.
Since the potential $V$ is $H$-invariant, then there is a Lagrangian function $\mathfrak{L}$ defined on $\bar{Q}/H$:
\begin{equation} \label{Def:Red-Lagragian}
\mathfrak{L}(x,\dot x) = \frac{1}{2}||\dot x||^2 - \mathcal{V}(x),
\end{equation}
where $\bar{\tau}^*\mathcal{V} = V$ for $\bar{\tau} :\bar{Q}\to \bar{Q}/H$. However, the function $C_{\bar\mu}$ is only $H_{\bar{\mu}}$-invariant and thus, using Lemma \ref{L:identifications}$(ii)$, we define the function $\mathfrak{C}_{\bar\mu}$ on $\tilde{\mathcal O}_{\bar\mu}$ such that $C_{\bar\mu} = (\psi_{\bar\mu} \circ \bar{\tau}_{\bar\mu})^* \mathfrak{C}_{\bar\mu}$.
Therefore,
\begin{equation} \label{Eq:L-C=R}
\mathfrak{L}  - \mathfrak{C}_{\bar\mu}= \mathfrak{R}^{\bar\mu}.
\end{equation}


%
%
%

\noindent {\bf The intrinsic (free) Lagrange-Routh equations.}
At each $\bar\mu \in \mathfrak{h}^*$, the Lagrangian function $\mathfrak{L}:T(\bar{Q}/H) \to \R$ given in \eqref{Def:Red-Lagragian} induces the isomorphism
\begin{equation} \label{Eq:Reduced-Legendre}
 (\mathbb{F}\mathfrak{L})_{\bar\mu}: T(\bar{Q}/H) \times_{\mbox{\tiny{$\bar{Q}/H$}}} \tilde{\mathcal O}_{\bar\mu} \to T^*(\bar{Q}/H) \times_{\mbox{\tiny{$\bar{Q}/H$}}} \tilde{\mathcal O}_{\bar\mu},
 \end{equation}
given by $(\mathbb{F}\mathfrak{L})_{\bar\mu} = (\mathbb{F}\mathfrak{L}) \times \textup{Id}_{\bar\mu}$, where $(\mathbb{F}\mathfrak{L}) : T(\bar{Q}/H) \to T^*(\bar{Q}/H)$ is the Legendre transformation associated to $\mathfrak{L}$ and $\textup{Id}_{\bar\mu}$ is the identity over $\tilde{\mathcal O}_{\bar\mu}$.
The following result is inspired by the work of \cite{LC} but we decided to develop it here for completeness. 

\begin{lemma} \label{Ap:L:Symplectomorphism} Let $\Omega_{\mbox{\tiny{$\mathfrak{L}$}}}$ be  the Lagrangian 2-form on $T(\bar{Q}/H)$ associated to the Lagrangian $\mathfrak{L}$.  Then
 \begin{enumerate}
\item[$(i)$] $(\mathbb{F}\mathfrak{L})_{{\bar\mu}}$ is a symplectomorphism between $( T(\bar{Q}/H) \times_{\mbox{\tiny{$\bar{Q}/H$}}} \tilde{\mathcal O}_{\bar\mu} , \Omega_{\mbox{\tiny{$\mathfrak{L}$}}} - \beta_{\bar\mu})$ and $( T^*(\bar{Q}/H) \times_{\mbox{\tiny{$\bar{Q}/H$}}} \tilde{\mathcal O}_{\bar\mu}, \Omega_{\mbox{\tiny{$\bar{Q}/H$}}} - \beta_{\bar\mu} )$.
  \item[$(ii)$] The manifold $(T(\bar{Q}/H) \times \tilde{\mathcal O}, \Omega_{\mbox{\tiny{$\mathfrak{L}$}}} - \beta_{\bar\mu})$ is symplectormophic to the Lagrangian Marsden-Weinstein quotients   $(\bar{J}^{-1}_l({\bar\mu})/H_{\bar\mu}, \omega^l_{\bar\mu})$ such that $\iota^*_l \Omega_l = \pi_l^*\omega_{\bar\mu}^l$, where $\iota_l : \bar{J}^{-1}_l({\bar\mu}) \to T\bar{Q}$ is the inclusion and $\pi_l: \bar{J}^{-1}_l({\bar\mu}) \to \bar{J}^{-1}_l({\bar\mu})/H_{\bar\mu} $ is the projection to the quotient (see \eqref{Eq:Leaf}).
 \end{enumerate}

\end{lemma}

\begin{proof}
$(i)$ First, observe that $(\mathbb{F}\mathfrak{L})^* \Omega_{\mbox{\tiny{$\bar{Q}/H$}}}=\Omega_{\mbox{\tiny{$\mathfrak{L}$}}}$ and, since $\beta_{\bar\mu}$ is a 2-form on $\tilde{\mathcal O}$, it is not affected by the Legendre transform.  Therefore,
$$
(\mathbb{F}\mathfrak{L})_{{\bar\mu}} ^*(\Omega_{\mbox{\tiny{$\bar{Q}/H$}}} - \beta_{\bar\mu} ) = \Omega_{\mbox{\tiny{$\mathfrak{L}$}}} - \beta_{\bar\mu}.
$$
$(ii)$ Using the diffeomorphisms $\Upsilon^{-1}=:\Phi: \bar{J}^{-1}({\bar\mu})/H_{\bar\mu} \to T^*(\bar{Q}/H) \times \tilde{\mathcal O}_{\bar\mu}$ and $\Upsilon_l^{-1}=: \Phi_l: \bar{J}_{l}^{-1} ({\bar\mu})/H_{{\bar\mu}} \to T(\bar{Q}/H) \times \tilde{\mathcal O}_{\bar\mu}$ defined in Lemma \ref{L:identifications} and \eqref{Eq:Tangent-Diffeo} and denoting also by $(\mathbb{F}l) : \bar{J}^{-1}_{l}({\bar\mu}) \to \bar{J}^{-1}({\bar\mu})$ the restriction of the Legendre transform, the following diagram commutes \cite{Bavo}:
\begin{equation} \label{Ap:commuting_diagr}
\xymatrix{\bar{J}^{-1}({\bar\mu}) \ar[r]^{\pi} \ar[d]^{(\mathbb{F}l)} & \bar{J}^{-1} ({\bar\mu})/H_{{\bar\mu}} \ar[r]^{\Phi} & T^*(\bar{Q}/H) \times \tilde{\mathcal O}_{\bar\mu} \ar[d]^{(\mathbb{F}\mathfrak{L})_{{\bar\mu}}}  \\
\bar{J}^{-1}_{l}({\bar\mu}) \ar[r]^{\pi_{l}} & \bar{J}_{l}^{-1} ({\bar\mu})/H_{{\bar\mu}} \ar[r]^{\Phi_l}  &
 T(\bar{Q}/H) \times \tilde{\mathcal O}_{\bar\mu}.
}
\end{equation}
Now, we see that $\Phi_l^*(\Omega_{\mbox{\tiny{$\mathfrak{L}$}}} - \beta_{\bar\mu})= \omega^l_{\bar\mu}$, which is equivalent to  $\pi_l^* \Phi_l^*(\Omega_{\mbox{\tiny{$\mathfrak{L}$}}} - \beta_{\bar\mu}) = \iota_l^*\Omega_l$.
Observe that $\pi_l^* \Phi_l^*(\Omega_{\mbox{\tiny{$\mathfrak{L}$}}} - \beta_{\bar\mu}) = ((\mathbb{F}\mathfrak{L}) \circ\Phi_l \circ \pi_l)^*(\Omega_{\mbox{\tiny{$\bar{Q}/H$}}} - \beta_{\bar\mu})$.
The commuting diagram then gives
$((\mathbb{F}\mathfrak{L}) \circ\Phi_l \circ \pi_l)^*(\Omega_{\mbox{\tiny{$\bar{Q}/H$}}} - \beta_{\bar\mu})= (\Phi \circ \pi \circ (\mathbb{F}l))^*(\Omega_{\mbox{\tiny{$\bar{Q}/H$}}} - \beta_{\bar\mu}) = (\mathbb{F}l)^*( \pi^*  \omega_{\bar\mu}) = (\mathbb{F}l))^*(\iota^*\Omega_{\bar{Q}}) = \iota^*_l \Omega_l$.
\end{proof}

Now we study the dynamics on the manifold $( T(\bar{Q}/H) \times_{\mbox{\tiny{$\bar{Q}/H$}}} \tilde{\mathcal O}_{\bar\mu}, \Omega_{\mbox{\tiny{$\mathfrak{L}$}}} - \beta_{\bar\mu})$.
Recall that the Euler-Lagrange equations are given by the integral curves $\bar{X}^l \in \mathfrak{X}(T\bar{Q})$ such that ${\bf i}_{\bar{X}^l}\Omega_l = dE_l$, where $E_l$ is the lagrangian energy. Using Lemma \ref{Ap:L:Symplectomorphism}, we have 

\begin{lemma} \label{L:FreeRouth} The reduced dynamics at the $\bar\mu$-level is described by the integral curves of $\mathcal{X}\in \mathfrak{X}( T(\bar{Q}/H) \times_{\mbox{\tiny{$\bar{Q}/H$}}} \tilde{\mathcal O}_{\bar\mu})$ given by
\begin{equation} \label{Ap:Eq:IntrinsicEL} {\bf i}_{\mathcal X} (\Omega_{\mbox{\tiny{$\mathfrak{L}$}}} - \beta_{\bar\mu}) = d \mathcal{E}_{\bar\mu},
\end{equation}
where $\mathcal{E}_{\bar\mu}\in C^\infty (T(\bar{Q}/H) \times \tilde{\mathcal O}_{\bar\mu}))$ satisfies $\iota^*_l E_l = ((\Upsilon_l)^{-1} \circ \pi_l)^* \mathcal{E}_{\bar\mu}$.  Moreover, $\mathcal{E}_{\bar\mu} = E_{\mathfrak{L}} + \mathfrak{C}_{\bar\mu}$,
where $E_{\mathfrak{L}} \in C^\infty(T(\bar{Q}/H))$ is the Lagrangian energy associated to the Lagrangian $\mathfrak{L}$ and $\mathfrak{C}_{\bar\mu} \in C^\infty(\bar{Q}/H_{\bar\mu})$ is defined in \eqref{Eq:L-C=R}.
\end{lemma}
\begin{proof}
Since the free dynamics is restricted to $\bar{J}_l^{-1}(\bar\mu)$ then the reduced dynamics is described by \eqref{Ap:Eq:IntrinsicEL} as a consequence of Lemma \ref{Ap:L:Symplectomorphism}.
To show that $\mathcal{E}_{\bar\mu} = E_{\mathfrak{L}} + \mathfrak{C}_{\bar\mu}$, we recall that $l$ is of mechanical type and observe that for $v_q \in \bar{J}^{-1}_l({\bar\mu})$ the energy associated to the Lagrangian $l$ is
$$
 E_l(v_{\bar{q}}) = \frac{1}{2}\bar \kappa(P_H (v_{\bar{q}}), P_H (v_{\bar{q}})) + \frac{1}{2} \kappa(P_V(v_{\bar{q}}), P_V(v_{\bar{q}})) +V({\bar{q}}).
 $$
 where $P_H:T\bar{Q} \to H$ and $P_V:T\bar{Q}\to V$ are the horizontal and vertical projections with respect to the mechanical connection $\bar{\mathcal A}$.
Since $\bar\kappa(P_V(v_{\bar{q}}), P_V(v_{\bar{q}})) = \langle {\bar\mu} , \bar{\mathcal A}(v_{\bar{q}}) \rangle =  \langle {\bar\mu} , I^{-1}({\bar{q}}){\bar\mu} \rangle= 2 C_{\bar\mu}$ we obtain that the energy $E_l$ restricted to $\bar{J}^{-1}_l({\bar\mu})$ is $H_{\bar\mu}$-invariant. Moreover, as we saw in Sec.\ref{Ss:TheRouthian}, the kinetic energy metric $\bar\kappa$ and the potential $V$ drop to a metric $|| \cdot || $ and a potential $\mathcal{V}$ on $\bar{Q}/H$. However, the function $C_{\bar\mu} \in C^\infty(\bar{Q})$ drops to a function on $\bar{Q}/H_{\bar\mu}$.
Therefore, we see that the reduced Lagrangian energy $\mathcal{E}_{\bar\mu}$ is a function on $T(\bar{Q}/H) \times \tilde{\mathcal O}_{\bar\mu}$ given by
$$
\mathcal{E}_{\bar\mu} = \frac{1}{2}||\cdot ||^2 + \mathcal{V} + \mathfrak{C}_{\bar\mu} = E_{\mathfrak{L}} + \mathfrak{C}_{\bar\mu}.
$$
The energy $\mathcal{E}_{\bar\mu}$ satisfies $\iota^*_l E_l = (\Phi_l \circ \pi_l)^* \mathcal{E}_{\bar\mu}$, and thus the Lemma is proven. \ \ \ \ \
\end{proof}

\begin{remark} The reduced Hamiltonian $h_{\bar\mu} \in C^\infty(T^*(\bar{Q}/H) \times \tilde{\mathcal O})$ verifies that $[(\mathbb{F}\mathfrak{L})_{\bar\mu} ]^*h_{\bar\mu} = \mathcal{E}_{\bar\mu}$.
\end{remark}

We will call the equations \eqref{Ap:Eq:IntrinsicEL} the {\it intrinsic  Lagrange-Routh equations}.

Finally we can observe that

\begin{proposition} The equations of motion on each leaf of the Poisson bracket $\{ \cdot, \cdot \}_{\Lambda}$ given in  \eqref{Eq:LambdaBracket} are related to the (intrinsic) Lagrange-Routh equations \eqref{Ap:Eq:IntrinsicEL} via the isomorphism $\Upsilon^{\bar\mu} \circ (\mathbb{F}\mathfrak{L})_{\bar\mu}$ defined in \eqref{Eq:Reduced-Legendre}.
\end{proposition}

\begin{proof} By Lemma \ref{L:identifications}(i) and Lemma \ref{Ap:L:Symplectomorphism}(i) we have that the leaves of $\{\cdot, \cdot \}_\Lambda$ given by the connected components of $({\bar J}^{-1}(\bar{\mu})/H_{\bar{\mu}} , {\Omega}_{\bar{\mu}})$ (eq.\eqref{Eq:Proof:equivalence}) are symplectomorphic to $(T(\bar{Q}/H) \times_{\mbox{\tiny{$\bar{Q}/H$}}} \tilde{\mathcal{O}}_{\bar\mu}, \Omega_{\mbox{\tiny{$\mathfrak{L}$}}} -\beta_{\bar\mu})$. 
\end{proof}

\noindent {\bf Intrinsic nonholonomic Lagrange-Routh equations.}
From Section \ref{Ss:2S-H_action} and by \eqref{Eq:Lagr-AlmostSymplectic} we see that the Lie group $H= G/G_\subW$ is a symmetry of the nonholonomic system given on the almost symplectic manifold $(T\bar Q, \Omega_{l} - \overline{\langle {\mathcal J}_L, {\mathcal K}_\subW \rangle} )$ with Lagrangian energy $E_l$.
  In our case, assuming that $\overline{\langle {\mathcal J}_L, \mathcal{K}_\subW \rangle}$ is basic with respect to the orbit projection $\phi_H:T\bar{Q} \to T\bar{Q}/H$ the Lagrange-Routh equations are modified by the presence of the $H$-invariant 2-form $\overline{\langle\mathcal{J}_L, \mathcal{K}_\subW\rangle}_\red$ (recall that $\overline{\langle {\mathcal J}_L, \mathcal{K}_\subW \rangle} = \phi_H^* \overline{\langle {\mathcal J}_L, \mathcal{K}_\subW \rangle}_\red$).

\begin{lemma}\label{L:J_LK&JK} For each $\bar\mu \in \mathfrak{h}^*$,
 $$
 (\mathbb{F}\mathfrak{L})_{\!\bar\mu}^* \left[ \Upsilon^*[\overline{\langle {\mathcal J}, \mathcal{K}_\subW\rangle}_{\emph\red}]_{\bar\mu} \right]=  \Upsilon_l^* [\overline{\langle {\mathcal J}_L, \mathcal{K}_\subW\rangle}_{\emph\red}]_{\bar\mu}.
 $$
 where $\Upsilon$ and $\Upsilon_l$ are the diffeomorphisms defined in Lemma \ref{L:identifications} and \eqref{Eq:Tangent-Diffeo} respectively.
 \end{lemma}

\begin{proof}
This Lemma is a consequence of the commuting diagrams in \eqref{Ap:commuting_diagr} and Lemma \ref{L:JK-Coord}.
\end{proof}


Using the isomorphism  $(\mathbb{F}\mathfrak{L})_{\bar\mu}$ given in \eqref{Eq:Reduced-Legendre} and Lemmas \ref{Ap:L:Symplectomorphism}(i) and \ref{L:J_LK&JK}, we obtain that the almost symplectic manifold \eqref{Eq:Ham-Red_Sympl_mfld} is symplectomorphic to
\begin{equation} \label{Eq:Lag-Red_Sympl_mfld}
\left(T(\bar{Q}/H) \times_{\mbox{\tiny{$\bar{Q}/H$}}} \tilde{\mathcal O}_{\bar\mu} , \Omega_{\mbox{\tiny{$\mathfrak{L}$}}} -\beta_{\bar\mu} - \Upsilon_l^*[\overline{\langle {\mathcal J}_L, \mathcal{K}_\subW\rangle}_\red]_{\bar\mu}\right).
\end{equation}
where  $\Omega_{\mbox{\tiny{$\mathfrak{L}$}}}$ is the Lagrangian 2-form on $T(\bar{Q}/H)$ associated to $\mathfrak{L}$.
Recall that the reduced nonholonomic dynamics is described by the reduced nonholonomic bracket $\{\cdot, \cdot \}_\red$  given in Theorem \ref{T:Reduced-Dyn}. Then we conclude that

\begin{theorem} \label{T:RouthIntrinsic}
If the 2-form $\overline{\langle {\mathcal J}_L, \mathcal{K}_\subW\rangle}$ is basic with respect to $\phi_H:T\bar{Q} \to (T\bar{Q})/H$ then 
\begin{enumerate}
 \item[$(i)$] the leaves of the twisted Poisson bracket $\{\cdot, \cdot \}_{\emph\red}$ are identified to (the connected components of) the almost symplectic manifolds \eqref{Eq:Lag-Red_Sympl_mfld} via the isomorphism $\Upsilon^{\bar\mu} \circ (\mathbb{F}\mathfrak{L})_{\bar\mu}$ defined in \eqref{Eq:Reduced-Legendre}.
\item[$(ii)$] The equations of motion on each leaf of the reduced nonholonomic bracket $\{ \cdot, \cdot \}_{\emph\red}$ are  $\Upsilon^{\bar\mu} \circ (\mathbb{F}\mathfrak{L})_{\bar\mu}$-related to the nonholonomic Lagrange-Routh equations
\begin{equation} \label{Eq:NH-Routh} 
{\bf i}_{{\mathcal X}_{\emph\nh}} (\Omega_{\mbox{\tiny{$\mathfrak{L}$}}} - \beta_{\bar\mu}- \Upsilon_l^*[\overline{\langle {\mathcal J}_L, \mathcal{K}_\subW\rangle}_{\emph\red}]_{\bar\mu}) = d \mathcal{E}_{\bar\mu},
\end{equation} for $\mathcal{X}_{\emph\nh} \in \mathfrak{X}( T(\bar{Q}/H) \times_{\mbox{\tiny{$\bar{Q}/H$}}} \tilde{\mathcal O}_{\bar\mu})$ with $\mathcal{E}_{\bar\mu} = E_{\mathfrak{L}}+\mathfrak{C}_{\bar\mu}$, where $E_{\mathfrak{L}}$ is the energy on $T(\bar{Q}/H)$ induced by the Lagrangian $\mathfrak{L}$.
\end{enumerate}
\end{theorem}

\begin{proof} 
$(i)$ The connected components of $(\bar{J}^{-1}(\bar{\mu})/H_{\bar{\mu}}, \bar{\omega}_{\bar{\mu}})$ given in \eqref{Eq:Proof:equivalence} are the leaves of $\{\cdot, \cdot \}_\red$ (Theorem \ref{T:Equivalence}) and on the other hand $(\bar{J}^{-1}(\bar{\mu})/H_{\bar{\mu}}, \bar{\omega}_{\bar{\mu}})$ is symplectomorphic to \eqref{Eq:Lag-Red_Sympl_mfld}.

$(ii)$ It is a direct consequence of Lemmas \ref{L:FreeRouth}, \ref{L:J_LK&JK} and \eqref{Eq:Leaf}.
\end{proof}

\begin{remark}
It is worth mentioning that in \cite[Sec.~4]{LC} the authors also discuss Routh reduction for non-conservative systems. 
\end{remark}

\noindent {\bf Intrinsic formulation after a gauge transformation.}
If $\overline{\langle \mathcal{J}, \mathcal{K}_\subW\rangle}$ is not basic with respect to $\rho_H:T^*\bar{Q} \to T^*\bar{Q}/H$ then we consider a dynamical gauge transformation of the 2-form $\bar{\Omega}$ by a 2-form $\bar{B}$ (i.e., ${\bf i}_{\bar{X}_\nh} \bar{B}=0$ and  $\bar\Omega + \bar{B}$ nondegenerate) such that $\overline{\langle \mathcal{J}, \mathcal{K}_\subW\rangle} + \bar{B}$ is basic. Recall that we denote by $\mathfrak{B}$ the 2-form on $T^*\bar{Q}/H$ such that $\rho_H^* \mathfrak{B} = \overline{\langle \mathcal{J}, \mathcal{K}_\subW\rangle} + \bar{B}$.
Then the almost symplectic manifolds $(\bar{J}^{-1}(\bar{\mu})/H_{\bar{\mu}}, \bar{\omega}_{\bar{\mu}}^{\bar\B})$ given in \eqref{Eq:Gauge:Reduced2form} are symplectomorphic to
\begin{equation} \label{Eq:Lag-Red_Sympl_mfld-Gauge}
\left(T(\bar{Q}/H) \times_{\mbox{\tiny{$\bar{Q}/H$}}} \tilde{\mathcal O}_{\bar\mu} , \Omega_{\mbox{\tiny{$\mathfrak{L}$}}} -\beta_{\bar\mu} - (\Upsilon_l \circ (\mathbb{F}\mathfrak{L})_{\bar\mu}) ^*\mathfrak{B}_{\bar\mu}\right),
\end{equation}
where $\mathfrak{B}_{\bar\mu}$ is the pull back of $\mathfrak{B}$ to $\bar{J}^{-1}(\bar{\mu})/H_{\bar{\mu}}$.
Then the {\it intrinsic nonholonomic Lagrange-Routh} equations are the equations of motion defined on the almost symplectic manifold \eqref{Eq:Lag-Red_Sympl_mfld-Gauge} with energy $\mathcal{E}_{\bar\mu}$ given in Lemma \ref{L:FreeRouth}. Finally we see that the almost symplectic leaves of the reduced bracket $\{\cdot , \cdot \}_\red^\B$ are $\Upsilon^{\bar\mu} \circ (\mathbb{F}\mathfrak{L})_{\bar\mu}$-related to \eqref{Eq:Lag-Red_Sympl_mfld-Gauge}.

\subsection{Local version of the Lagrange-Routh equations}

\noindent {\bf Expression in coordinates.}  \ Following \cite{Co}, let us consider local coordinates $(r^\a, s^A)$ on the manifold $Q$ adapted to a local trivialization $U \times G_\subW$ of the principal bundle $Q \to Q/G_\subW$ for $\a = 1, ..., m=n- \textup{dim}\, G_\subW$ and $A=1,..., \textup{dim}\, G_\subW$.
Let $\{ e_A\}$ be a basis of $\mathfrak{g}_\subW$; then (by left trivialization $TG_\subW \simeq G_\subW \times \mathfrak{g}_\subW$) an element $v \in T_{(r,s)}(U \times G_\subW)$ can be represented by $v_{(r,s)} = (\dot r, \eta)$, where $\eta = \eta^A e_A \in \mathfrak{g}_\subW$.

Recall that $TQ=D \oplus W$ and $W$ is also the vertical space associated to the $G_\subW$-action.  Let us first denote by ${\mathcal A}_\subW:TQ \to \mathfrak{g}_\subW$ the principal connection for which the horizontal space is the distribution $D$.
If we denote by $Z_A$ the infinitesimal generator of the element $e_A \in \mathfrak{g}_\subW$, i.e., $Z_A:=(e_A)_Q$, then ${\mathcal A}_\subW(\frac{\partial}{\partial r^\a}) = {\mathcal A}_\a^A(r,s)e_A$ and ${\mathcal A}(Z_A) = e_A$, and so
\begin{equation}\label{Eq:D+W-Coord}
 D = \textup{span} \{ X_\a:= \frac{\partial}{\partial r^\a} - {\mathcal A}_\a^A (r,s) Z_A \} \qquad \mbox{and} \qquad W = \textup{span} \{ Z_A \}.
 \end{equation}
If $(\dot r^\a, v^A)$ are the coordinates associated to local the basis $\{X_\a,Z_A\}$ then the orbit projection $\phi_\subW: D \to D/G_\subW \simeq T\bar{Q}$ is given by $\phi_\subW(r^\a, s^A, \dot r^\a ) = (r^\a, \dot r^\a)$, and thus the reduced 
equations of motion on $T\bar{Q}$ given in \eqref{Eq:Lagr-AlmostSymplectic} are expressed by:
\begin{equation}\label{rcdyn}
\frac{\partial l}{\partial r^{\a}} -  \frac{d}{dt}\frac{\partial l}{\partial \dot{r}^{\a}}=\left(\frac{\partial L}{\partial v^B}\right)^{\!*}\!K^{B}_{\alpha\beta}\dot{r}^{\beta},
\end{equation}
where the star indicates that we have substituted the constraints $v^B = 0$ after differentiation (or equivalently, we can write $\left(\frac{\partial L}{\partial \eta^B}\right)^{\!*}$ and substitute $\eta^A = {\mathcal A}_\a^A \dot r^\a$), and where
\begin{equation} \label{Eq:Coef_K-C}
K^{B}_{\alpha\beta}= C_{\a\beta}^B -C_{\beta\a}^B,  \qquad \mbox{for} \quad C_{\a\beta}^B =\frac{\partial {\mathcal A}^{B}_{\a}}{\partial r^{\beta}}+{\mathcal A}^{A}_{\a}\frac{\partial {\mathcal A}^{B}_{\beta}}{\partial s^A}
\end{equation}
are the components of the $\W$-curvature (or equivalently the curvature associated to the principal connection $\mathcal{A}_\subW$) in the basis $\{X_\a, Z_A\}$, i.e.,  $\mathcal{K}_\subW(X_\a,X_\beta)= d{\mathcal A}_\subW(X_\a,X_\beta)=K_{\a\beta}^Ae_A$ (see \cite{Bloch:Book} and \cite{Co}).

\begin{remark} The functions $\left(\frac{\partial L}{\partial v^B}\right)^{\!*}\!K^{B}_{\alpha\beta}$ are well defined on $T\bar Q$ because they are the coefficients of the 2-form $\overline{\langle {\mathcal J}_L, \mathcal{K}_\subW \rangle}$.
\end{remark}



Now, we will choose coordinates on $\bar{Q}$ and on $T\bar{Q}$ adapted to the $H$-action. That is, we endow the manifold $\bar{Q}$ with local coordinates $(x^i, y^a)$ for $i=1,..., m-\textup{dim}\, H$ and $a=1,..., \textup{dim}\,H$, associated to the local trivialization $\bar U \times H$ of the principal bundle $\bar{\tau}:\bar Q \to \bar Q /H$. Let $\{ \bar{e}_a\}$ be a basis of the Lie algebra $\mathfrak{h}$; then (by left trivialization $TH \simeq H \times \mathfrak{h}$) an element $v \in T_{(x,y)}(\bar{U} \times H)$ is represented by $v_{(x,y)} = (\dot x, w)$, where $w = w^a \bar{e}_a \in \mathfrak{h}$. A principal connection $\bar{\mathcal{A}}:T\bar{Q} \to \mathfrak{h}$ induces a (local) basis on $T\bar{Q}$ given by
\begin{equation} \label{Eq:BasisTBarQ}
\mathcal{B}= \{ \bar{X}_i := \frac{\partial}{\partial x^i} - \bar{\mathcal{A}}_i^a \bar{Z}_a \ , \ \bar{Z}_a \},
 \end{equation}
 where $\bar{Z}_a$ is the infinitesimal generator of the element $\bar{e}_a \in \mathfrak{h}$, i.e., $\bar{Z}_a:=(\bar{e}_a)_{\bar{Q}}$ and $\bar{\mathcal{A}}_i^a = \bar{\mathcal{A}}_i^a(x^i, y^a)$ such that $\bar{\mathcal{A}}(\frac{\partial}{\partial x^i}) = \bar{\mathcal{A}}_i^a \bar{e}_a$. We denote by $(\dot x^i, \xi^a)$ the (local) coordinates on $T\bar{Q}$ associated to the basis \eqref{Eq:BasisTBarQ}. Therefore in local coordinates, the  orbit projection $\phi_H:T\bar{Q} \to T\bar{Q}/H$ is given by $\phi_H(x^i, y^a; \dot x^i, \xi^a) = (x^i; \dot x^i, \xi^a)$.

Following \cite{Marsden2000} we consider $\bar{\mathcal{A}}$ to be the mechanical connection with respect to the metric given by the kinetic energy $\bar{\kappa}$ of the reduced Lagrangian $l: T\bar{Q} \to \R$. In this case, the Lagrangian written in the local coordinates $(x^i, y^a, \dot x^i, \xi^a)$ is
\begin{equation} \label{Eq:RedLagrangian}
l(x,y,\dot x, \xi) = \frac{1}{2}\bar\kappa_{ij}(x,y)\dot x^i \dot x^j + \frac{1}{2}\bar\kappa_{ab}(x,y) \xi^a \xi^b - V(x,y).
\end{equation}

The (Lagrangian) momentum map $\bar{J}_l:T\bar{Q} \to \mathfrak{h}^*$ given in \eqref{Eq:LagMoment} is expressed in local coordinates by
$$
\langle \bar{J}_l(v_{\bar{q}}), \bar{e}_a \rangle = \frac{\partial l}{\partial \xi^a} (v_{\bar{q}}),
$$ 
for each element $\bar{e}_a$ of the basis of $\mathfrak{h}$ and $v_{\bar{q}} \in T_{\bar{q}}\bar{Q}$.
  In other words, the function $g_{\eta}$ on $T^*\bar{Q}$ defined in \eqref{Eq:Defg_xi} verifies $\frac{\partial l}{\partial \xi^a} \eta^a  = g_{\eta} \circ (\mathbb{F}l),$  where $\eta^a$ are the coordinates of $\eta \in \mathfrak{h}$ in the basis $\bar{e}_a$.

\noindent {\bf Routh reduction for the free system.}\ Since $\bar{J}_l:T\bar{Q} \to \mathfrak{h}^*$ is the canonical momentum map for $\Omega_l$ we can restrict the system to the level set $\bar{\mu} \in \mathfrak{h}^*$ given by
\begin{equation}\label{pAeqn}
\mu_{a}=\frac{\partial l}{\partial \xi^{a}}(v_{\bar{q}}),
\end{equation}
where $\bar\mu = \mu_a \bar{e}^a$.
For each $\bar{\mu} \in \mathfrak{h}^*$, the {\em Routhian} $R^{\bar\mu} : T\bar{Q} \to \R$ given in \eqref{re} can be seen as a partial Legendre transformation in the variables $\xi^{a}$:
$$
R^{\bar\mu}(x,y,\dot{x})=\left.\left[l(x,y,\dot{x},\xi)-\mu_{a}\xi^{a}\right]\right|_{\xi^a=\xi^a(x,y)},
$$
where $\xi(x,y)$ is the unique solution for $\xi$ in \eqref{pAeqn} for a particular $\mu_a$. It can be checked that the Routhian does not depend on $\xi^a$ \cite{Marsden2000,Bavo}.

\begin{remark}
 The functions $\xi^a(x,y)$ do not depend on the velocities $(\dot x^i)$ since we are considering coordinates adapted to the mechanical connection $\bar{\mathcal A}$, in this case, $\mu_a = \bar\kappa_{ab}(x,y) \xi^b$.
\end{remark}

The next step is to express the intrinsic Lagrange-Routh equations \eqref{Ap:Eq:IntrinsicEL} in local coordinates.  
To do that, we set a basis of vector fields on $T(\bar{Q}/H) \times_{\mbox{\tiny{$\bar{Q}/H$}}} \tilde{\mathcal O}_{\bar\mu}$ as follows.  First consider our local coordinates $(x^i, y^a; \dot x^i, \xi^a)$ on $T\bar{Q}$ and observe that $\bar{J}_l^{-1}(\bar{\mu}) \subset T\bar{Q}$ is represented by $(x^i, y^a; \dot x^i)$ (and  $\bar{J}_l^{-1}(\bar{\mu}) \simeq \bar{Q} \times_{\bar{Q}/H} T(\bar{Q}/H)$, see \cite{CM1}). Following \cite{CM1}, we set the (local) basis of $\mathfrak{X}(\bar{J}_l^{-1}(\bar{\mu}))$ given by $\{ \bar{X}^\cc_i, \bar{X}^\vv_i, \bar{Z}_a^\cc\}$ where $\bar{X}^\cc_i$, $\bar{Z}_a^\cc$ and $\bar{X}_i^\vv$ are the complete and vertical lifts of $\bar{X}_i$ and $\bar{Z}_a$ given in \eqref{Eq:BasisTBarQ} to $T(\bar{J}_l^{-1}(\bar{\mu}))$. In local coordinates we have
$$
\bar{X}_i^\cc = \frac{\partial}{\partial x^i} - \bar{\mathcal{A}}_i^a \bar{Z}_a \ , \qquad \bar{Z}^\cc_a= B_a^b\frac{\partial}{\partial y^b} \ , \qquad \bar{X}^\vv_i= \frac{\partial}{\partial \dot x^i},
$$
for $B_a^b$ functions on $\bar{Q}$. 
Here we have to be careful because the coordinate expression of $\bar{X}_i$ and $\bar{X}^\cc_i$ coincide but they mean different vector fields: $\bar{X}^\cc_i \in \mathfrak{X}(\bar{J}_l^{-1}(\bar{\mu}))$ such that $T\bar\tau_{\bar\mu} (\bar{X}^\cc_i ) = \bar{X}_i \in \mathfrak{X}(\bar{Q})$ for $\bar{\tau}_{\bar\mu} : \bar{J}_l^{-1}(\bar{\mu})\to \bar{Q}$. 

Observe that $\bar{X}_i^\cc$ and $\bar{X}_i^\vv$ are $H_{\bar\mu}$-invariant vector fields and so they define a set of independent vector fields $\tilde{X}_i^\cc$ and $\tilde{X}_i^\vv$ on $\bar{J}_l^{-1}(\bar{\mu})/H_{\bar\mu}$. However, $\bar{Z}_a^\cc$ are not necessarily $H_{\bar\mu}$-invariant. Using a local trivialization of $\bar{Q}/H_{\bar\mu}$ associated to the principal bundle $\bar{Q} / H_{\bar\mu} \to \bar{Q}/H$ we obtain local coordinates $(x^i, y^{\tilde a})$ on $\bar{Q}/H_{\bar\mu}$ where $y^{\tilde a}$ denote the coordinates on the fiber $H/H_{\bar\mu}$.  Again, considering the principal bundle $\bar{Q}\to \bar{Q}/H_{\bar\mu}$ and a local trivialization, we may consider local coordinates $(x^i, y^{\tilde a}, y^{\bar{a}})$ for $\bar{Q}$. In this case, $(y^{\tilde a}, y^{\bar{a}})$ are local coordinates on the fiber of $\bar{Q} \to \bar{Q}/H$ and also we can see them as coordinates on $H$ associated to a local trivialization of $H \to H/H_{\bar\mu}$. Therefore, the vertical space $V_{\bar\mu}$ with respect to $\pi_{\bar\mu} : \bar{Q} \to \bar{Q}/H_{\bar\mu}$ is given by $ V_{\bar\mu} = \{\tilde{Z}_{\bar a} := \bar{B}_{\bar a}^{\bar b} \frac{\partial}{\partial y^{\bar b}} \} $. Then, if we denote by $\mathfrak{h}_{\bar\mu}$ the Lie algebra associated to $H_{\bar\mu}$, we define the principal connection ${\mathcal A}_{\bar\mu}:T\bar{Q}\to \mathfrak{h}_{\bar\mu}$  such that the $G_{\bar\mu}$-invariant horizontal space $H_{\bar\mu}$ is given by 
\begin{equation}\label{Eq:Connection-Hmu}
H_\mu := \{ \bar{X}_i\ ,\  Y_{\tilde a} := \bar{B}_{\tilde a}^{\tilde b} \frac{\partial}{\partial y^{\tilde b} }+ \bar{B}_{\tilde a}^{\bar b} \frac{\partial}{\partial y^{\bar{b}}} \}.
\end{equation}
Then, by construction, $\tilde{Z}_{\tilde a} := \bar{B}_{\tilde a}^{\tilde b} \frac{\partial}{\partial y^{\tilde b} } \in \mathfrak{X}(\bar{Q}/H_{\bar\mu}) =  T\pi_{\bar\mu} ( Y_{\tilde a}  ) $.
Finally we have a local basis of vector fields on $T(\bar{Q}/H) \times_{\mbox{\tiny{$\bar{Q}/H$}}} \tilde{\mathcal O}_{\bar\mu}$ given by 
$$
\{ \tilde{X}_ i^\cc =  \frac{\partial}{\partial x^i} - \bar{\mathcal{A}}_i^{\tilde a} \bar{Z}_{\tilde a}, \  \tilde{X}^\vv_i= \frac{\partial}{\partial \dot x^i}, \  \tilde{Z}_{\tilde a} = {\bar B}_{\tilde a}^{\tilde b} \frac{\partial}{\partial y^{\tilde b}} \}.
$$
Following again \cite{CM1}, a vector field $\mathcal{X}(x^i, \dot x^i, y^{\tilde a}) = a_i \tilde{X}_i^\cc + b_i \tilde{X}_i^\vv+ c_{\tilde a} \tilde{Z}_{\tilde a}$ on $T(\bar{Q}/H) \times_{\mbox{\tiny{$\bar{Q}/H$}}} \tilde{\mathcal O}_{\bar\mu}$ is a {\it second order equation} if $\dot x^i = a_i$, $\ddot x^i = b_i$ and 
$$
\dot y^{\tilde a} = c_{\tilde a} {\bar B}_{\tilde a}^{\tilde b} - \dot x^i \bar{\mathcal A}_i^{\tilde c} \bar{B}_{\tilde c}^{\tilde b} .
$$

%
%
%
%
%
%
%
%

\begin{proposition} \label{Ap:Prop}
The intrinsic Lagrange-Routh equations \eqref{Ap:Eq:IntrinsicEL} are expressed in local coordinates $(x^i, \dot x^i, y^{\tilde a})$ as
  \begin{subequations}
\begin{eqnarray}
\dot y^{\tilde a}& =& c_{\tilde a} {\bar B}_{\tilde a}^{\tilde b} - \dot x^i \bar{\mathcal A}_i^{\tilde c} \bar{B}_{\tilde c}^{\tilde b}  \label{Eq:Classic:Splitted-RedLag1} \\
\frac{\partial {\mathfrak{L}}}{\partial x^{i}} -  \frac{d}{dt}\frac{\partial \mathfrak{L}}{\partial \dot{x}^{i}} & = &   \mu_b F^b_{ij} \, \dot x^j + (d\mathfrak{C}_{\bar\mu})_i \label{Eq:Classic:Splitted-RedLag2}
\end{eqnarray}
\end{subequations}
where $F^b_{ij}$ are the components of the curvature ${\mathcal F}$  associated to the mechanical connection $\bar{\mathcal A}:T\bar{Q}\to\mathfrak{h}$ (i.e., ${\mathcal F}(\bar{X}_i, \bar{X}_j) = F^b_{ij}\bar{e}_b$),  and $(d\mathfrak{C}_{\bar\mu})_i = d \mathfrak{C}_{\bar\mu}(\tilde X_i^\cc)$. 
\end{proposition}

 \begin{proof} 
 Eq. \eqref{Eq:Classic:Splitted-RedLag1} represents the fact that the vector field $\mathcal{X}$ is a second order equation.

 To prove \eqref{Eq:Classic:Splitted-RedLag2}, we call ${\bf p}_1$ and ${\bf p}_2$ the projection to the first and second factor respectively of $T(\bar{Q}/H) \times_{\mbox{\tiny{$\bar{Q}/H$}}} \tilde{\mathcal O}_{\bar\mu}$ and observe that $\Omega_{\bar\mu} = {\bf p}_1^* \Omega_{\mbox{\tiny{$\mathfrak{L}$}}} - {\bf p}_2^*\beta_{\bar\mu}$ and $\mathcal{E}_{\bar\mu} = {\bf p}_1^*E_{\mathfrak{L}} - {\bf p}_2^* \mathfrak{C}_{\bar\mu}$.
 
Consider the element $\tilde{X}_i^\cc \in \mathfrak{X}(T(\bar{Q}/H) \times_{\mbox{\tiny{$\bar{Q}/H$}}} \tilde{\mathcal O}_{\bar\mu})$ and observe that $T{\bf p}_2 (\tilde{X}_i^\cc) = T\pi_H (\bar X_i,  Y )$ for $\bar{X}_i$ given in \eqref{Eq:BasisTBarQ}.
Then, for  $\mathcal{X} \in \mathfrak{X}(T(\bar{Q}/H) \times_{\mbox{\tiny{$\bar{Q}/H$}}} \tilde{\mathcal O}_{\bar\mu})$, the evaluation of \eqref{Ap:Eq:IntrinsicEL}  in $\tilde{X}_i^\cc$ gives 
 \begin{equation}\label{Ap:Eq:proofProp1}
 \Omega_{\mbox{\tiny{$\mathfrak{L}$}}}(T{\bf p}_1({\mathcal X}), T{\bf p}_1(\tilde{X}_i^\cc) ) - d E_{\mathfrak{L}}(T{\bf p}_1(\tilde{X}_i^\cc) ) =  \beta_{\bar\mu}(T{\bf p}_2(\mathcal{X}), T{\bf p}_2(\tilde{X}_i^\cc) )+ d\mathfrak{C}_{\bar\mu}(T{\bf p}_2(\tilde{X}_i^\cc) ).
 \end{equation}
Since $T{\bf p}_1(\tilde{X}_i^\cc) = \partial/\partial x^i$ and $T{\bf p}_1(\mathcal{X}) = \dot x^i \partial/\partial x^i +v^i \partial/ \partial \dot x^i$ then the left hand side of \eqref{Ap:Eq:proofProp1} gives the Euler-Lagrange equations for $\mathfrak{L}$, which is the left hand side of \eqref{Eq:Classic:Splitted-RedLag2}.  In order to analyze the right hand side of \eqref{Ap:Eq:proofProp1} consider $X$ and $Y_i \in T(\bar{Q} \times \mathcal{O}_{\bar\mu})$ such that $T\pi_H(X) = T{\bf p}_2(\mathcal{X})$ and $T\pi_H(Y_i) = T{\bf p}_2(\tilde{X}_i^\cc)$ for $\pi_H:\bar{Q}\times \mathcal{O}_{\bar\mu} \to \tilde{\mathcal{O}}_{\bar\mu}$. Then, using the definition of $\beta_{\bar\mu}$ in \eqref{Ap:DefBeta} and denoting, as usual,  $\pi_1:\bar{Q} \times \mathcal{O} \to\bar{Q}$ and  $\pi_2:\bar{Q} \times \mathcal{O} \to \mathcal{O}$ the projections to the first and  second factors respectively, we obtain
  \begin{equation*}
   \begin{split}
    \beta_\mu(T{\bf p}_2({\mathcal X}), T{\bf p}_2(\tilde{X}^\cc_i)) & = d\alpha_{\bar{\mathcal A}}(X, Y_i) + \omega_{\mathcal O}(T\pi_2(X), T\pi_2(Y_i))   \\
    & = \langle \bar\mu, {\mathcal F} (T\pi_1 (X) , \bar{X}_i) \rangle - \langle \mu , [\bar{\mathcal A}(T\pi_1(X)), T\pi_2(Y_i)] \rangle + \omega_{\mathcal{O}}(T\pi_2(X), T\pi_2(Y_i)),
   \end{split}
  \end{equation*}
  where we follow \cite{Marsden2000} and \cite{MMORR} to compute $d\alpha_{\bar{\mathcal A}}(\tilde X, Y_i)$ taking into account that $T\pi_1(Y_i) = \bar{X}_i$ is horizontal. Moreover, since $\bar{\mathcal A}(T\pi_1(X)) = T\pi_2(X)$ then $\langle \mu , [\bar{\mathcal A}(T\pi_1(X)), T\pi_2(Y_i)] \rangle = \omega_{\mathcal{O}}(T\pi_2(X), T\pi_2(Y_i))$. Finally, denoting by $(d\mathfrak{C}_{\mu})_i := d\mathfrak{C}_\mu (T{\bf p}_2(\tilde{X}^\cc_i)) =  (\frac{\partial}{\partial x^i} - \bar{\mathcal{A}}_i^{\tilde a} \bar{Z}_{\tilde a})(\mathfrak{C}_{\mu})$, we obtain that 
  $$
   \beta_\mu(T{\bf p}_2({\mathcal X}), T{\bf p}_2(\tilde{X}^\cc_i))+ d\mathfrak{C}_\mu (T{\bf p}_2(\tilde{X}^\cc_i)) =  \mu_b F^b_{ij} \, \dot x^j + (d\mathfrak{C}_{\bar\mu})_i.
  $$
 \end{proof}

\begin{remark}
 Observe that \eqref{Eq:Classic:Splitted-RedLag2} is the classical Lagrange-Routh equations given in \cite{Marsden2000} while  \eqref{Eq:Classic:Splitted-RedLag1} is the second order equation condition as it was pointed out in \cite{CM1}.
\end{remark}

\noindent {\bf Routh reduction for nonholonomic systems with conserved momentum.} 
Recall from Section \ref{Ss:TheRouthian} that if $\overline{\langle {\mathcal J}_L, \mathcal{K}_\subW \rangle}$ is basic with respect to the projection $\phi_H:T\bar{Q} \to T\bar{Q}/H$ then the intrinsic nonholonomic Lagrange-Routh equations are given by \eqref{Eq:NH-Routh}.

\begin{lemma} \label{L:JK-Coord} If $\overline{\langle {\mathcal J}_L, \mathcal{K}_\subW\rangle}$ is basic with respect to $\phi_H:T\bar{Q} \to T\bar{Q}/H$, then, in the local coordinates $(x^i, \dot x^i, \xi^a) $ on $T\bar{Q}/H$
\begin{enumerate}
 \item[$(i)$] the 2-form $\overline{\langle {\mathcal J}_L, \mathcal{K}_\subW \rangle}_{\emph\red}$ has the form
$$
\overline{\langle {\mathcal J}_L, \mathcal{K}_\subW \rangle}_{\emph\red} = \frac{\partial l}{\partial \dot x^i} {\mathcal D}^i_{jk} dx^j\wedge dx^k + \frac{\partial l}{\partial \xi^a} {\mathcal D}^a_{jk} dx^j\wedge dx^k,
$$
where ${\mathcal D}^i_{jk}$ and ${\mathcal D}^a_{jk}$ are functions on $\bar{Q}/H$ defined in \eqref{Eq:Proof:JK-Coord}.
\item[$(ii)$]  For each $\bar\mu=\mu_ae^a \in \mathfrak{h}^*$,
 $$
 \Upsilon_l^* [\overline{\langle {\mathcal J}_L, \mathcal{K}_\subW\rangle}_{\emph\red}]_{\bar\mu} = \frac{\partial \mathfrak{L}}{\partial \dot x^i} {\mathcal D}^i_{jk} dx^j\wedge dx^k + \mu_a {\mathcal D}^a_{jk} dx^j\wedge dx^k,
$$
where $ [\overline{\langle {\mathcal J}_L, \mathcal{K}_\subW\rangle}_{\emph\red}]_{\bar\mu}$ is the pullback of $\overline{\langle {\mathcal J}_L, \mathcal{K}_\subW\rangle}_{\emph\red}$ to ${\bar J}_l^{-1}(\bar\mu)/H_{\bar\mu}$.
\end{enumerate}
\end{lemma}

\begin{proof}
$(i)$ From \eqref{Eq:Coef_K-C} we see that
\begin{equation} \label{Eq:Proof-JKCoord}
\langle {\mathcal J}_L, {\mathcal K}_\subW \rangle = \left(\frac{\partial L}{\partial v^A}\right)^*C_{\a\b}^A dr^\a\wedge dr^\b = \dot r^\gamma \kappa_{\gamma A}C_{\a\b}^A dr^\a\wedge dr^\b,
\end{equation}
where $\kappa_{\gamma A} =  \kappa(X_\gamma,Z_A)$ are the coefficients of the kinetic energy metric $\kappa$ of $L$ in the basis $\{X_\gamma, Z_A\}$ in \eqref{Eq:D+W-Coord}.

In local coordinates $(x^i, y^a)$ on $\bar{Q}$, define the functions $\lambda_i^\a$ and $\lambda _a^\a$ on $\bar{Q}$ such that $dr^\a = \lambda_i^\a dx^i + \lambda_a^\a \epsilon^a$, where $\{dx^i, \epsilon^a := \bar{Z}^a + \bar{\mathcal{A}}_i^a dx^i\}$ is the dual basis of \eqref{Eq:BasisTBarQ}. Therefore,
$$
\overline{\langle {\mathcal J}_L, \mathcal{K}_\subW \rangle} = \dot r^\gamma \kappa_{\gamma A} C_{\a\b}^A \lambda_i^\a\lambda_j^\b \, dx^i \wedge dx^j +  \dot r^\gamma \kappa_{\gamma A} K_{\a\b}^A \lambda_i^\a \lambda_a^\b \, dx^i \wedge \epsilon^a + \dot r^\gamma \kappa_{\gamma A} C_{\a\b}^A \lambda_a^\a\lambda_b^\b \, \epsilon^a \wedge \epsilon^b,
$$
where $\dot r^\gamma = \dot r^\gamma (x,y) = \lambda_i^\gamma \dot x^i  + \lambda_a^\gamma\xi^a$.
Finally, if $\overline{\langle {\mathcal J}_L, \mathcal{K}_\subW \rangle}$ is basic with respect to the orbit projection $\rho_H: T\bar{Q} \to T\bar{Q}/H$ then it does not have the $\epsilon^a$-coordinate, that is,
$$
\overline{\langle {\mathcal J}_L, \mathcal{K}_\subW \rangle}_\red =  (\dot x^k\lambda_k^\gamma  \kappa_{\gamma A} C_{\a\b}^A \lambda_i^\a\lambda_j^\b  + \xi^b \lambda_b^\gamma    \kappa_{\gamma A} C_{\a\b}^A \lambda_i^\a\lambda_j^\b  )  dx^i \wedge dx^j.
$$
Since we are working with the mechanical connection then $\partial l/\partial \dot x^l  = \dot x^k\bar{\kappa}_{kl}$ and $\partial l /\partial \xi^a = \bar{\kappa}_{ab}\xi^b$ . If we denote by $\bar{\kappa}^{ij}$ (respectively $\bar{\kappa}^{ab}$) the elements of the matrix $[\bar{\kappa}]^{-1}$---here $[\bar\kappa]$ is the matrix associated to the kinetic energy $\bar{\kappa}$---then the functions $\mathcal{D}_{ij}^l$ and $\mathcal{D}_{ij}^a$ on $\bar{Q}/H$ are given by
\begin{equation}\label{Eq:Proof:JK-Coord}
\mathcal{D}_{ij}^l = \bar{\kappa}^{kl}\lambda_k^\gamma  \kappa_{\gamma A} C_{\a\b}^A \lambda_i^\a\lambda_j^\b \qquad \mbox{and} \qquad
\mathcal{D}_{ij}^a = \bar{\kappa}^{ab} \lambda_b^\gamma    \kappa_{\gamma A} C_{\a\b}^A \lambda_i^\a\lambda_j^\b .
\end{equation}
Note that, in principle, the functions ${\mathcal D}_{jk}^i$ and ${\mathcal D}_{jk}^a$ are functions on $\bar{Q}$ but since we assume that $\overline{\langle \mathcal{J}_L,\mathcal{K}_\subW \rangle}$ is basic then ${\mathcal D}_{jk}^i$ and ${\mathcal D}_{jk}^a$ are well defined functions on $\bar{Q}/H$.
$(ii)$ It is a direct consequence of $(i)$ and \eqref{pAeqn}.
\end{proof}

If $\overline{\langle {\mathcal J}_L, \mathcal{K}_\subW \rangle}$ is basic with respect to the orbit projection $\phi_H: T\bar{Q} \to T\bar{Q}/H$, we denote by $\Gamma^j_{ki}$ and $\Gamma^a_{ki}$ the functions on $\bar{Q}/H$ given by
$$
\Gamma^j_{ki} = {\mathcal D}^j_{ki}-{\mathcal D}^j_{ik} \qquad \mbox{and} \qquad \Gamma^a_{ki} = {\mathcal D}^a_{ki} - {\mathcal D}^a_{ik}.
$$

Hence, from \eqref{Eq:NH-Routh} we obtain the nonholonomic version of the {\it Lagrange-Routh} equations:

\begin{proposition} \label{P:Lag-Routh}
If the 2-form $\overline{\langle {\mathcal J}_L, \mathcal{K}_\subW\rangle}$ is basic with respect to $\phi_H:T\bar{Q} \to T\bar{Q}/H$, then the coordinate version of the {\it intrinsic nonholonomic Lagrange-Routh equations} \eqref{Eq:NH-Routh}  are given by
\begin{subequations}
\begin{eqnarray}
\dot y^{\tilde a}& =& c_{\tilde a} {\bar B}_{\tilde a}^{\tilde b} - \dot x^i \bar{\mathcal A}_i^{\tilde c} \bar{B}_{\tilde c}^{\tilde b} \label{Eq:Splitted-RedLag1} \\
\frac{\partial {\mathfrak{L}}}{\partial x^{i}} -  \frac{d}{dt}\frac{\partial \mathfrak{L}}{\partial \dot{x}^{i}} & = &   \mu_b F^b_{ij} \, \dot x^j + (d\mathfrak{C}_{\bar\mu})_i + (\frac{\partial \mathfrak{L}}{\partial \dot{x}^{j}}  \Gamma^j_{ki} - \mu_a  \Gamma^a_{ki})\dot{x}^{k}
 \label{Eq:Splitted-RedLag2}
\end{eqnarray}
\end{subequations}
\end{proposition}

\begin{remark}
Equations \eqref{Eq:Splitted-RedLag1} and \eqref{Eq:Splitted-RedLag2} are a modification of the {\it Lagrange-Routh} equations given in \cite{CM1} by the extra term $\left( \frac{\partial \mathfrak{L}}{\partial \dot{x}^{j}}  \Gamma^j_{ki} - \mu_a  \Gamma^a_{ki}\right)\dot{x}^{k}$. This extra term carries the nonholonomic information of the system after the two reductions: first by $G_\subW$ and then by $H$.
\end{remark}

We conclude that equations \eqref{Eq:Splitted-RedLag1} and \eqref{Eq:Splitted-RedLag2} are the equations of motion on the level set $\bar{J}_l^{-1}(\bar\mu)/H_{\bar\mu}$ when the 2-form $\overline{\langle {\mathcal J}_L, \mathcal{K}_\subW\rangle}$ of the nonholonomic system is basic.
\bigskip

\noindent {\bf Nonholonomic Lagrange-Routh equations after a gauge transformation.} As seen in Section \ref{Ss:TheRouthian}, if we consider a gauge transformation $B$, the intrinsic nonholonomic Lagrange-Routh equations are given by \eqref{Eq:Lag-Red_Sympl_mfld-Gauge}. If we denote $\bar{B}_l = (\mathbb{F}l)^*\bar{B}$ then $ \left( \overline{\langle {\mathcal J}_L, \mathcal{K}_\subW \rangle} + \bar{B}_l \right) = (\mathbb{F}l)^* \left( \overline{\langle {\mathcal J}, \mathcal{K}_\subW \rangle} + \bar{B} \right)$ and thus we denote $\mathfrak{B}_l$ the 2-form on $T\bar{Q}/H$ such that $\phi_H^* \mathfrak{B}_l = \overline{\langle \mathcal{J}_L, \mathcal{K}_\subW\rangle} + \bar{B}_l$. 

%
%

As we did in Section~\ref{Ss:LagragianFormulation}, define the local coordinates $(x^i, y^a)$ associated to the local trivialization of $\bar{Q}\to \bar{Q}/H$ and consider the coordinates $(\dot x^i, \xi^a)$ on $T_{(x,y)}\bar{Q}$ associated to the local basis \eqref{Eq:BasisTBarQ}. If $\overline{\langle {\mathcal J}_L, \mathcal{K}_\subW \rangle} + \bar{B}_l$ is basic with respect to $\phi_H:T\bar{Q}\to T\bar{Q}/H$ then, from Lemma \ref{L:JK-Coord}, we may consider 2-forms $\bar{B}_l$ such that 
$$
\overline{\langle {\mathcal J}_L, \mathcal{K}_\subW \rangle} + \bar{B}_l = \left(\frac{\partial l}{\partial \dot x^i} (D^i_{ji} + \bar{B}^i_{jk}) + \frac{\partial l}{\partial \xi^a} (D_{jk}^a + \bar{B}_{jk}^a ) \right) dx^j \wedge dx^k ,
$$
where $B_{jk}^i, B_{jk}^a \in C^\infty(\bar{Q})$.
Therefore, 
\begin{equation}\label{Eq:Gauge:B-coord}
\Upsilon_l^*[\mathfrak{B}_l]_{\bar\mu} = \left( \frac{\partial \mathfrak{L}}{\partial \dot x^i}( {\mathcal D}^i_{jk} + \bar{B}^i_{jk} )+ \mu_a ({\mathcal D}^a_{jk} +  \bar{B}^a_{jk} ) \right) dx^j\wedge dx^k,
\end{equation}
where $ {\mathcal D}^i_{jk} + \bar{B}^i_{jk}$ and ${\mathcal D}^a_{jk} + \bar{B}^a_{jk}$ are functions on $\bar{Q}/H$.
If we denote by $(\Gamma_{\bar\B})^j_{ki}$ and $(\Gamma_{\bar\B})^a_{ki}$ the functions on $\bar{Q}/H$ given by
\begin{equation} \label{Eq:Gauge:Coef_JK-B}
(\Gamma_{\bar\B})^i_{jk} =({\mathcal D}^i_{jk} +\bar{B}^i_{jk}) - ({\mathcal D}^i_{kj} + \bar{B}^i_{kj})  \qquad \mbox{and} \qquad (\Gamma_{\bar\B})^a_{jk} = ({\mathcal D}^a_{jk} + \bar{B}^a_{jk}) - ({\mathcal D}^a_{kj} + \bar{B}^a_{kj}),
\end{equation}
then the reduced nonholonomic dynamics are restricted to the leaves $T(\bar{Q}/H) \times_{\mbox{\tiny{$\bar{Q}/H$}}} \bar{Q}/H_{\bar\mu}$ and the equations of motion are given by
\begin{subequations}\label{Eq:Gauge:LagrRouth}
\begin{eqnarray}
\dot y^{\tilde a}& =& c_{\tilde a} {\bar B}_{\tilde a}^{\tilde b} - v^i \bar{\mathcal A}_i^{\tilde c} \bar{B}_{\tilde c}^{\tilde b}\label{Eq:Gauge:LagraRouth1} \\
\frac{\partial {\mathfrak{L}}}{\partial x^{i}} -  \frac{d}{dt}\frac{\partial \mathfrak{L}}{\partial \dot{x}^{i}} & = &   \mu_b F^b_{ij} \, \dot x^j + (d\mathfrak{C}_{\bar\mu})_i + (\frac{\partial \mathfrak{L}}{\partial \dot{x}^{j}}  (\Gamma_{\bar\B})^j_{ki} - \mu_a  (\Gamma_{\bar\B})^a_{ki})\dot{x}^{k}.\label{Eq:Gauge:LagraRouth2}
\end{eqnarray}
\end{subequations}


\section{The conformal factor}\label{sec:conf}

In general \eqref{Eq:ChapReduction} and \eqref{Eq:Hamilt-MWreduction} are not Hamiltonian systems. However, there are several techniques that one may employ to attempt to make them Hamiltonian systems. The pursuit of this goal is referred to as {\em  Hamiltonization}. Let us now discuss the two most popular approaches to Hamiltonization for Chaplygin systems.

\subsection{Approaches to Hamiltonization at the compressed level} \label{sec:conf1}

In \cite{Stanchenko} Stanchenko showed that if one can find a function $f \in C^\infty(\bar{Q})$ such that $f\bar{\Omega}$ is closed then \eqref{Eq:ChapReduction} becomes Hamiltonian with respect to the new 2-form $f\bar{\Omega}$ and the new vector field $\bar{X}_\nh/f$. He also showed that a sufficient condition for this to occur is that
\begin{equation}\label{suffcondHam}
df \wedge \Theta_{\bar{Q}}=f\overline{\langle {\cal J},{\cal K}_{\cal W}\rangle}.
\end{equation}
Thus, in this approach to Hamiltonization one obtains dynamics with respect to the non-canonical 2-form $f\bar{\Omega}$.

Much earlier, Chaplygin showed in \cite{Chap} that in some instances a function $f \in C^\infty(\bar{Q}/H)$ may be found such that after the reparameterization of time $d\tau = f(r)\,dt$ the local equations of \eqref{Eq:ChapReduction} become Hamiltonian in the new time $\tau$, with respect to the canonical form $\Omega_{\mbox{\tiny{$\bar{Q}/H$}}}$. In \cite[eq.(2.17)]{Fernandez} the authors derived a set of coupled first-order partial differential equations for $f$ for $G$-Chaplygin systems whose shape space has arbitrary dimension. These conditions, as well as the reparameterization itself, were described from a global perspective in \cite{Oscar}. 

In both the Stanchenko and Chaplygin approaches the function $f$ is called the {\it conformal factor} or {\em multiplier}. However, as described above, in Stanchenko Hamiltonization one generally ends up with a non-canonical 2-form, whereas in Chaplygin Hamiltonization one generally ends up with a canonical 2-form. Lastly, we mention that the non-existence of a conformal factor at a certain level of reduction does not preclude the existence of a conformal factor at a further level of reduction. Indeed, we will discuss in Section \ref{sec:examples} two examples where this happens.

\subsection{Further reduction and Hamiltonization.} \label{sec:conf2}

In order to perform an (almost) symplectic reduction of the compressed system \eqref{Eq:Chap2form}, we assume that there is a 2-form $\bar{B}$ such that ${\bf i}_{\bar{X}_\nh} \bar{B}=0$ and $\overline{\langle {\mathcal J},{\mathcal K}_\subW\rangle} + \bar{B}$ is basic with respect to the orbit projection $\rho_H:T^*\bar{Q} \to T^*\bar{Q}/H$ (see  Section~\ref{Ss:Gauge2Setps}). Then, the (almost) symplectic reduction of  \eqref{Eq:Gauge-Chap2form} gives the reduced almost symplectic manifold $(J^{-1}(\bar\mu)/H_\mu,\bar\omega_{\bar\mu}^{\bar\B})$ or equivalently \eqref{Eq:Gauge-Ham-Red_Sympl_mfld}.

Following \cite{Oscar} let us assume that the Lie group $H=G/G_\subW$ is abelian.
In this case,  \eqref{Eq:Gauge-Ham-Red_Sympl_mfld} is the almost symplectic manifold
\begin{equation} \label{Eq:AbelianLeaf}
(T^*(\bar{Q}/H), \Omega_{\mbox{\tiny{$\bar{Q}/H$}}} - \beta_{\bar\mu} -\Upsilon^*\mathfrak{B}_{\bar\mu})
\end{equation}
where $\beta$ can be viewed as a 2-form on $\bar{Q}/H$ such that $\pi_H^* \beta_{\bar\mu} = \langle \bar\mu , d\bar{\mathcal{A}}\rangle$ for $\pi_G : \bar{Q} \to \bar{Q}/H$ (see \eqref{Ap:DefBeta}) and $\bar{\mathcal A} :T\bar{Q} \to \mathfrak{h}^*$ is the (mechanical) connection. We also remind the reader that $\mathfrak{B}_{\bar\mu}$ is the pullback of $\mathfrak{B}$ to the level $\bar{J}^{-1}(\bar{\mu})/H_{\bar{\mu}}$, where $\mathfrak{B}$ is the 2-form on $T^*\bar{Q}/H$ such that $\rho^*_H\mathfrak{B}=\overline{\langle {\cal J},{\cal K}_\subW\rangle}-\bar{B}$.

Our goal is now to find a conformal factor $f_\mu \in C^\infty(\bar{Q}/H)$  of the 2-form $\Omega_{\mbox{\tiny{$\bar{Q}/H$}}} - \beta_{\bar\mu} -\Upsilon^*\mathfrak{B}_{\bar\mu}$  on $T^*(\bar{Q}/H)$. We can proceed in a similar way as in the Chaplygin case, and now generalize the Stanchenko approach to Hamiltonization.

\begin{proposition}
At each $\bar\mu\in \mathfrak{h}^*$, a sufficient condition for the existence of a $f_{\bar\mu} \in C^{\infty}(\bar{Q}/H)$ such that $f_{\bar\mu} (\Omega_{\mbox{\tiny{$\bar{Q}/H$}}} - \beta_{\bar\mu} - \Upsilon^*\mathfrak{B}_{\bar\mu})$ is closed, is that
\begin{equation} \label{eq:stcond}
df_{\bar\mu} \wedge \Theta_{\mbox{\tiny{$\bar{Q}/H$}}} - f_{\bar\mu} (\beta_\mu + \Upsilon^*\mathfrak{B}_{\bar\mu}) = 0,
\end{equation}
where $\Omega_{\mbox{\tiny{$\bar{Q}/H$}}} = -d \Theta_{\mbox{\tiny{$\bar{Q}/H$}}}$.
\end{proposition}

\begin{proof}If we differentiate both sides of \eqref{eq:stcond}, we obtain that $d\left(f_{\bar\mu} (\Omega_{\mbox{\tiny{$\bar{Q}/H$}}} - \beta_{\bar\mu} - \Upsilon^*\mathfrak{B}_{\bar\mu}) \right) = 0$.
\end{proof}

\noindent When $f$ satisfies \eqref{eq:stcond}, the dynamics \eqref{Eq:Gauge:Hamilt-MWreduction} become the Hamiltonian system
\begin{equation} \label{Eq:ConfFactor:HamiltSystem}
{\bf i}_{(1/f_{\bar\mu} ) X_\red^{\bar{\mu}}} \left( f_{\bar\mu} ( {\Omega}_{\bar\mu} - \mathfrak{B}_{\bar\mu}) \right) = d (\Ham_{\bar\mu} ),
\end{equation}
where $\Omega_{\bar\mu} = \Omega_{\mbox{\tiny{$\bar{Q}/H$}}} - \beta_{\bar\mu}$ given in \eqref{Eq:Proof:equivalence}.
If $\beta_\mu=0$ and there is no need for a gauge transformation---i.e., $\overline{\langle {\mathcal J},{\mathcal K}_\subW\rangle}$ is already basic and thus $\mathfrak{B}= \overline{\langle {\mathcal J},{\mathcal K}_\subW\rangle}_\red$---then \eqref{eq:stcond} and \eqref{Eq:ConfFactor:HamiltSystem} reproduce the results of \cite{Stanchenko}.

Let us now discuss the generalization of Chaplygin Hamiltonization to our present setting. Specifically, for each $\bar\mu \in \mathfrak{h}^*$ we are now interested in transforming \eqref{Eq:Gauge:Hamilt-MWreduction} into a Hamiltonian system relative to the standard symplectic form $\Omega_{\mbox{\tiny{$\bar{Q}/H$}}}$. We will follow the discussion in \cite{Oscar}.

To begin, consider the vector field $(1/f_{\bar\mu}) X_\red^{\bar{\mu}}$ and define the diffeomorphism $\Psi_{f_{\bar\mu}}: T^*\bar{Q} \to T^*\bar{Q}$ such that $\Psi_{f_{\bar\mu}}(\alpha)=f_{\bar\mu}\alpha$.

\begin{proposition}\label{chapcfprop}
The system \eqref{Eq:Gauge:Hamilt-MWreduction} can be transformed into the Hamiltonian system
\begin{equation}\label{eq:heqXC}
{\bf i}_{\bar{X}^{\bar\mu}_C}\Omega_{\mbox{\tiny{$\bar{Q}/H$}}} = d\bar{\Ham}^{\bar\mu}_C,
\end{equation}
where $T\Psi_{f_{\bar\mu}}\circ ((1/f_{\bar\mu}) X_{\emph\red}^{\bar{\mu}} )=\bar{X}^{\bar\mu}_C \circ \Psi_{f_{\bar\mu}}$ and $\bar{\Ham}^{\bar\mu}_C=\Ham_{\bar\mu}\circ \Psi_{1/f_{\bar\mu}}$, if and only if $f_{\bar\mu}$ satisfies
\begin{equation} \label{Eq:SuffNec-Cond}
{\bf i}_{ X_{\emph\red}^{\bar{\mu}}} \left( df_{\bar\mu} \wedge \Theta_{\mbox{\tiny{$\bar{Q}/H$}}} - f_{\bar\mu} (\beta_\mu + \Upsilon^*\mathfrak{B}_{\bar\mu}) \right) = 0.
\end{equation}
\end{proposition}

\begin{proof}
Since $\rho^*_H\mathfrak{B}=\overline{\langle {\mathcal J},{\mathcal K}_\subW \rangle}+\bar{B}$ is merely a shift by $\bar{B}$ of $\overline{\langle {\mathcal J},{\mathcal K}_\subW\rangle}$, then from the calculations in \cite{Oscar} we have
$$ 
{\bf i}_{\bar{X}^{\bar\mu}_C}\Omega_{\mbox{\tiny{$\bar{Q}/H$}}} +  \Psi_{1/f_{\bar\mu}}^* \left( {\bf i}_{(1/f_{\bar\mu}) X_\red^{\bar{\mu}}} (df_{\bar\mu} \wedge \Theta_{\mbox{\tiny{$\bar{Q}/H$}}} - f_{\bar\mu} (\beta_\mu + \Upsilon^*\mathfrak{B}_{\bar\mu}) ) \right)    = d\bar{\Ham}^{\bar\mu}_C.
$$
\end{proof}
Note that condition \eqref{eq:stcond} is a sufficient condition for the dynamics to be described as in \eqref{eq:heqXC}.

\begin{remark} In \cite{Oscar} the necessary and sufficient condition for the Hamiltonization \eqref{eq:heqXC} is expressed as
${\bf i}_{\bar{X}^{\bar\mu}_C} \left(df_{\bar\mu} \wedge \Theta_{\mbox{\tiny{$\bar{Q}/H$}}} - f_{\bar\mu}^2 (\beta_\mu + \Psi_{1/f_{\bar\mu}}^* \Upsilon^*\mathfrak{B}_{\bar\mu}) \right) = 0,$
which is an equivalent condition to \eqref{Eq:SuffNec-Cond}.
\end{remark}

Following \cite{Fernandez} we now give a set of coupled differential equations describing the conformal factor $f_{\bar\mu}$ for each leaf \eqref{Eq:AbelianLeaf}. These equations are the local version of \eqref{Eq:SuffNec-Cond}.

Observe that the equations of motion on each almost symplectic leaf are given by the {\it nonholonomic Lagrange-Routh equations} (Prop.\,\ref{P:Lag-Routh} or equivalently by \eqref{Eq:Gauge:LagrRouth}). If $H$ is abelian, these equations reduce to equation \eqref{Eq:Splitted-RedLag1} (or \eqref{Eq:Gauge:LagraRouth1}).
Moreover, using the $H$-invariance of the Lagrangian $l:T\bar{Q}\to \R$ in \eqref{Eq:RedLagrangian} we can write the reduced Lagrangian $\mathfrak{L} :T^*(\bar{Q}/H) \to \R$ defined in \eqref{Def:Red-Lagragian} in local coordinates as $\mathfrak{L}(x,\dot x) = \frac{1}{2} \bar{\kappa}_{ij}(x) \dot x^i \dot x^j - \mathcal{V}(x)$.
By the same argument as in \cite{Fernandez} and using that $\mathfrak{B}_{\bar\mu}$ is given in \eqref{Eq:Gauge:B-coord}, we get that the conformal factor $f_\mu$ exists if and only
\begin{subequations}\label{Eq:Gauge:PDE}
\begin{equation}\label{Eq:PDEconfFactor}
\frac{\partial f}{\partial x^i}\bar\kappa_{jk} + \frac{\partial f}{\partial x^k}\bar{\kappa}_{ij} -2\frac{\partial f}{\partial x^j}\bar\kappa_{ik}=f\left(\bar\kappa_{lk} (\Gamma_{\bar{\B}})^{l}_{ij} + \bar\kappa_{li} (\Gamma_{\bar{\B}})^{l}_{kj} \right)
\end{equation}
\begin{equation}\label{Eq:PDE-condition}
\mu_b \left( F_{ij}^b - (\Gamma_{\bar\B})_{ji}^b\right) = 0.
\end{equation}
\end{subequations}
for each triple $(i,j,k)$ and where $(\Gamma_{\bar\B})_{ji}^k$ and $(\Gamma_{\bar\B})_{ji}^b$  are defined in
\eqref{Eq:Gauge:Coef_JK-B}. 
Observe that if such a conformal factor exists, then it does not depend on the leaf.

When $\beta_\mu$ and the $(\Gamma_{\bar\B})_{ji}^b$ are all zero these conditions reduce to those found in \cite[Eq. (2.17)]{Fernandez}. Those authors also derive local conditions for $f$ to exist in the general case when $H$ is not abelian (though no gauge transformations are considered); this results in a much more complicated set of first-order partial differential equations in $f$.

Finally, let us discuss Hamiltonization in the context of twisted Poisson brackets.
Suppose that $f$ is a nowhere vanishing function on a manifold $P$ and that $\{\cdot , \cdot \}$ is a twisted Poisson bracket on $P$ with almost symplectic leaves $(\mathcal{O}_\mu, \omega_\mu)$. Then $1/f$ is a conformal factor  $\{\cdot,\cdot\}$ if and only if the restriction of $f$ to $\mathcal{O}_\mu$, denoted by $f|_{\mathcal{O}_\mu}$,  is a conformal factor of $\omega_\mu$.
 In our context, since $f$ is a function on $\bar{Q}/H= Q/G$ then  $\tau_G^* (1/f) \in C^\infty(\M/G)$ is a conformal factor for a twisted Poisson bracket  $\{\cdot,\cdot\}_\red^\B$ if and only if the restriction of  $\tau_G^* f$ to  $T^*(\bar{Q}/H) \times_{\mbox{\tiny{$\bar{Q}/H$}}} \tilde{\mathcal{O}}_{\bar\mu}$ is a conformal factor of \eqref{Eq:Gauge-Ham-Red_Sympl_mfld}, for $\tau_G:\M/G \to Q/G$.


\section{Examples}\label{sec:examples}


\subsection{The Snakeboard} \label{Ss:Ex:Snake}

In this section we discuss the snakeboard example (or skateboard in some literature). We use the formulation of the problem given in \cite{KoMa1997} and subsequently studied in \cite{Bloch:Book,Fernandez,Oscar}.

The system is modeled as a rigid body (the board) with two sets of independent actuated wheels, one on each end of the board. The human rider is modeled as a momentum wheel which sits in the middle of the board and is allowed to spin about the vertical axis. We denote by $m$ the total mass of the board, $J$ the inertia of the board, $J_{0}$ the inertia of the rotor, $J_{1}$ the inertia of each of the wheels, and assume the relation $J + J_{0} + 2 J_{1} = m R^{2}$, where $R$ is the radius from the center of the board to the pivot point of the wheel axle. The configuration space is
\begin{equation*}
Q = SE(2) \times S^{1} \times S^{1}= \{ (\theta, x, y, \varphi, \psi) \}.
\end{equation*}
The Lagrangian $L: TQ \to \R$ is 
\begin{equation}\label{Lsb}
L = \frac{1}{2} \brackets{m \parentheses{ \dot{x}^{2} + \dot{y}^{2} + R^{2}\dot{\theta}^{2} }+ 2J_{0}\,\dot{\theta}\,\dot{\psi}+ 2J_{1}\,\dot{\varphi}^{2}+ J_{0}\,\dot{\psi}^{2}}
\end{equation}
and whenever $\varphi \neq 0,\pi$, the nonholonomic constraints can be represented by the annihilator of the one-forms
\begin{equation}\label{epsbdef}
\epsilon^x=dx + R \cot\varphi\,\cos\theta\,d\theta, \qquad \epsilon^y=dy + R \cot\varphi\,\sin\theta\,d\theta.
\end{equation}

The Lie group $G=\R^2 \times \mathbb{S}^1$ is a symmetry of the nonholonomic system with the action on $Q$ given by
\begin{equation*}
G \times Q \to Q; \quad \parentheses{(a, b, \beta), (\theta, x, y, \varphi, \psi )} \mapsto (\theta, x + a, y + b, \varphi, \psi+\beta).
\end{equation*}
The distribution $D$ and a vertical complement $W$ \eqref{Eq:D+W} of $D$ in $TQ$ are given by
\begin{equation} \label{Ex:Snake:D+W}
 D= \textup{span} \left\{ \frac{\partial}{\partial \theta} - R \cos\theta\cot\varphi \frac{\partial}{\partial x} - R \sin\theta\cot\varphi \frac{\partial}{\partial y} , \frac{\partial}{\partial \varphi}, \frac{\partial}{\partial \psi}\right\} \quad \mbox{and} \quad W=\textup{span}\left\{ \frac{\partial}{\partial x} , \frac{\partial}{\partial y} \right\}.
\end{equation}
As was observed in \cite{paula}, $W$ satisfies the vertical symmetry condition \eqref{Eq:VerticalSymmetries} for the Lie algebra $\mathfrak{g}_\subW = \R^2$.
\bigskip

\noindent {\bf Description of the almost symplectic leaves.}  Let \ $(\theta, x, y, \varphi, \psi ; p_\theta, p_\varphi, p_\psi, p_x, p_y)$ \ be \ the \ coordinates \ associated \ to  \ the \ basis \ $\{ d\theta, d\varphi, d\psi, \epsilon^x, \epsilon^y\}$, which is dual to \eqref{Ex:Snake:D+W}.
The associated constrained momentum space $ \kappa^\sharp(D) = \mathcal{M} \subset T^{*}Q$ is given by
$$
\M = \left \{ p_x= -\frac{m R \cos\theta \sin ^2 \varphi \cot
\varphi}{mR^2 - {J}_0 \sin^2\varphi}  ( p_\theta - p_\psi),
\ \ \ \  p_y = -\frac{mR \sin\theta  \sin ^2 \varphi \cot
\varphi}{mR^2 - {J}_0 \sin^2\varphi}  ( p_\theta - p_\psi)
\right \}. $$



The Lie group $G_\subW = \R^2$ ---defining the translational symmetry--- is a normal subgroup of $G$ and has Lie algebra $\mathfrak{g}_\subW=\R^2$. Thus the system is $G_\subW$-Chaplygin and the compressed nonholonomic motion takes place in $T^*\bar{Q}$, where $\bar{Q} = SO(2) \times \mathbb{S}^1\times \mathbb{S}^1$ with coordinates $(\theta, \varphi,\psi)$. The 2-form on $T^*\bar{Q}$ describing the dynamics is $\bar{\Omega} = \Omega_{\bar{Q}} - \overline{\langle \mathcal{J}, \mathcal{K}_\subW \rangle}$, where
\begin{equation} \label{Ex:Snake:JKRed}
\overline{\langle {\mathcal J}, \mathcal{K}_\subW \rangle}=  - \frac{mR^2 \cot\varphi}{mR^2 - {J}_0 \sin^2\varphi} ( p_\theta - p_\psi)
d\theta \wedge d\varphi.
 \end{equation}

The action of the Lie group $H = G/G_\subW = \mathbb{S}^1$---which represents rotations in the $\psi$-variable---is a symmetry of the compressed system, with the action on $\bar{Q}$ defined as in \eqref{Eq:H-action}:
\begin{equation*}
H \times \bar{Q} \rightarrow \bar{Q};
    \quad
\parentheses{ \beta, (\theta, \varphi, \psi) } \mapsto (\theta, \varphi, \psi + \beta).
\end{equation*}

Since  $\overline{\langle {\mathcal J}, \mathcal{K}_\subW \rangle}$ is basic with respect to the orbit projection $\rho_H:T^*\bar{Q} \to
T^*\bar{Q}/H$, by Corollary \ref{C:JKbasic=conservationLaws} the canonical momentum map $\bar{J}:T^*\bar{Q}\to \R$ is a momentum map for $\Omega_{\mbox{\tiny{$\bar{Q}$}}} - \overline{\langle \mathcal{J}, \mathcal{K}_\subW \rangle}$. Thus, for $\xi = 1 \in \R$, the function $\bar{J}_\xi = g_\xi = p_\psi$ is conserved by the compressed motion. Moreover, recalling that $\varrho: \mathfrak{g}\to \mathfrak{h}$ and since $\varrho(0,0,1) =1$, by Props.\,\ref{Prop:f-conserved} and \ref{Prop:g_xi}, $f_{\mbox{\tiny{$(0,0,1)$}}} = \rho_H^*g_1 = p_\psi$ is conserved by the nonholonomic motion $X_\nh$.

%
%
%

At each $\mu_\psi \in \mathfrak{h}^*$ the Marsden-Weinstein reduction (Section~\ref{Ss:2S-Reduction}) gives the almost symplectic manifolds $(\bar{J}^{-1}(\mu_\psi)/H, \bar\omega_{\mu_\psi})$, where $\bar{J}^{-1}(\mu_\psi)/H$ is represented by the coordinates $(\theta, \varphi,  p_\theta,  p_\varphi, \mu_\psi= p_\psi)$ and the almost symplectic 2-form is given by $\bar\omega_{\mu_\psi} = {\Omega}_{\mu_\psi} - [\overline{\langle {\mathcal J}, \mathcal{K}_\subW \rangle}_\red]_{\mu_\psi}$,
where ${\Omega}_{\mu_\psi}$ is the 2-form obtained by the reduction of the canonical symplectic manifold $(T^*\bar{Q}, \Omega_{\bar{Q}})$ and $[\overline{\langle {\mathcal J}, \mathcal{K}_\subW \rangle}_\red]_{\mu_\psi}$ is the 2-form \eqref{Ex:Snake:JKRed} on $\bar{J}^{-1}(\mu_\psi)/H$.

Now, since $\overline{\langle {\mathcal J}, \mathcal{K}_\subW \rangle}$ is basic with respect to the orbit projection $\rho_H:T^*\bar{Q} \to T^*\bar{Q}/H$ then reduced bracket $\{\cdot , \cdot \}_\red$ is twisted Poisson (Theorem \ref{T:Reduced-Dyn}).
Thus, by Theorem \ref{T:Equivalence}, $(\bar{J}^{-1}(\mu_\psi)/H, {\Omega}_{\mu_\psi} - [\overline{\langle {\mathcal J}, \mathcal{K}_\subW \rangle}_\red]_{\mu_\psi})$ are the almost symplectic leaves of the reduced bracket $\{\cdot , \cdot \}_\red$.
Moreover, using Lemma \ref{L:identifications} and \eqref{Eq:Ham-Red_Sympl_mfld} these almost symplectic leaves are symplectomorphic to
\begin{equation} \label{Ex:Snake:RedMfld}
(T^*(\mathbb{S}^1\times\mathbb{S}^1), \Omega_{\mbox{\tiny{$\mathbb{S}^1\!\!\times\!\mathbb{S}^1$}}} - \beta_{\mu_\psi} -  \Upsilon^*[\overline{\langle {\mathcal J}, \mathcal{K}_\subW \rangle}_\red]_{\mu_\psi}),
\end{equation}
where $\Omega_{\mbox{\tiny{$\mathbb{S}^1\!\!\times\!\mathbb{S}^1$}}}$ is the canonical form on $T^*(\mathbb{S}^1\times \mathbb{S}^1)$ and $\beta_{\mu_\psi} = \langle \mu_\psi, d\tilde{\mathcal{A}}\rangle$ for $\tilde{\mathcal{A}} :T\bar{Q} \to \mathfrak{h}^*$ a principal connection (see \eqref{Ap:DefBeta}).
\bigskip

\noindent {\bf Coordinate version of the Lagrange-Routh equations.}  \ The compressed Lagrangian $l: T\bar{Q} \to \R$, given by the kinetic energy $\bar\kappa$ (see e.g. \cite{KoMa1998}), is
\begin{equation}\label{Ex:sb:redl}
l(\theta,\varphi,\psi, \dot\theta, \dot\varphi,\dot\psi)=\frac{1}{2}mR^2\csc^2\varphi \, \dot{\theta}^2+\frac{1}{2}J_0 \dot{\psi}^2+J_0\dot{\psi}\dot{\theta}+J_1\dot{\varphi}^2.
\end{equation}

\begin{remark}
From Section \ref{Ss:TheRouthian} it follows that $\frac{d}{dt}\frac{\partial l}{\partial \dot{\psi}} = 0$.
Therefore, we can now carry out Routhian reduction to further reduce the dynamics to the level set $p_{\psi}=\mu_{\psi}$. Instead of proceeding with the reduction, we will first change variables on $T\bar{Q}$ using the mechanical connection in order to discuss the reduced Routhian $\mathfrak{R}_{\mu_{\psi}}$ and the function $\mathfrak{C}_{\mu_{\psi}}$.
\end{remark}

Recall that the vertical space $V_H$ associated to the $H$-action on $\bar{Q}$ is generated by $\partial/\partial \psi$. Using \eqref{Ex:sb:redl} the mechanical connection $\bar{\mathcal A}:T\bar{Q} \to \R$ is then given by the 1-form $\bar{\mathcal{A}}= d\psi - d\theta$.
This connection induces a basis on $T\bar{Q}$:
$$
\mathcal{B} = \{ Z_\theta:= \frac{\partial}{\partial \theta} - \frac{\partial}{\partial \psi} , \ Z_\varphi := \frac{\partial}{\partial \varphi} , \ Z_\psi :  = \frac{\partial}{\partial \psi} \},
$$
where $Z_\theta$ and $Z_\varphi$ are both $\bar\kappa$-orthogonal to the vertical vector field $Z_\psi$. Let us denote by $(\dot\theta, \dot\varphi, \eta)$ the coordinates associated to $\mathcal{B}$ and by $(\tilde p_\theta, \tilde p_\varphi ,\tilde p_\eta =  p_\psi)$ the coordinates associated to the dual basis $\{ d\theta, d\varphi, d\psi + d\theta \}$ on $T^*\bar{Q}$.
In the local coordinates $(\theta, \varphi, \psi, \dot\theta, \dot\varphi, \eta)$, \eqref{Ex:sb:redl} becomes
$$ 
l(\theta, \varphi, \psi, \dot\theta, \dot\varphi, \eta)=\left(\frac{mR^2-J_0\sin^2\varphi}{2\sin^2\varphi}\right)\dot{\theta}^2+J_1\dot{\varphi}^2+\frac{J_0}{2}\eta^2,
$$
which now has the form \eqref{Eq:RedLagrangian}, and the (basic) 2-form $\overline{\langle {\mathcal J}_L, \mathcal{K}_\subW \rangle}$ is
\begin{equation} \label{Ex:Snake:JKLine}
\overline{\langle {\mathcal J}_L, \mathcal{K}_\subW \rangle} =  - \frac{mR^2 \cot \varphi}{\sin^2\varphi}  \dot \theta d\theta \wedge d\varphi.
\end{equation}

As we have already observed, since $\overline{\langle {\mathcal J}_L, \mathcal{K}_\subW \rangle}$ is basic with respect to $\phi_H:T\bar{Q}\to T\bar{Q}/H$, then $\partial l/\partial \eta = \mu_\psi$ is a conserved quantity and we can therefore define the Routhian $R^{\mu_\psi}$ on $J_l^{-1}(\mu_\psi)$ as in \eqref{re}:
$$ 
R^{\mu_\psi}(\theta, \varphi, \psi; \dot \theta, \dot \varphi)=\left(\frac{mR^2-J_0\sin^2\varphi}{2\sin^2\varphi}\right)\dot{\theta}^2+J_1\dot{\varphi}^2-\frac{\mu^2_{\psi}}{2J_0}.
$$


The Routhian $R^{\mu_\psi}$ is $H$-invariant and thus, following \eqref{Eq:RedRouthian} and since $\bar{J}_l^{-1}(\mu_\psi)/H \simeq T(\mathbb{S}^1\times \mathbb{S}^1)$,  we obtain the {\it reduced Routhian} $\mathfrak{R}^{\mu_\psi} : T(\mathbb{S}^1 \times \mathbb{S}^1) \to \R$.
Finally, the {\it reduced Lagrangian} $\mathfrak{L}:T(\mathbb{S}^1\times\mathbb{S}^1) \to \R$,
\begin{equation} \label{Ex:Snake:RedLagrangian}
\mathfrak{L}(\theta, \varphi; \dot \theta, \dot \varphi) = \left(\frac{mR^2-J_0\sin^2\varphi}{2\sin^2\varphi}\right)\dot{\theta}^2+J_1\dot{\varphi}^2,
\end{equation}
is defined by the {\it reduced kinetic energy} on $T(\mathbb{S}^1\times\mathbb{S}^1)$, and the function $\mathfrak{C}_{\mu_\psi}$ is the constant function on $\mathbb{S}^1 \times \mathbb{S}^1$ given by
$\displaystyle{\mathfrak{C}_{\mu_\psi} = \frac{\mu_\psi^2}{2J_0} }$, so that $\mathfrak{R}^{\mu_\psi} =  \mathfrak{L} - \mathfrak{C}_{\mu_\psi}$.

To derive the nonholonomic version of the Lagrange-Routh equation, first observe that the 2-form $\overline{\langle {\mathcal J}_L, \mathcal{K}_\subW \rangle}$ does not depend on $\eta$, that is, the local expression of $[\overline{\langle {\mathcal J}_L, \mathcal{K}_\subW \rangle}_\red]_{\mu_\psi}$ on $J^{-1}_l(\mu_\psi)/H$ coincides with \eqref{Ex:Snake:JKLine}. Second, note also that the structure constants of the Lie algebra $\mathfrak{h}$ are zero and that the curvature associated to the mechanical connection $\bar{\mathcal{A}}$ is also zero.  Therefore, from Proposition \ref{P:Lag-Routh} the reduced dynamics becomes
\begin{eqnarray}
\displaystyle \ddot{\theta}\left(mR^2 -J_0\sin^2\varphi \right)- mR^2 \cot\varphi\,\dot{\varphi}\,\dot{\theta} & = &0, \nonumber \\
2 J_1\ddot{\varphi} - mR^2 \frac{\cot \varphi}{\sin^2\varphi} \dot \theta^2 & = & 0, \label{sb3rdeq2} \\
\dot{\mu}_{\psi} & = & 0, \nonumber
\end{eqnarray}
which is the nonholonomic version of the Lagrange-Routh equations.

Next we study the almost symplectic manifolds  $(\bar{J}^{-1}(\mu)/H, \bar\omega_{\bar\mu})$ (as in Section \ref{Ss:TheRouthian}).
\bigskip

\noindent {\bf Intrinsic version of the Lagrange-Routh equations.}  \ Since the mechanical connection $\bar{\mathcal{A}}$ has zero curvature, then $\beta_{\mu_\psi} =d[\mu_\psi \bar{\mathcal{A}}]=0$ and thus the 2-form on $T^*(\mathbb{S}^1\times \mathbb{S}^1)$ defining the dynamics is simply $\Omega_{\mbox{\tiny{$\mathbb{S}^1\!\!\times\!\mathbb{S}^1$}}} -  \Upsilon^*[\overline{\langle {\mathcal J}, \mathcal{K}_\subW \rangle}_\red]_{\mu_\psi}$. where
\begin{equation} \label{Ex:Snake:JKRedMu}
[\overline{\langle {\mathcal J}, \mathcal{K}_\subW \rangle}_\red]_{\mu_\psi} =  - \frac{mR^2 \cot\varphi}{mR^2 - {J}_0 \sin^2\varphi} \tilde p_\theta d\theta \wedge d\varphi.
 \end{equation}

The {\it nonholonomic Lagrange-Routh equations} \eqref{sb3rdeq2} on $T(\mathbb{S}^1 \times \mathbb{S}^1)$ are the equations of motion given by the almost symplectic 2-form
$$
\Omega_{\mathfrak{L}} - \Upsilon_l^*[\overline{\langle {\mathcal J}_L, \mathcal{K}_\subW \rangle}_\red]_{\mu_\psi},$$
where  $\Omega_{\mathfrak{L}}$ is the Lagrangian 2-form associated to the {\it reduced Lagrangian} $\mathfrak{L}:T(\mathbb{S}^1\times\mathbb{S}^1) \to \R$
given in \eqref{Ex:Snake:RedLagrangian}
and $\Upsilon_l^*[\overline{\langle {\mathcal J}_L, \mathcal{K}_\subW \rangle}_\red]_{\mu_\psi}$ coincides with the formulation given in local coordinates in \eqref{Ex:Snake:JKLine}.
The associated energy $\mathcal{E}_{\mu_\psi}$ is given by $\mathcal{E}_{\mu_\psi} =  E_{\mathfrak{L}}+ \mathfrak{C}_{\mu_\psi}$, where $E_{\mathfrak{L}}$ is the Lagrangian energy associated to \eqref{Ex:Snake:RedLagrangian}, and $\mathfrak{C}_{\mu_\psi}$ is the constant function on $\mathbb{S}^1 \times \mathbb{S}^1$ given by
$\displaystyle{\mathfrak{C}_{\mu_\psi} = \mu_\psi^2/2J_0 }$.
\bigskip

\noindent {\bf The conformal factor. } \ Since $\beta_{\mu_{\psi}}=0$, condition \eqref{eq:stcond} is given by
$$\frac{\partial f}{\partial \theta} \tilde p_\varphi - \frac{\partial f}{\partial \varphi}  \tilde p_\theta + f \, mR^2\sin^{-2}\varphi  \cot \varphi \,  \dot \theta = 0.$$ That is,
\begin{equation*}
 0 = \frac{\partial f}{\partial \theta} 2 J_1 \qquad \mbox{and} \qquad  0 =  - \frac{\partial f}{\partial \varphi} (mR^2 - J_0\sin^2\varphi)  + f  mR^2 \cot \varphi.
\end{equation*}
Integrating we obtain
\begin{equation} \label{Ex:Snake:Conf-Factor}
f(\theta,\varphi)=\frac{\sin\varphi}{\sqrt{ m R^{2} - J_{0} \sin^{2}\varphi }}.
\end{equation}

Now, as discussed in Section \ref{sec:conf} the sufficient condition \eqref{eq:stcond} for Hamiltonization via the Stanchenko and Chaplygin approaches are the same. Thus, \eqref{Ex:Snake:Conf-Factor} accomplishes the Hamiltonization of the almost symplectic manifold \eqref{Ex:Snake:RedMfld} (with $\beta_{\mu_\psi} = 0$ and $[\overline{\langle {\mathcal J}, \mathcal{K}_\subW \rangle}_\red]_{\mu_\psi}$ given in \eqref{Ex:Snake:JKRedMu}) in both the Stanchenko approach---since $d(f(\theta,\varphi)(\Omega_{\mbox{\tiny{$\mathbb{S}^1\!\!\times\!\mathbb{S}^1$}}} - \Upsilon^*[\overline{\langle {\mathcal J}, \mathcal{K}_\subW \rangle}_\red]_{\mu_\psi})$ is closed---and in the Chaplygin approach, since the dynamics can be written in terms of the $\Psi_f$-related vector field $\bar{X}^{\bar\mu}_C$ and the canonical 2-form $\Omega_{\mbox{\tiny{$\mathbb{S}^1\!\!\times\!\mathbb{S}^1$}}}$ on $T^*(\mathbb{S}^1\times \mathbb{S}^1)$, as in Proposition \ref{chapcfprop}. The latter result reproduces the findings in \cite{Oscar}. Note also that condition \eqref{Eq:PDE-condition} is satisfied since, in coordinates $(\tilde p_\theta,\tilde p_\varphi, \tilde p_\eta)$ induced by the mechanical connection, the 2-form $\overline{\langle {\mathcal J}, \mathcal{K}_\subW \rangle}_\red$ does not depend on $\mu_\psi$ and the curvature terms $F_{ij}^a$ associated to this connection are zero.


Finally, from Theorem \ref{T:Reduced-Dyn} it follows that the reduced nonholonomic bracket $\{ \cdot, \cdot \}_\red$ is twisted Poisson, since $\langle {\mathcal J}, \mathcal{K}_\subW \rangle$ is basic.
In coordinates $(\theta, \varphi, \tilde p_\theta, \tilde p_\varphi, \tilde p_\psi)$ on $\M/G$, the only nonzero entries of that are
$$
\{ \theta, \tilde p_\theta\}_\red = 1, \quad \{ \varphi, \tilde p_\varphi\}_\red = 1, \quad \{ \tilde p_\varphi, \tilde p_\theta\}_\red =  \frac{mR^2 \cot\varphi}{mR^2 - {J}_0 \sin^2\varphi} \tilde p_\theta.
$$
The function $g$ on $\M/G$ given by
$$
g(\theta,\varphi) = \frac{1}{f(\theta, \varphi)} = \frac{\sqrt{ m R^{2} - J_{0} \sin^{2}\varphi }}{\sin\varphi},
$$
is a conformal factor for the twisted Poisson bracket $\{ \cdot, \cdot \}_\red$.   That is, we have a Hamiltonization in the context of twisted Poisson brackets, since $g \{\cdot, \cdot\}_\red$ is Poisson.  It is worth mentioning that the conformal factor already appear in \cite{Oscar}. However, in this case, we computed the conformal factor as a conformal factor on each leaf of the almost symplectic foliation \eqref{Eq:Lag-Red_Sympl_mfld} which in our case is $T(\mathbb{S}^1 \times \mathbb{S}^1)$ having to deal with less variables as in \cite{Oscar} and thus obtaining simpler calculations. 

\begin{remark}
Without the velocity shift induced by the mechanical connection we could still Hamiltonized the compressed reduced system \eqref{Eq:Hamilt-MWreduction} in certain cases using the results of \cite{BorMa} and \cite{Fernandez}. There the Hamiltonization of Chaplygin systems whose constraint reactions forces contain linear velocity terms was studied, and it was shown that in general after rescaling by the conformal factor one obtains non-canonical {\em Poisson} brackets that depend on the level $\mu$ associated with the $H$ reduction. The particular example of the snakeboard was discussed in \cite{Fernandez}. \end{remark}



\subsection{The Chaplygin Ball} \label{Ss:Ex:Ball}

The Chaplygin ball \cite{ChapB,BorMa2001} is an inhomogeneus sphere of radius $r$ and mass $m$ whose center of mass coincides with the geometric center, and that rolls without sliding on a plane.

The treatment of this example will follow \cite{paula} and \cite{Hoch} (unifying both approaches to Hamiltonization). The configuration manifold is $Q= SO(3) \times \R^2$ and the non-sliding constraints are given by
$$ \dot x = r \omega_2 \qquad \mbox{and} \qquad \dot y = -r \omega_1,
$$
where $(x,y)$ represents the center of mass of the ball and $\vecom = (\omega_1, \omega_2, \omega_3)$ is the angular velocity of the body with respect to an inertial frame.
If $g \in SO(3)$ is the orthogonal matrix that specifies the orientation of the ball by relating the orthogonal frame attached to the body with the one fixed in space, then the angular velocity in space coordinates $\vecom$ and in body coordinates $\vecOm=(\Omega_1, \Omega_2, \Omega_3)$ are related by $\vecom=g\vecOm$.
The Lagrangian is given by the kinetic energy
$$
\mathcal{L}(g,{\bf x}; \vecOm , \dot {\bf x}) = \kappa((\vecOm , \dot {\bf x}), (\vecOm , \dot {\bf x})) =
\frac{1}{2}\langle \mathbb{I} \vecOm,   \vecOm \rangle + \frac{m}{2}||\dot {\bf x}||^2,
$$
where $\mathbb{I}$ is the inertia tensor that is represented as a diagonal $3\times 3$ matrix whose positive entries $\mathbb{I}_1, \mathbb{I}_2, \mathbb{I}_3$, are the principal moments of inertia.

Let $\vecL = (\lambda_1,\lambda_2, \lambda_3)$ and $\vecR=(\rho_1, \rho_2, \rho_3)$ be the left and right Maurer-Cartan 1-forms on $\mathfrak{so}(3)\simeq \R^3$, respectively; if $v_g \in T_gSO(3)$ then $\vecom= \vecR(v_g)$ and $\vecOm=\vecL(v_g)$.  Therefore the constraints 1-forms, in terms of body coordinates, are written as
$$
\epsilon^x= dx - r \langle \vecbeta ,\vecL \rangle = dx - r\beta_i \lambda_i \qquad \mbox{and} \qquad \epsilon^y = dy + r \langle \vecalpha , \vecL \rangle = dy + r \alpha_i \lambda_i,
$$
where $\vecalpha = (\alpha_1, \alpha_2, \alpha_3)$ and $\vecbeta = (\beta_1, \beta_2, \beta_3)$ are the first and second rows of the matrix $g \in SO(3)$.

Following \cite{Naranjo2008}, consider the Lie group $G=\{ (h,a) \in SO(3)\times \R^2 \ : \ h{\bf e}_3 = {\bf e}_3\}$, where ${\bf e}_3$ is the canonical vector ${\bf e}_3=(0,0,1)$, and its action on $Q$ given by
\begin{equation} \label{Ex:Ball:G-action}
(h,a): (g,{\bf x}) \mapsto (hg, h{\bf x} + a) \in \SO(3) \times \R^3.
\end{equation}
Then $G$ is a symmetry of the nonholonomic system, and in this case the vertical space $V$ is given by
$$
V= \textup{span} \left\{ \vecgamma \cdot {\bf X}^{\mbox{\tiny left}}
- y  \frac{\partial}{\partial x} + x  \frac{\partial}{\partial y}\, , \,
\frac{\partial}{\partial x}\, , \, \frac{\partial}{\partial y}  \right\},
$$
where $\vecgamma = (\gamma_1, \gamma_2,\gamma_3)$ is the third row of the matrix $g \in \SO(3)$ and ${\bf X}^{\mbox{\tiny left}} = (X^{\mbox{\tiny left}}_1, X^{\mbox{\tiny left}}_2, X^{\mbox{\tiny left}}_3)$ is the moving frame of $SO(3)$ dual to $\vecL$. Therefore, we have a basis of $TQ$ adapted to $D\oplus W$, where $W$ is a vertical complement of $D$ in $TQ$:
\begin{equation} \label{Ex:Ball:D+W}
D =  \textup{span} \left\{ X_1^{\mbox{\tiny left}} + r\beta_1 \frac{\partial}{\partial x} -r\a_1 \frac{\partial}{\partial y} , X_2^{\mbox{\tiny left}} + r\beta_2 \frac{\partial}{\partial x} -r\a_2 \frac{\partial}{\partial y}, X_3^{\mbox{\tiny left}}  \right\} \quad \mbox{and} \quad W = \textup{span} \left\{ \frac{\partial}{\partial x},  \frac{\partial}{\partial y} \right\}.
\end{equation}
The distribution $W$ satisfies the vertical-symmetry condition for the Lie algebra $\mathfrak{g}_\subW = \R^2$,  as was observed in \cite{paula}.
\bigskip

\noindent {\bf Reduction in two stages and the (almost) symplectic foliation.}
For ${\bf K}=(K_1, K_2, K_3)$, let us denote by $(g,x,y,{\bf K},  p_x,  p_y)$ the coordinates of $T^*Q$ associated the dual basis to $\{\vecL, \epsilon^x, \epsilon^y\}$ of \eqref{Ex:Ball:D+W}.
As was computed in \cite{PL2011} and \cite{Naranjo2008},
$$\M =  \{(g,{\bf x} ; {\bf K},  {\bf p} ) \in T^*Q \ : \ {\bf K} = \mathbb{I}\vecOm + mr^2 \langle \vecgamma , \vecOm \rangle \vecgamma  \quad \mbox{and} \quad  p_x =m r \langle \vecbeta ,  \vecOm \rangle , \ p_y = - m r \langle \vecalpha ,  \vecOm \rangle  \}.$$






The Lie subgroup $G_\subW = \R^2$ of $G$, acting by translation on $Q$, has Lie algebra $\mathfrak{g}_\subW = \R^2$ and it is a symmetry of the nonholonomic system.  As studied in \cite{Hoch}, the compressed motion takes place in $T^*(SO(3))$ and is described by the almost symplectic 2-form
\begin{equation} \label{Ex:Ball:barOmega}
 \bar{\Omega} = \Omega_{SO(3)} - \overline{\langle \mathcal{J}, \mathcal{K}_\subW \rangle},
\end{equation}
where
\begin{equation}\label{Ex:Ball:JK}
\overline{\langle \mathcal{J}, \mathcal{K}_\subW \rangle} = r^2 m \langle \vecOm , \vecL \times \vecL \rangle -  r^2 m \langle \vecgamma , \vecOm \rangle \vecgamma  \cdot d\vecgamma \times d\vecgamma.
\end{equation}

To compute \eqref{Ex:Ball:JK} we follow \cite{paula} where $d\vecgamma = \vecgamma \times \vecL$ and $d\vecL = (\lambda_2\wedge\lambda_3, \lambda_3\wedge\lambda_1, \lambda_1\wedge\lambda_2)$.



The $G$-action \eqref{Ex:Ball:G-action} induces the action of the Lie group $H = G / \R^2 = SO(2)$ on $SO(3)$ as in \eqref{Eq:H-action}, that is, $\varphi_H: H \times SO(3) \to SO(3)$ is given by
$\varphi_H(h, g) = hg$ where $h \in SO(2)$ is viewed as an element of $SO(3)$ such that $h{\bf e}_3 = {\bf e}_3$ for ${\bf e}_3=(0,0,1)$ and $g\in SO(3)$.
The $H$-action lifted to $T^*SO(3)$ is a symmetry of the compressed nonholonomic system \eqref{Eq:ChapReduction} (this is the Lie group used in \cite{Hoch} ad-hoc). The orbit projection $\rho_H:T^*(SO(3)) \to T^*(SO(3))/H$ is given in coordinates by $\rho_H(g, {\bf K}) = (\vecgamma, {\bf K})$, where $\vecgamma$ is the third row of the orthogonal matrix $g$, since ${\bf K}$ are left-invariant coordinates.

Since $W$ satisfies the vertical-symmetry condition, Prop.~\ref{Prop:f-conserved} asserts that, for $\eta=(1;0,0)$ the function $f_\eta = \langle J^\nh, P_{\mathfrak{g}_\subS}(\eta) \rangle  = \langle J^\nh, (1; -y,x) \rangle= \langle \vecgamma , {\bf K}\rangle$ drops to the well-defined function $g_1 = \langle \vecgamma , {\bf K}\rangle$ on $T^*(SO(3))$ (note that $\varrho (1;0;0) = 1 = \xi \in \mathfrak{h}$).
By Prop.~\ref{Prop:g_xi} the function $g_1$ coincides with the canonical momentum map on  $T^*(SO(3))$. Observe also that $\xi_{\mbox{\tiny{$T^*SO(3)$}}} = \vecgamma \cdot {\bf X}^{\mbox{\tiny{left}}}$.
Moreover, one can check that $\overline{\langle \mathcal{J}, \mathcal{K}_\subW \rangle} (\bar{X}_\nh, \vecgamma \cdot {\bf X}^{\mbox{\tiny{left}}} ) = 0$  and thus by Prop. \ref{Prop:JK=0}, $g_1 = \langle \vecgamma , {\bf K}\rangle$ (respectively $f_\eta$) is conserved by the compressed motion \eqref{Eq:NHdynamics} (resp. by the nonholonomic motion  \eqref{Eq:ChapReduction}) as it was observed in \cite{BorMa2001}. However, $\bar{J}: T^*SO(3) \to \R$, given at each $\xi \in \R$ by ${\bar J}_\xi = \xi g_1$, is not a momentum map for $\bar\Omega$ since $\overline{\langle \mathcal{J}, \mathcal{K}_\subW \rangle}$ given in \eqref{Ex:Ball:JK} is not basic with respect to $\rho_H:T^*\bar{Q} \to T^*\bar{Q}/H$.


Next, we look for a dynamical gauge $\bar{B}$ on $T^*\bar{Q}$ so that the 2-form $\bar{\Omega} - \bar{B}$ admits the canonical momentum map $\bar{J}:T^*\bar{Q} \to \mathfrak{h}^*$  as a momentum map (see Section \ref{S:Gauge}).

Following \cite{Hoch}, consider now the dynamical gauge defined by the 2-form
$\bar{B}=-r^2 m \langle \vecOm , \vecL \times \vecL \rangle$ (it can be checked that ${\bf i}_{\bar{X}_\nh} \bar{B} = 0$). By Corollary \ref{C:B+JK-Basic}, since $\overline{\langle \mathcal{J}, \mathcal{K}_\subW \rangle}+\bar{B}$ is basic with respect to the orbit projection $\rho_H:T^*(SO(3)) \to T^*(SO(3))/H$ then the canonical momentum map $\bar{J}: T^*(SO(3)) \to \mathfrak{h}^*$ becomes a momentum map of
$$
\bar{\Omega} - \bar{B} = -d(K_i \lambda_i)  + r^2 m \langle \vecgamma , \vecOm \rangle \vecgamma  \cdot d\vecgamma \times d\vecgamma.
$$

Let us denote by $\mathfrak{B}$ the 2-form on $T^*SO(3)/H$ given by
\begin{equation}\label{Ex:Ball:Bred}
\mathfrak{B} = - r^2 m \langle \vecgamma , \vecOm \rangle \vecgamma  \cdot d\vecgamma \times d\vecgamma.
\end{equation}
Therefore, for each $\bar\mu \in \mathfrak{h}^* \simeq \R$, we can perform a Marsden-Weinstein reduction obtaining the almost symplectic leaves $(\bar{J}^{-1}(\bar\mu)/H, \bar\omega_{\bar\mu}^{\bar\B})$ with $\bar\omega_{\bar\mu}^\B =  \Omega_{\bar\mu} - \mathfrak{B}_{\bar\mu}$ and where $\Omega_{\bar\mu}$ is the reduction of $\Omega_{\mbox{\tiny{$SO(3)$}}} = -d\langle {\bf K} , \vecL\rangle$ and $\mathfrak{B}_{\bar\mu}$ is the restriction of $\mathfrak{B}$ to $\bar{J}^{-1}(\bar\mu)/H$.
 From Lemma \ref{L:identifications} we conclude that the leaves $(\bar{J}^{-1}(\bar\mu)/H, \bar\omega_{\bar\mu}^{\bar\B})$ are symplectomorphic to
\begin{equation} \label{Ex:Ball:ReducedFolation}
(T^*(\mathbb{S}^2), \Omega_{\mbox{\tiny{$\mathbb{S}^2$}}} - \beta_\mu - \Upsilon^*\mathfrak{B}_\mu),
\end{equation}
where $\Omega_{\mbox{\tiny{$\mathbb{S}^2$}}} = -d \langle \vecgamma \times {\bf K} , d\vecgamma \rangle$ is the canonical 2-form on $T^*(\mathbb{S}^2)$ and $\beta_\mu$ defined in \eqref{Ap:DefBeta} is given by $\beta_\mu = \mu \, d \langle \vecgamma , \vecL\rangle = - \mu \langle \vecgamma, d \vecgamma \times d \vecgamma \rangle$.
\bigskip

\noindent {\bf The twisted Poisson bracket $\{\cdot, \cdot \}_\red^\B$.}
Following Section \ref{S:Gauge}, consider the $G$-invariant 2-form $B = \rho_\subW^* \bar{B}$ on $\M$.  The gauge transformation of the nonholonomic bracket $\{\cdot, \cdot\}_\nh$ by the 2-form $B$ gives a $G$-invariant bracket $\{\cdot, \cdot \}_\nh^\B$ (see \cite{PL2011} for an expression in coordinates).
The reduction of $\{\cdot, \cdot \}_\nh^\B$ by the Lie group $G$ induces a reduced bracket $\{ \cdot, \cdot \}_\red^\B$ on $\M/G$.  This bracket $\{ \cdot, \cdot \}_\red^\B$ describes the dynamics and it is twisted Poisson since $\langle \mathcal{J}, \mathcal{K}_\subW \rangle + B$ is basic with respect to the orbit projection $\rho:\M \to \M/G$ (Theorem \ref{T:Reduced-Dyn}').

Theorem \ref{T:Gauge:equivalence} asserts that the twisted Poisson bracket $\{ \cdot, \cdot \}_\red^\B$ has almost symplectic leaves  symplectomorphic to \eqref{Ex:Ball:ReducedFolation}.
\bigskip

\noindent {\bf The conformal factor.}
Since the classical reduced bracket $\{ \cdot, \cdot \}_\red$ has a nonintegrable characteristic distribution \cite{Naranjo2008,PL2011}, it does not make sense to look for a conformal factor (as it was discussed in \cite{PL2011}).  However, it is possible to look for a conformal factor of the twisted Poisson bracket  $\{\cdot, \cdot \}_\red^\B$.

Following Sec.~\ref{sec:conf}, we look for a conformal factor on each almost symplectic leaf associated to the twisted Poisson bracket  $\{\cdot, \cdot \}_\red^\B$. That is, we look for a function $f \in C^\infty(\bar{Q}/H)$ such that $f (\Omega_{\mbox{\tiny{$\mathbb{S}^2$}}} - \beta_\mu - \Upsilon^*\mathfrak{B}_\mu)$ is closed.
Using that $\beta_\mu= - \mu \langle \vecgamma, d\vecgamma \times d\vecgamma \rangle$ is a top form on $\mathbb{S}^2$  we obtain that
$$
d \left( f (\Omega_{\mbox{\tiny{$\mathbb{S}^2$}}} - \beta_\mu - \Upsilon^*\mathfrak{B}_\mu) \right) = df \wedge  \Omega_{\mbox{\tiny{$\mathbb{S}^2$}}}  - f \Upsilon^* (d \mathfrak{B}_\mu) .$$
Therefore, $df \wedge  \Omega_{\mbox{\tiny{$\mathbb{S}^2$}}}  = f \Upsilon^* (d \mathfrak{B}_\mu)$ if and only if
$$
\frac{\partial f}{\partial \gamma_i} = mr^2 f \frac{\partial}{\partial K_i}\langle \vecgamma, \vecOm\rangle,
$$
where we used that $\Omega_{\mbox{\tiny{$\mathbb{S}^2$}}} =  -d\langle {\bf K}, d\vecgamma\times \vecgamma \rangle$ and also that $\Upsilon^*\mathfrak{B}_\mu$  is a multiple of a top form on $\mathbb{S}^2$ (see \eqref{Ex:Ball:Bred}).
Since $\partial \langle \vecgamma, \vecOm\rangle / \partial K_i  = A_i^{-1} \gamma_ i (1-mr^2\langle \gamma,A^{-1}\gamma\rangle)^{-1}$ for $A=\mathbb{I}+mr^2\textup{Id}$ and  $\textup{Id}$ the $3\times 3$ identity matrix (see \cite{Naranjo2008}),  we obtain a set of decoupled differential equations for $f$:
$$
\frac{\partial f}{\partial \gamma_i} = mr^2 f  A_i^{-1} \gamma_ i (1-mr^2\langle \gamma,A^{-1}\gamma\rangle)^{-1}.
$$
Therefore $2\ln (f) =  \ln (1-mr^2\langle \gamma,A^{-1}\gamma\rangle)$ and so $f(\gamma) = (1-mr^2\langle \gamma,A^{-1}\gamma\rangle)^{1/2}$  is a conformal factor for the almost symplectic leaf $(T^*(\mathbb{S}^2), \Omega_{\mbox{\tiny{$\mathbb{S}^2$}}} - \beta_\mu - \Upsilon^*\mathfrak{B}_\mu)$.
Then
$$
g(\gamma)=(1-mr^2\langle \gamma,A^{-1}\gamma\rangle)^{-1/2},
$$
is a conformal factor of $\{\cdot, \cdot \}_\red^\B$, i.e., $f\{\cdot, \cdot \}_\red^\B$ is a Poisson bracket on $T^*\bar{Q}/H$, as was shown in \cite{PL2011,BorMa2001}.
Observe that the conformal factor $f$ was found without using Stanchenko condition \eqref{eq:stcond}.
\bigskip

\noindent {\bf Coordinate version of the Lagrange-Routh equations.}
The (compressed) Lagrangian $l: TSO(3) \to \R$ is given by the reduced kinetic energy $\bar{\kappa}$,
$$l(g, \vecOm) = \frac{1}{2} \bar{\kappa}((g,\vecOm), (g,\vecOm)) = \frac{1}{2} [ \langle A\vecOm,  \vecOm \rangle  -  mr^2 \langle \vecOm , \vecgamma\rangle^2 ],$$
where $A=\mathbb{I} +mr^2\textup{Id}$.
The vertical space $V_H$ associated to the $H$-action is generated by the vector field $Z_v= \vecgamma \cdot {\bf X}^{\mbox{\tiny{left}}} = \gamma_ i X^{\mbox{\tiny{left}}}_i.$ Assuming that $\gamma_3 \neq 0$, we can complete the basis to obtain a (local) basis of $TSO(3)$ such that
$$
\bar{\mathcal{B}} = \{ Z_1:= \langle \mathbb{I} \vecgamma , \vecgamma \rangle X_1^{\mbox{\tiny{left}}} - \gamma_1 \mathbb{I}_1 Z_v, \ Z_2 : = \langle \mathbb{I} \vecgamma , \vecgamma \rangle X_2^{\mbox{\tiny{left}}} - \gamma_2 \mathbb{I}_2 Z_v, \ Z_v\},$$
where $Z_1, Z_2$ are both vector fields $\bar{\kappa}$-orthogonal to $Z_v$ (if $\gamma_3=0$ then $\gamma_1$ or $\gamma_2$ are different from zero and we can therefore complete the basis in an analogous way). Letting
$$
\bar{\mathcal{B}}^* = \left\{\epsilon^1 = \frac{(\vecgamma \times \vecL)_2}{\gamma_3\langle \vecgamma, \mathbb{I}\vecgamma \rangle} , \  \epsilon^2 =  \frac{-(\vecgamma \times \vecL)_1}{\gamma_3\langle \vecgamma, \mathbb{I}\vecgamma \rangle} ,  \ \epsilon^v = \frac{\langle \mathbb{I}\vecgamma, \vecL \rangle}{\langle \mathbb{I} \vecgamma, \vecgamma \rangle } \right\}
$$
be the dual basis of $\bar{\mathcal{B}}$, we then denote by $(\Gamma_ 1, \Gamma_2,\Gamma_v = \frac{\langle \mathbb{I}\vecOm, \vecgamma \rangle}{\langle \mathbb{I} \vecgamma, \vecgamma \rangle } )$ and $(p_1, p_2,p_v = \langle {\bf K} , \vecgamma \rangle)$ the coordinates associated to the bases $\bar{\mathcal{B}}$ and  $\bar{\mathcal{B}}^*$ respectively.
In these coordinates the Lagrangian $l : TSO(3) \to \R$ becomes
$$
l(g; \Gamma_ 1, \Gamma_2,\Gamma_v) = \frac{\langle \mathbb{I} \vecgamma , \vecgamma \rangle}{2} \left[ ( \langle \mathbb{I} \vecgamma , \vecgamma \rangle (A_i - mr^2 \gamma_i^2) - (\mathbb{I}_i\gamma_i)^2 ) \Gamma_i^2 -  \gamma_1 \gamma_2 ( \langle \mathbb{I} \vecgamma , \vecgamma \rangle +   \mathbb{I}_1\mathbb{I}_2 ) \Gamma_1\Gamma_2 +  \Gamma_v^2 \right],
$$
for $i=1,2$.
The manifold $\bar{J}_l^{-1}(\bar\mu)$ is represented in local coordinates by $(g; \Gamma_1, \Gamma_2, \Gamma_v=\frac{\mu}{\langle \mathbb{I}\vecgamma, \vecgamma\rangle} )$, and the Routhian $R^{\bar\mu}$ is given by
$$
R^\mu(g; \Gamma_1, \Gamma_2) = \frac{\langle \mathbb{I} \vecgamma , \vecgamma \rangle}{2} \left[ ( \langle \mathbb{I} \vecgamma , \vecgamma \rangle (A_i - mr^2 \gamma_i^2) - (\mathbb{I}_i\gamma_i)^2 ) \Gamma_i^2 -  \gamma_1 \gamma_2 ( \langle \mathbb{I} \vecgamma , \vecgamma \rangle +   \mathbb{I}_1\mathbb{I}_2 ) \Gamma_1\Gamma_2 \right] - \frac{1}{2}\frac{\mu^2}{\langle \mathbb{I} \vecgamma, \vecgamma \rangle}.
$$

The manifold $\bar{J}_l^{-1}(\bar\mu)/H$ is diffeomorphic to $T\mathbb{S}^2$ and it is represented by the local coordinates $(\vecgamma, \Gamma_1, \Gamma_2)$ which are associated to the local basis $\{-\langle \mathbb{I} \vecgamma , \vecgamma \rangle (\vecgamma \times \partial/\partial \vecgamma)_1, -\langle \mathbb{I} \vecgamma , \vecgamma \rangle (\vecgamma \times \partial/\partial \vecgamma )_2\}$.   Then the reduced Lagrangian $\mathfrak{L}:T\mathbb{S}^2 \to \R$ and the function $\mathfrak{C}_\mu$ on $\mathbb{S}^2$  are given by
$$
\mathfrak{L}(\vecgamma; \Gamma_1, \Gamma_2) = \frac{\langle \mathbb{I} \vecgamma , \vecgamma \rangle}{2}\!\left[ ( \langle \mathbb{I} \vecgamma , \vecgamma \rangle (A_i - mr^2 \gamma_i^2) - (\mathbb{I}_i\gamma_i)^2 ) \Gamma_i^2 -  \gamma_1 \gamma_2 ( \langle \mathbb{I} \vecgamma , \vecgamma \rangle +   \mathbb{I}_1\mathbb{I}_2 ) \Gamma_1\Gamma_2 \right]  \mbox{ and }  \mathfrak{C}_\mu = \frac{\mu^2}{2\langle \mathbb{I} \vecgamma, \vecgamma \rangle}.
$$
The 2-form $\Upsilon_l^* \mathfrak{B}_{\bar\mu}$ on $T\mathbb{S}^2$ is given by
$$
\Upsilon_l^* \mathfrak{B}_{\bar\mu} = m r^2\langle \mathbb{I}\vecgamma, \vecgamma \rangle\left( (\langle \mathbb{I}\vecgamma, \vecgamma \rangle - \mathbb{I}_1)\gamma_1\Gamma_1 + (\langle \mathbb{I}\vecgamma, \vecgamma \rangle - \mathbb{I}_2 )\gamma_2  \Gamma_2  + \frac{\mu}{\langle \mathbb{I}\vecgamma, \vecgamma \rangle} \right)\bar{\epsilon}^1 \wedge \bar{\epsilon}^2,
$$
where
$\displaystyle{ \bar{\epsilon}^1 = \frac{d \gamma_2}{\gamma_3\langle \mathbb{I}\vecgamma, \vecgamma \rangle}, \bar{\epsilon}^2 = -\frac{d\gamma_ 1}{\gamma_3\langle \mathbb{I}\vecgamma, \vecgamma \rangle} }$ is the dual local basis on $T^*\mathbb{S}^2$.

The results thus far show that the Lagrange-Routh equations are the equations of motion on $T\mathbb{S}^2$ described by the almost symplectic 2-form
\begin{equation} \label{Ex:Ball:Reduced2form}
\Omega_{\mathfrak{L}} - \mu\cdot \beta + \Upsilon_l^*\mathfrak{B}_{\bar\mu},
\end{equation}
where $\tau_H^*\beta = d\epsilon^v$ for $\tau_H:SO(3) \to \mathbb{S}^2$ (see \eqref{Ap:DefBeta}),
with energy
$$
\mathcal{E}_\mu =  \frac{1}{2} \langle \mathbb{I} \vecgamma , \vecgamma \rangle^2 [ (\mathbb{I}_i + mr^2)\Gamma_i^2 - mr^2 \gamma_i\gamma_j\Gamma_i\Gamma_j ] - \frac{1}{2} \langle \mathbb{I} \vecgamma , \vecgamma \rangle  \mathbb{I}_i\mathbb{I}_j\gamma_i \gamma_j \Gamma_i\Gamma_j+  \frac{1}{2} \frac{\mu^2}{\langle \mathbb{I} \vecgamma, \vecgamma \rangle}.
$$
Observe that in our local coordinates,
$$
\Omega_{\mathfrak{L}} = d \left(\frac{\partial \mathfrak{L}}{\partial \Gamma_1} {\bar\epsilon}^1 + \frac{\partial \mathfrak{L}}{\partial \Gamma_2} \bar\epsilon^2 \right) \qquad \mbox{and}  \qquad \beta =  \left( Tr(\mathbb{I}) - 2 \frac{\langle\mathbb{I}\vecgamma , \mathbb{I}\vecgamma\rangle}{\langle \vecgamma, \mathbb{I}\vecgamma \rangle} \right) \bar{\epsilon}^1\wedge \bar{\epsilon}^2.
$$
where $Tr(\mathbb{I})$ is the trace of the matrix $\mathbb{I}$.

In other words, the leaves of $\{\cdot, \cdot \}_\red^\B$ are symplectomorphic to $T\mathbb{S}^2$ with the almost symplectic form given in \eqref{Ex:Ball:Reduced2form}. The {\it nonholonomic Lagrange-Routh} equations are the equations of motion on each leaf.

\subsection{A Mathematical Example} \label{Ss:Ex:MM}

In this section we describe a nonholonomic version of one of the unconstrained systems studied in \cite{CM1}. The configuration space is
$$
Q=\mathbb{R}^4 \times SE(2)=\{(u,v,x_1,x_2,y,z,\theta)\}.
$$ 
The Lagrangian $L:TQ \to \mathbb{R}$ is
\begin{equation}\label{MM-Lag}
L=\frac{1}{2}\left(\dot{u}^2+\dot{v}^2+\dot{x}_1^2+\dot{x}_2^2+\dot{y}^2+\dot{z}^2+\dot{\theta}^2\right)+2\left((\sin\theta)\dot{z}+(\cos\theta)\dot{y}\right)\dot{\theta},
\end{equation}
and the nonholonomic constraints can be represented by the annihilator of the one-forms
\begin{equation}\label{MM-1fs}
\epsilon^{u}=du-(1+\cos x_1)dx_2, \qquad \epsilon^{v}=dv-(\sin x_1)dx_2.
\end{equation}
The Lie group $G=\mathbb{R}^2 \times SE(2)$ is a symmetry of the nonholonomic system with the action on $Q$ given by
$$
G \times Q \to Q; \quad ((a,b,c,d,\beta),(u,v,x_1,x_2,y,z,\theta)) \mapsto (u+a,v+b,x_1,x_2,y\cos\beta-x\sin\beta+c,y\sin\beta+z\cos\beta+d,\theta+\beta).
$$
The distribution $D$ and a vertical complement $W$ \eqref{Eq:D+W} of $D$ in $TQ$ are given by
\begin{equation}\label{MM-DW}
D=\textup{span}\left\{\frac{\partial}{\partial \theta}+(1+\cos x_1)\frac{\partial}{\partial u}+(\sin x_1)\frac{\partial}{\partial v},\frac{\partial}{\partial x_1},\frac{\partial}{\partial x_2},\frac{\partial}{\partial y}, \frac{\partial}{\partial z}\right\} \quad \text{and} \quad W=\textup{span}\left\{\frac{\partial}{\partial u},\frac{\partial}{\partial v}\right\}.
\end{equation}

It can be seen that $W$ satisfies the vertical-symmetry condition \eqref{Eq:VerticalSymmetries} for the Lie algebra $\mathfrak{g}_\subW = \R^2$. The Lie group $G_\subW = \R^2$ ---defining the translational symmetry--- is a normal subgroup of $G$ and has Lie algebra $\mathfrak{g}_\subW=\R^2$. Thus the system is $G_\subW$-Chaplygin and the compressed nonholonomic motion takes place in $T^*\bar{Q}$, where $\bar{Q} = \mathbb{R}^2 \times SE(2)$ with coordinates $(x_1,x_2,y,z,\theta)$.  The 2-form on $T^*\bar{Q}$ describing the dynamics is $\bar{\Omega} = \Omega_{\bar{Q}} - \overline{\langle \mathcal{J}, \mathcal{K}_\subW \rangle}$, where
\begin{equation} \label{Ex:MM:JKRed}
\overline{\langle {\mathcal J}, \mathcal{K}_\subW \rangle}=  (\sin x_1)p_{x_2}dx_2 \wedge dx_1.
 \end{equation}
 The action of the Lie group $H = G/G_\subW = SE(2)$ is a symmetry of the compressed system, with the action on $\bar{Q}$ defined as in \eqref{Eq:H-action}:
\begin{equation*}
H \times \bar{Q} \rightarrow \bar{Q};
    \quad
\parentheses{ (c,d,\beta), (x_1,x_2,y,z,\theta) } \mapsto (x_1,x_2,y\cos\beta-z\sin\beta+c,y\sin\beta+z\cos\beta+d,\theta+\beta).
\end{equation*}
Since  $\overline{\langle {\mathcal J}, \mathcal{K}_\subW \rangle}$ is basic with respect to the orbit projection $\rho_H:T^*\bar{Q} \to
T^*\bar{Q}/H$, by Corollary \ref{C:JKbasic=conservationLaws} the canonical momentum map $\bar{J}:T^*\bar{Q}\to \R$ is a momentum map for $\Omega_{\mbox{\tiny{$\bar{Q}$}}} - \overline{\langle \mathcal{J}, \mathcal{K}_\subW \rangle}$.
So now, we proceed to perform a almost symplectic reduction as in Section \ref{Ss:2S-H_action} and \ref{Ss:TheRouthian}. In order to do that, we observe that the mechanical connection $\mathcal{A}:T(\R^2\times SE(2)) \to \mathfrak{h}$ associated to the compressed lagrangian 
 $l: T\bar{Q} \to \R$
\begin{equation}\label{Ex:MM:redl}
l(x_1,x_2,y,z,\theta,\dot{x}_1,\dot{x}_2,\dot{y},\dot{z},\dot\theta)=\frac{1}{2}\left(\dot{x}^2_1+2(2+\cos x_1)\dot{x}^2_2+\dot{y}^2+\dot{z}^2+\dot{\theta}^2\right)+2\left((\sin\theta)\dot{z}+(\cos\theta)\dot{y}\right)\dot{\theta},
\end{equation}
induces the Horizontal space $H=\textup{span} \{ X_1:= {\partial}/{\partial x^1}, X_2:={\partial}/{\partial} x^2\}$ on $T(\R^2\times SE(2))$. 

Now we set coordinates adapted to the projection $H\to H/ H_{\bar\mu}$. A typical element in the isotropy algebra $\mathfrak{h}_{\bar\mu}$ is $\xi=\xi^{1}(e_1+\mu e_2)$ \cite{CM1}, where $\{e_1,e_2,e_3\}$ are a basis for $\mathfrak{se}(2)$, and $\bar\mu \in \mathfrak{h}^*$ has been written as $\bar\mu=(1,\mu,0)$. The corresponding $H_{\bar\mu}$-action on $H$ is given by
$$
H_{\mu} \times SE(2) \to SE(2); \quad ((a),(y,z,\theta)) \mapsto (y+a,z,\theta).
$$
Thus, $H_{\bar\mu} \neq H$---in contrast to the previous two examples (though $H_{\bar\mu}$ is still Abelian).

Let us now follow \cite{CM1} and choose a basis of $\mathfrak{g}$ that contains the basis $e_1+\mu e_2$ of $\mathfrak{g}_{\bar\mu}$:
$$
E_1=e_1+\mu e_2, \quad E_2=e_2, \quad E_3=e_3.
$$
It follows that, with respect to this basis, $v \in T(SE(2))$ can be written as 
$$
v=v^1\widetilde{E}_1+v^2\widetilde{E}_2+v^3\widetilde{E}_3,
$$
where  
$$
\widetilde{E}_1 = \frac{\partial}{\partial y}+\mu\frac{\partial}{\partial z}, \qquad
\widetilde{E}_2=\frac{\partial}{\partial z}, \qquad \widetilde{E}_3=\frac{\partial}{\partial \theta}+y\frac{\partial}{\partial z}-z\frac{\partial}{\partial y},
$$
are the fundamental vector fields corresponding to $\{E_i\}$.  Now, \cite{CM1} introduces the coordinate transformation
$$
y'=y, \qquad z'=z-\mu y, \qquad \theta'=\theta
$$
so that $\widetilde{E}_1=\partial/\partial y'$, and therefore the fundamental vector field corresponding to the $H_{\bar\mu}$ action is $\widetilde{E}_1$. Then we have coordinates $(x_1, x_2, y', z',\theta')$ on $J^{-1}_l(\bar{\mu})$ and $(x_1, x_2, z',\theta')$ on $J^{-1}_l(\bar{\mu})/H_{\bar\mu}$. 
We consider the basis of $T(\R^2\times SE(2))$ given by 
$$
\left\{X_1=\frac{\partial}{\partial x_1},X_1=\frac{\partial}{\partial x_2},\widetilde{E}_1,\widetilde{E}_2,\widetilde{E}_3\right\},
$$ with associated coordinates $(\dot x_1, \dot x_2, v^1, v^2, v^3)$ where
\begin{equation}\label{v2d}
v^1=\dot{y}+z\dot{\theta}, \qquad v^2=\dot{z}-\mu\dot{y}-\mu z\dot{\theta}-y\dot{\theta}, \qquad v^3=\dot\theta.
\end{equation}

The reduced Lagrangian $\mathfrak{L}: T(\R^2) \to \R$ is given by (see \eqref{Def:Red-Lagragian}) 
$$
\mathfrak{L} (x_1, x_2 ; \dot x_1, \dot x_2) = \frac{1}{2}\left(\dot{x}^2_1+2(2+\cos x_1)\dot{x}^2_2\right).
$$

By Theorem \ref{T:RouthIntrinsic}, the almost symplectic leaves of the reduced nonholonomic bracket $\{ \cdot \, , \cdot \}_\red$ are given by $(T\R^2 \times_{\R^2} \tilde{\mathcal{O}}, \ \omega_{\mathfrak{L}} - \beta_{\bar\mu} - [\overline{\langle {\mathcal J}, \mathcal{K}_\subW \rangle}_\red]_{\bar\mu})$ where 
$$ 
 \overline{\langle {\mathcal J}, \mathcal{K}_\subW \rangle}_\red=  \frac{\partial \mathfrak{L}}{\partial \dot x_2} (\sin x_1)dx_2 \wedge dx_1 = 2(2+\cos x_1)(\sin x_1) \dot x_2 \, dx_2 \wedge dx_1.
$$

In order to write the equations of motion on the leaves, recall that the curvature $\mathcal{F}$ associated to the mechanical connection is zero and also that $\partial\mathfrak{C}_{\bar\mu}/ \partial x_i = 0$, since the $(y,z,\theta)$ variables are not coupled to $x_1$ or $x_2$ in our lagrangian $l$.
Then \eqref{Eq:Splitted-RedLag2} are given, in this example by 
\begin{equation}
 \begin{split}
  -\dot x^2 (\sin x_1) - \ddot x_1 & = 2(2+\cos x_1)(\sin x_1)\dot x_2^2 ,  \\ 
  -2(2+\cos x_1)\ddot x_2 + 2 (\sin x_1) \dot x_1 \dot x_2 & = -2(2+\cos x_1) (\sin x_1) \, \dot x_1 \dot x_2.
 \end{split}
\end{equation}

To obtain equations \eqref{Eq:Splitted-RedLag1}, we follow \cite{CM1} writing the compressed lagrangian $l:T\bar{Q} \to \R$ in the variables $(x_1,x_2,y,z,\theta,\dot{x}_1,\dot{x}_2,v^1,v^2,v^3)$ and finding the functions $c_a$ by solving $(p_1, p_2, p_3) =(1+\mu^2, \mu,0)= \bar\mu\in \mathfrak{h}^* $ for $(v^1, v^2,v^3)$. Hence we obtain the same equations as in \cite{CM1} describing the {\it second order equation} condition:  
\begin{equation}
 \begin{split}\label{z} 
  \dot{z}' & =  \frac{2}{3}\left(z'(\sin\theta'-\mu\cos\theta')-2(1-\mu^2)\sin\theta'\cos\theta'-2\mu+4\mu\cos^2\theta'\right), \\
\dot{\theta}' & = \frac{1}{3}\left(2\cos\theta'+2\mu\sin\theta'-z'\right). 
 \end{split}
\end{equation}

\end{document}